\documentclass[a4paper,12pt,centertags,reqno,psamsfonts]{amsart}

\usepackage[headings]{fullpage}
\usepackage{setspace}
\usepackage{kerkis}
\usepackage{letltxmacro}
\usepackage[inline]{enumitem}
\usepackage{pifont}
\usepackage{array}
\usepackage{fancybox}
\usepackage{hyperref}
\usepackage{ifthen}
\usepackage{graphicx}
\usepackage{subfigure}
\usepackage{amsmath}
\usepackage{amsthm}
\usepackage{amssymb}				
\usepackage{mathtools}
\usepackage{mathrsfs}
\usepackage{stmaryrd}
\usepackage{yhmath}
\usepackage{accents}
\usepackage{xfrac}
\usepackage[all]{xy}
\usepackage[round,comma,authoryear,sort&compress]{natbib}

\bibpunct{(}{)}{,}{a}{}{,}

\setlist[itemize]{leftmargin=*,topsep=1ex,itemsep=0ex}
\setlist[enumerate]{leftmargin=*,topsep=1ex,itemsep=0ex}

\allowdisplaybreaks[4]

\makeatletter
\let\oldr@@t\r@@t
\def\r@@t#1#2{%
\setbox0=\hbox{$\oldr@@t#1{#2\,}$}\dimen0=\ht0
\advance\dimen0-0.2\ht0
\setbox2=\hbox{\vrule height\ht0 depth -\dimen0}%
{\box0\lower0.4pt\box2}}
\LetLtxMacro{\oldsqrt}{\sqrt}
\renewcommand*{\sqrt}[2][\ ]{\oldsqrt[#1]{#2}}
\makeatother

\theoremstyle{plain}
\newtheorem{theorem}{Theorem}
\newtheorem{lemma}[theorem]{Lemma}
\newtheorem{proposition}[theorem]{Proposition}

\theoremstyle{definition}

\theoremstyle{remark}

\numberwithin{theorem}{section}
\numberwithin{equation}{section}
\numberwithin{figure}{section}
\numberwithin{table}{section}

\newcommand{\N}{\mathbb{N}}

\newcommand{\R}{\mathbb{R}}

\newcommand{\func}[3]{#1:#2\rightarrow #3}

\newcommand{\C}[2]{
\ifthenelse{#1=0 \and #2=0}{\textsf{\upshape C}}
{\ifthenelse{#2=0}{\textsf{\upshape C}^{#1}}
{\textsf{\upshape C}^{#1,#2}}}
}

\newcommand{\e}{\mathrm{e}}

\renewcommand{\d}{\mathrm{d}}

\newcommand{\SigAlg}[1]{\mathscr{#1}}

\newcommand{\Filt}[1]{\mathfrak{#1}}
\newcommand{\StopTimes}{\mathfrak{S}}

\newcommand{\E}{\textsf{\upshape E}}

\renewcommand{\P}{\textsf{\upshape P}}

\newcommand{\ind}[1]{\mathbf{1}_{#1}}

\renewcommand{\>}{\rangle}

\renewcommand{\L}{\mathcal{L}}

\newcommand{\scale}{\mathfrak{s}}
\newcommand{\speed}{\mathfrak{m}}

\newcommand{\Resolvent}[1]{\mathscr{R}_{#1}}

\begin{document}
\title{Short Selling with Margin Risk and Recall Risk}

\author{Kristoffer Glover}
\author{Hardy Hulley}

\address{Kristoffer Glover\\
Finance Discipline Group\\
University of Technology Sydney\\
P.O. Box 123\\
Broadway, NSW 2007\\
Australia}
\email{kristoffer.glover@uts.edu.au}

\address{Hardy Hulley\\
Finance Discipline Group\\
University of Technology Sydney\\
P.O. Box 123\\
Broadway, NSW 2007\\
Australia}
\email{hardy.hulley@uts.edu.au}

\subjclass[2010]{Primary: 91G80; Secondary: 60G40, 60J60, 35R35}

\keywords{Short selling; Margin risk; Recall risk; Optimal stopping; Free-boundary problem; Local time-space formula}

\date{\today}

\begin{abstract}
Short sales are regarded as negative purchases in textbook asset pricing theory. In reality, however, the symmetry between purchases and short sales is broken by a variety of costs and risks peculiar to the latter. We formulate an optimal stopping model in which the decision to cover a short position is affected by two short sale-specific frictions---margin risk and recall risk. Margin risk refers to the fact that short sales are collateralised transactions, which means that short sellers may be forced to close out their positions involuntarily if they cannot fund margin calls. Recall risk refers to a peculiarity of the stock lending market, which permits lenders to recall borrowed stock at any time, once again triggering involuntary close-outs. We examine the effect of these frictions on the optimal close-out strategy and quantify the loss of value resulting from each. Our results show that realistic short selling constraints have a dramatic impact on the optimal behaviour of a short seller, and are responsible for a substantial loss of value relative to the first-best situation without them. This has implications for many familiar no-arbitrage identities, which are predicated on the assumption of unfettered short selling.
\end{abstract}

\maketitle

\section{Introduction}
\label{Sec1}
Short sales facilitate negative exposures to financial securities. The salient features of a short sale in the equity market may be summarised as follows.\footnote{See \citet{CHR04}, \citet{D'A02}, and \citet{Ree13} for comprehensive overviews of the short selling process and the equity lending market.} First, a prospective short seller identifies a willing lender of the desired stock. He then borrows the stock, sells it in the market, and posts collateral equal to its market value plus a haircut with the lender.\footnote{The haircut is usually set at 2\% of the market value of the stock.} The collateral is marked to market daily, so that an increase in the stock price prompts a margin call for more collateral, while a decrease entitles the short seller to withdraw some collateral. The lender invests the collateral at the prevailing interest rate, and pays part of the resulting income to the short seller, in the form of a negotiated rebate.\footnote{The difference between the interest rate and the rebate rate represents the lending fee.} In return, the short seller compensates the lender for the dividends forgone during the life of the loan. To complete the transaction, the short seller repurchases the stock, and returns it to the lender. Although this is likely to occur at the short seller's discretion, stock loan contracts generally include a recall provision that permits lenders to force short sellers to liquidate their positions involuntarily.

The previous overview highlights several costs and risks associated with short sales, which have no counterparts when stocks are purchased. First, short sellers may incur search costs, since the process of identifying willing lenders can be expensive. Second, the lending fees associated with stock loans impose a cost on short sellers that can be quite significant. Third, short sellers are exposed to margin risk, due to the possibility that successive increases in the price of borrowed stock may generate margin calls that eventually exhaust their collateral budgets. Finally, short sellers face recall risk, since lenders may recall borrowed stock at any time. Collectively, these frictions are referred to as short selling constraints.

Short selling constraints can have dramatic implications for textbook no-arbitrage relationships, since arbitrage portfolios invariably contain long and short positions. For example, \citet{LT03} highlight the role of short selling constraints in preventing the correction of mispriced equity carve-outs, while \citet{MPS02} identify their effect when a company trades at a discount relative to its subsidiaries. Short selling constraints have also been implicated in mispriced equity index futures by \citet{FD99}, and in put-call parity violations by \citet{ORW04}.

We focus on margin risk and recall risk. \citet{SV97} first revealed the importance of margin risk as a limit to arbitrage, by demonstrating that a hedge fund may be forced to close out a theoretically profitable arbitrage trade in a loss-making position, if interim losses trigger sufficient investor withdrawals to compromise the fund's ability to meet its margin calls. \citet{LL04} further reinforced the intuition that margin risk detracts from the attractiveness of arbitrage, by showing that a collateral-constrained risk-averse trader will underinvest in arbitrage opportunities, due to the possibility that the margin calls arising from interim losses may exhaust his collateral budget before he realises a profit.

As for the impact of recall risk, the empirical study by \citet{CR17} found that recalls induce profit sharing between short sellers and stock lenders, by forcing the former to liquidate otherwise profitable short positions prematurely. They estimated that informed short sellers lose around 20\% of their first-best profits due to recalls, indicating that recall risk is an economically significant short selling constraint. \citet{ERR18} provided further implicit evidence on the impact of recall risk, by demonstrating that deviations between spot and put-call parity-implied stock prices increase with option maturity. This suggests that short selling constraints intensify as the time-horizon of trading strategies involving short positions increases, consistent with the nature of recall risk.

To understand the impact of margin risk and recall risk on short sales, we formulate an optimal stopping problem where a trader must decide when to close out a short position initiated at time zero. To capture the effect of margin risk, we assume that the trader has a limited collateral budget to fund margin calls, while recall risk is modelled by assuming that forced close out due to stock recall occurs at some independent random time. The collateral constraint introduces a knock-out barrier above the initial stock price, at a level determined by the collateral budget. Immediate close out occurs when the stock price reaches this barrier, since the available collateral will have been depleted by then. The existence of a knock-out barrier whose location is determined by the initial stock price is a novel feature of the optimal stopping problem studied here. The possibility of early recall further complicates the analysis, by inserting an inhomogeneous term in the associated free-boundary problem. Together, these features produce an optimal stopping problem that is interesting in its own right.

Our economic analysis confirms that margin risk and recall risk have a dramatic impact on the value of a short sale and the optimal close-out strategy, effectively driving  a wedge between the solutions to the constrained and unconstrained short selling problems. The size of the effect is very sensitive to the drift of the stock price, which creates a practical challenge due to the difficulty of estimating that parameter accurately. These frictions are responsible for other surprising effects, as well. For example, unlike the case with unconstrained short sales, the value of the constrained short position is non-monotonic with respect to the volatility of the stock price. It is also an increasing function of the discount rate when the stock price has positive drift, whereas the value of the unconstrained short position is always monotonically decreasing with respect to the discount rate. The underlying intuitions behind these counterintuitive effects are quite subtle, and are discussed at length in our economic analysis.

This paper falls within an established literature on optimal stopping models for the related optimal margin lending problem \citep[see e.g.][]{XZ07,EW08,ZZ09,LWJ10,LX10,DX11,GG13,SYZ16,WW13,CS14,XY19,YXL19}. In the case of a margin loan (also known as a stock loan), an investor borrows money from a bank to purchase  stock, which the bank then holds as collateral. A fall in the stock price decreases the value of the collateral, triggering a margin call. If the investor fails to meet the margin call, the bank sells the stock to cover the loan. The investor's problem is to choose the optimal time to sell the stock and settle the loan. The similarity between short sales and margin loans rests on the fact that both are collateralised transactions, in which margin risk plays a prominent role. However, in contrast to short sales, margin loans are not exposed to recall risk and other short selling constraints.

Compared to the margin lending problem, optimal short selling has received little attention, with the recent studies by \citet{CT15} and \citet{Chu16} appearing to be the only exceptions. Those articles consider the situation of a trader who must choose when to close out a short position, subject to lending fees, recall risk and liquidity risk. Our model differs from the models presented in those studies, in terms of which constraints are considered. Specifically, while they ignore margin risk (which is an important short sale-specific constraint), we are not concerned with liquidity risk (which is not a short sale-specific friction). As a result, our analysis offers different economic insights.

The remainder of the article is structured as follows. Section~\ref{Sec2} models the price of a non-dividend-paying stock as a geometric Brownian motion and recalls several facts about such processes. The problem of when to close out a short position in the stock, in the presence of margin risk and recall risk, is then formulated as an optimal stopping problem. Section~\ref{Sec3} derives and solves a free-boundary problem for a simpler optimal stopping problem, and uses the equation for the free boundary to identify seven qualitatively different parameter regimes. The analysis in this section revolves around identifying the roots of a complicated non-linear function, whose behaviour is highly dependent on the parameter regime (see Figure~\ref{figSec3:H}). Section~\ref{Sec4} presents candidate value functions and optimal stopping strategies for the simpler problem, under each parameter regime, and verifies that they do in fact solve it. In Section~\ref{Sec5}, those solutions are used to construct the solution to the original optimal short selling problem. Finally, Section~\ref{Sec6} offers some economic insights, by examining the dependence of the optimal close-out policy and the value function on the model parameters.
\section{Short Selling as an Optimal Stopping Problem}
\label{Sec2}
\subsection{Modelling the price of a non-dividend-paying stock}
Let $B=(B_t)_{t\geq 0}$ be a standard Brownian motion on a complete filtered probability space $(\Omega,\SigAlg{F},\Filt{F},\P)$, whose filtration $\Filt{F}=(\SigAlg{F}_t)_{t\geq 0}$ satisfies the \emph{usual conditions} of right continuity and completion by the null-sets of $\SigAlg{F}$. We consider a non-dividend-paying stock, whose price $X=(X_t)_{t\geq 0}$ is modelled as the unique strong solution to the stochastic differential equation
\begin{equation}
\label{eqSec2:SDE}
\d X_t=\mu X_t\,\d t+\sigma X_t\,\d B_t,
\end{equation}
for all $t\geq 0$, where $X_0\in(0,\infty)$, $\mu\in\R$ and $\sigma\in(0,\infty)$. That is to say, we model the stock price as a geometric Brownian motion. This process is characterised by its scale function and speed measure
\begin{equation}
\label{eqSec2:ScaleSpeed}
\scale(x)\coloneqq
\begin{cases}
-\frac{x^{-2\nu}}{2\nu}&\text{if $\nu\neq 0$};\\
\ln x&\text{if $\nu=0$}	
\end{cases}
\qquad\text{and}\qquad
\speed(\d x)\coloneqq\frac{2}{\sigma^2}x^{2\nu-1}\,\d x,
\end{equation}
for all $x\in(0,\infty)$, where $\nu\coloneqq\sfrac{\mu}{\sigma^2}-\sfrac{1}{2}$ \citep[see][Appendix~1.20]{BS02}.

Given $\alpha>0$, let $\phi_\alpha,\psi_\alpha\in\C{2}{0}(0,\infty)$ denote the unique (up to multiplication by a positive scalar) decreasing and increasing solutions, respectively, to the second-order ordinary differential equation
\begin{equation}
\label{eqSec2:ODE}
\L_Xu(x)\coloneqq\frac{1}{2}\sigma^2x^2u''(x)+\mu xu'(x)=\alpha u(x),
\end{equation}
for all $x\in(0,\infty)$. These solutions are given explicitly by
\begin{equation}
\label{eqSec2:PhiPsi}
\phi_\alpha(x)\coloneqq	x^{-\sqrt{\nu^2+\sfrac{2\alpha}{\sigma^2}}-\nu}
\qquad\text{and}\qquad
\psi_\alpha(x)\coloneqq	x^{\sqrt{\nu^2+\sfrac{2\alpha}{\sigma^2}}-\nu},
\end{equation}
for all $x\in(0,\infty)$. Differentiating these expressions gives
\begin{equation}
\label{eqSec2:PhiPrimePsiPrime}
\phi_\alpha'(x)=-\biggl(\sqrt{\nu^2+\frac{2\alpha}{\sigma^2}}+\nu\biggr)\frac{\phi_\alpha(x)}{x}
\qquad\text{and}\qquad
\psi_\alpha'(x)=\biggl(\sqrt{\nu^2+\frac{2\alpha}{\sigma^2}}-\nu\biggr)\frac{\psi_\alpha(x)}{x},
\end{equation}
for all $x\in(0,\infty)$. It follows from \eqref{eqSec2:ScaleSpeed}, \eqref{eqSec2:PhiPsi} and \eqref{eqSec2:PhiPrimePsiPrime} that the Wronskian
\begin{equation}
\label{eqSec2:Wronskian}
w_\alpha\coloneqq\frac{\phi_\alpha(x)\psi_\alpha'(x)-\phi_\alpha'(x)\psi_\alpha(x)}{\scale'(x)}	
=2\sqrt{\nu^2+\frac{2\alpha}{\sigma^2}}
\end{equation}
is independent of $x\in(0,\infty)$ \citep[see][Appendix~1.20]{BS02}.

Let $\StopTimes$ denote the family of all $\Filt{F}$-stopping times. These include the first-exit times
\begin{equation*}
\label{eqSec2:FirstExitTimes}
\hat{\tau}_z\coloneqq\inf\{t\geq 0\,|\,X_t\geq z\}
\qquad\text{and}\qquad
\check{\tau}_z\coloneqq\inf\{t\geq 0\,|\,X_t\leq z\},
\end{equation*}
for all $z\in(0,\infty)$, whose Laplace transforms are given by
\begin{subequations}
\label{eqSec2:LaplaceTransforms}
\begin{align}
\E_x\bigl(\e^{-\alpha\hat{\tau}_z}\bigr)&=
\begin{cases}
\frac{\psi_\alpha(x)}{\psi_\alpha(z)}&\qquad\text{if $x\leq z$}\\
1&\qquad\text{if $x\geq z$},	
\end{cases}
\label{eqSec2:LaplaceTransforms_1}\\
\intertext{and}
\E_x\bigl(\e^{-\alpha\check{\tau}_z}\bigr)&=
\begin{cases}
1&\qquad\text{if $x\leq z$}\\
\frac{\phi_\alpha(x)}{\phi_\alpha(z)}&\qquad\text{if $x\geq z$},	
\end{cases}
\label{eqSec2:LaplaceTransforms_2}
\end{align}	
\end{subequations}
for all $\alpha>0$ and all $x,z\in(0,\infty)$ \citep[see][Section~II.10]{BS02}. As usual, $\E_x(\,\cdot\,)$ denotes the expected value operator with respect to the probability measure $\P_x$, under which $X_0=x$.

\subsection{The optimal time to close out a short sale}
We consider a trader with collateral budget $c\geq 0$, who sells the stock short at time zero. Forced close-out of his position due to collateral exhaustion occurs when the price of the stock first exceeds its initial price by more than his budgeted collateral. We denote this stopping time by
\begin{equation}
\label{eqSec2:CollatTime}
\zeta\coloneqq\inf\{t\geq 0\,|\,X_t=X_0+c\}.
\end{equation}
In other words, at time $\zeta$ the stock price will be high enough to ensure that the trader will have spent his entire collateral budget on margin calls. 

Involuntary close-out of the short sale due to stock recall occurs at an exponentially distributed random time $\rho\sim\text{Exp}(\lambda)$, where $\lambda>0$ is the recall intensity. The recall time is assumed to be $\SigAlg{F}$-measurable and independent of $\SigAlg{F}_\infty\coloneqq\bigvee_{t\geq 0}\SigAlg{F}_t$, which is to say that the recall event is independent of the stock price. Hence, the distribution of the recall time is given by
\begin{equation*}
\P_x(\rho\in\d t)=\lambda\e^{-\lambda t}\,\d t,
\end{equation*}
for all $t\geq 0$ and all $x\in(0,\infty)$. We model recall as an independent random event to capture the intuition that the lender may recall the stock at any time.

To simplify the analysis (and because it corresponds to a case of particular interest), we restrict our attention to the situation when the rebate rate on the stock loan is zero. That is to say, we assume that the stock lending fee is identical to the prevailing interest rate $r\geq 0$.\footnote{In the equity lending market, stocks with high lending fees are referred to as special stocks. The lending fees of some so-called extremely special stocks are so high that their rebate rates are negligible or even negative \citep[see][]{D'A02}.} This implies that the trader does not earn a return on the balance of his margin account, which initially contains the proceeds from borrowing the stock and selling it in the market. Hence, if the trader initiates the short sale at time zero and closes it out at time $t\geq 0$, the present value of his profit is $\e^{-rt}(X_0-X_t)$. His objective is to repurchase the stock at a time that maximises the expected value of his profit, in present value terms, subject to the frictions outlined above. This gives rise to the following optimal stopping problem:
\begin{equation}
\label{eqSec2:ShortSellProb}
V(x)\coloneqq\sup_{\tau\in\StopTimes}\E_x\Bigl(\e^{-r(\tau\wedge\zeta\wedge\rho)}\bigl(x-X_{\tau\wedge\zeta\wedge\rho}\bigr)\Bigr),
\end{equation}
for all $x\in(0,\infty)$. Figure~\ref{figSec2:Paths} illustrates Problem~\eqref{eqSec2:ShortSellProb}, by giving an example of a stock price path where the trader closes out the short position optimally, as well as examples of paths where the short sale is closed out involuntarily due to collateral exhaustion and recall.

\begin{figure}
\centering
\includegraphics[scale=0.6]{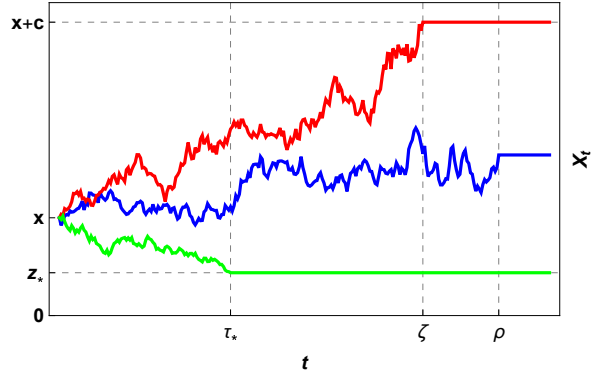}
\caption{Three possible outcomes for the short selling problem \eqref{eqSec2:ShortSellProb}. The lower (green) path illustrates the case when the short sale is closed out optimally at time $\tau_*$, which is the first time the stock price reaches a putative optimal close-out threshold $z_*$. (At this stage we do not know that the optimal close-out strategy is a threshold strategy, but we shall demonstrate that this is indeed the case.) The upper (red) path illustrates the case when the short sale is closed out due to collateral exhaustion at time $\zeta$, which is the first time the stock price exceeds its initial value $x$ by the collateral budget $c$. The middle (blue) path illustrates the case when the short sale is closed out due to stock recall at a random time $\rho$.}
\label{figSec2:Paths}
\end{figure}
\section{A Simpler Optimal Stopping Problem}
\label{Sec3}
\subsection{Fixing the collateral exhaustion level}
Optimal stopping problems for diffusions are often solved by reformulating them as free-boundary problems. This approach is not directly applicable to Problem~\eqref{eqSec2:ShortSellProb}, however, since the collateral exhaustion time \eqref{eqSec2:CollatTime} depends on the initial stock price. To overcome this difficulty, we formulate a simpler problem, where forced liquidation due to collateral exhaustion occurs as soon as the stock price exceeds some fixed level $\kappa>0$ by the collateral budget, rather than when it first exceeds its initial price by the collateral budget. Formally, we modify Problem~\eqref{eqSec2:ShortSellProb} as follows:
\begin{subequations}
\label{eqSec3:OptStopProb}
\begin{equation}
\label{eqSec3:OptStopProb_a}
\widetilde{V}(x)\coloneqq\sup_{\tau\in\StopTimes}J(x,\tau),
\end{equation}
for all $x\in(0,\infty)$, where
\begin{equation}
\label{eqSec3:OptStopProb_b}
\begin{split}
&J(x,\tau)\coloneqq\E_x\Bigl(\e^{-r(\tau\wedge\hat{\tau}_{\kappa+c}\wedge\rho)}\bigl(\kappa-X_{\tau\wedge\hat{\tau}_{\kappa+c}\wedge\rho}\bigr)\Bigr)\\
&=\E_x\Bigl(\E_x\Bigl(\ind{\{\rho\leq\tau\wedge\hat{\tau}_{\kappa+c}\}}\e^{-r\rho}(\kappa-X_\rho)\,\bigl|\,\SigAlg{F}_{\tau\wedge\hat{\tau}_{\kappa+c}}\Bigr)\\
&\hspace{4cm}+\e^{-r(\tau\wedge\hat{\tau}_{\kappa+c})}\bigl(\kappa-X_{\tau\wedge\hat{\tau}_{\kappa+c}}\bigr)\P_x\bigl(\rho>\tau\wedge\hat{\tau}_{\kappa+c}\,|\,\SigAlg{F}_{\tau\wedge\hat{\tau}_{\kappa+c}}\bigr)\Bigr)\\
&=\E_x\biggl(\int_0^{\tau\wedge\hat{\tau}_{\kappa+c}}\e^{-rt}(\kappa-X_t)\lambda\e^{-\lambda t}\,\d t+\e^{-r(\tau\wedge\hat{\tau}_{\kappa+c})}\bigl(\kappa-X_{\tau\wedge\hat{\tau}_{\kappa+c}}\bigr)\int_{\tau\wedge\hat{\tau}_{\kappa+c}}^\infty\lambda\e^{-\lambda t}\,\d t\biggr)\\
&=\E_x\biggl(\int_0^{\tau\wedge\hat{\tau}_{\kappa+c}}\lambda\e^{-(\lambda+r)t}(\kappa-X_t)\,\d t+\e^{-(\lambda+r)(\tau\wedge\hat{\tau}_{\kappa+c})}\bigl(\kappa-X_{\tau\wedge\hat{\tau}_{\kappa+c}}\bigr)\biggr),
\end{split}
\end{equation}
\end{subequations}
for all $\tau\in\StopTimes$. Note that $V(x)=\widetilde{V}(x)|_{\kappa=x}$, for all $x\in(0,\infty)$, since $\zeta=\hat{\tau}_{x+c}$ $\P_x$-a.s. Hence, a solution to Problem~\eqref{eqSec2:ShortSellProb} can be constructed from a solution to Problem~\eqref{eqSec3:OptStopProb}.

\subsection{An associated free-boundary problem}
The time-homogeneity of Problem~\eqref{eqSec3:OptStopProb} leads us to conjecture that the optimal stopping time for that problem is a threshold time. That is to say, it is the first time $\check{\tau}_{z_*}$ the stock price crosses some threshold $z_*\in(0,\kappa+c)$ from above. If that is true, we may be able to solve Problem~\eqref{eqSec3:OptStopProb} by solving the following free-boundary problem for $z_*\in(0,\kappa+c)$ and $\widehat{V}\in\C{0}{0}(0,\infty)\cap\C{2}{0}(z_*,\kappa+c)$ \citep[see][Chapter~III]{PS06}:
\begin{subequations}
\label{eqSec3:FreeBndProb}
\begin{align}
\label{eqSec3:FreeBndProb_a}
\L_X\widehat{V}(x)-(\lambda+r)\widehat{V}(x)+\lambda(\kappa-x)&=0,\\
\intertext{for all $x\in(z_*,\kappa+c)$;}
\label{eqSec3:FreeBndProb_b}
\widehat{V}(x)&=\kappa-x,\\
\intertext{for all $x\in(0,z_*]\cup[\kappa+c,\infty)$; and}
\label{eqSec3:FreeBndProb_c}
\widehat{V}'(z_*+)&=-1.
\end{align}
\end{subequations}
In other words, $z_*\in(0,\kappa+c)$ and $\widehat{V}\in\C{0}{0}(0,\infty)\cap\C{2}{0}(z_*,\kappa+c)$ must satisfy the ordinary differential equation \eqref{eqSec3:FreeBndProb_a} in the continuation region $(z_*,\kappa+c)$, the instantaneous stopping condition \eqref{eqSec3:FreeBndProb_b} in the stopping region $(0,z_*]\cup[\kappa+c,\infty)$, and the smooth pasting condition \eqref{eqSec3:FreeBndProb_c} at the free boundary $z_*$.

To analyse Problem~\eqref{eqSec3:FreeBndProb}, we first assume that it admits a solution, consisting of a boundary $z_*\in(0,\kappa+c)$ and a function $\widehat{V}\in\C{0}{0}(0,\infty)\cap\C{2}{0}(z_*,\kappa+c)$. The latter may be expressed in terms of the general solution to the homogeneous equation \eqref{eqSec2:ODE},  with $\alpha=\lambda+r$, and a particular solution $\widehat{v}\in\C{2}{0}(0,\infty)$ to the inhomogeneous equation \eqref{eqSec3:FreeBndProb_a}, as follows:
\begin{equation}
\label{eqSec3:ODEGenSol}
\widehat{V}(x)=A\phi_{\lambda+r}(x)+B\psi_{\lambda+r}(x)+\widehat{v}(x),
\end{equation}
for all $x\in(z_*,\kappa+c)$, where $A,B\in\R$ are constants. Letting $x\downarrow z_*$ and $x\uparrow\kappa+c$ in \eqref{eqSec3:ODEGenSol}, and substituting the resulting expressions for $\widehat{V}(z_*)$ and $\widehat{V}(\kappa+c)$ into \eqref{eqSec3:FreeBndProb_b}, produces the equations
\begin{align*}
A\phi_{\lambda+r}(z_*)+B\psi_{\lambda+r}(z_*)+\widehat{v}(z_*)&=\kappa-z_*\\
\intertext{and}
A\phi_{\lambda+r}(\kappa+c)+B\psi_{\lambda+r}(\kappa+c)+\widehat{v}(\kappa+c)&=-c,
\end{align*}
which can be solved to give
\begin{align*}
A&=\frac{\bigl(\kappa-z_*-\widehat{v}(z_*)\bigr)\psi_{\lambda+r}(\kappa+c)+\bigl(c+\widehat{v}(\kappa+c)\bigr)\psi_{\lambda+r}(z_*)}{\phi_{\lambda+r}(z_*)\psi_{\lambda+r}(\kappa+c)-\phi_{\lambda+r}(\kappa+c)\psi_{\lambda+r}(z_*)}\\
\intertext{and}
B&=\frac{\bigl(\kappa-z_*-\widehat{v}(z_*)\bigr)\phi_{\lambda+r}(\kappa+c)+\bigl(c+\widehat{v}(\kappa+c)\bigr)\phi_{\lambda+r}(z_*)}{\phi_{\lambda+r}(\kappa+c)\psi_{\lambda+r}(z_*)-\phi_{\lambda+r}(z_*)\psi_{\lambda+r}(\kappa+c)}.
\end{align*}
When these expressions are inserted into \eqref{eqSec3:ODEGenSol}, we obtain
\begin{equation}
\label{eqSec3:FreeBndProbSol}
\begin{split}
\widehat{V}(x)&=\bigl(\kappa-z_*-\widehat{v}(z_*)\bigr)\frac{\phi_{\lambda+r}(\kappa+c)\psi_{\lambda+r}(x)-\phi_{\lambda+r}(x)\psi_{\lambda+r}(\kappa+c)}{\phi_{\lambda+r}(\kappa+c)\psi_{\lambda+r}(z_*)-\phi_{\lambda+r}(z_*)\psi_{\lambda+r}(\kappa+c)}\\
&\hspace{1cm}+\bigl(c+\widehat{v}(\kappa+c)\bigr)\frac{\phi_{\lambda+r}(z_*)\psi_{\lambda+r}(x)-\phi_{\lambda+r}(x)\psi_{\lambda+r}(z_*)}{\phi_{\lambda+r}(\kappa+c)\psi_{\lambda+r}(z_*)-\phi_{\lambda+r}(z_*)\psi_{\lambda+r}(\kappa+c)}+\widehat{v}(x),
\end{split}
\end{equation}
for all $x\in(z_*,\kappa+c)$. Substituting the derivative of this expression into \eqref{eqSec3:FreeBndProb_c} produces the following implicit characterisation of the free boundary $z_*\in(0,\kappa+c)$:
\begin{equation}
\label{eqSec3:FreeBndEqn}
\begin{split}
&\bigl(\kappa-z_*-\widehat{v}(z_*)\bigr)\frac{\phi_{\lambda+r}(\kappa+c)\psi_{\lambda+r}'(z_*)-\phi_{\lambda+r}'(z_*)\psi_{\lambda+r}(\kappa+c)}{\phi_{\lambda+r}(\kappa+c)\psi_{\lambda+r}(z_*)-\phi_{\lambda+r}(z_*)\psi_{\lambda+r}(\kappa+c)}\\
&\hspace{0.5cm}+\bigl(c+\widehat{v}(\kappa+c)\bigr)\frac{\phi_{\lambda+r}(z_*)\psi_{\lambda+r}'(z_*)-\phi_{\lambda+r}'(z_*)\psi_{\lambda+r}(z_*)}{\phi_{\lambda+r}(\kappa+c)\psi_{\lambda+r}(z_*)-\phi_{\lambda+r}(z_*)\psi_{\lambda+r}(\kappa+c)}+\widehat{v}\,'(z_*)=-1.
\end{split}
\end{equation}
Finally, given $\alpha>0$, recall that the resolvent operator $\Resolvent{\alpha}$ acts on suitable functions $\func{g}{(0,\infty)}{\R}$ as follows:
\begin{equation*}
(\Resolvent{\alpha}g)(x)\coloneqq\E_x\biggl(\int_0^\infty\e^{-\alpha t}g(X_t)\,\d t\biggr),
\end{equation*}
for all $x\in(0,\infty)$ \citep[see][Section~I.7]{BS02}. Moreover, $\Resolvent{\alpha}g$ satisfies the ordinary differential equation
\begin{equation*}
\L_X(\Resolvent{\alpha}g)(x)-\alpha(\Resolvent{\alpha}g)(x)+g(x)=0,
\end{equation*}
for all $x\in(0,\infty)$. By comparing this equation with \eqref{eqSec3:FreeBndProb_a}, we see that the particular solution to the latter equation is given by
\begin{equation}
\label{eqSec3:ODEPartSol}
\begin{split}
\widehat{v}(x)&\coloneqq\bigl(\Resolvent{\lambda+r}\lambda(\kappa-\,\cdot\,)\bigr)(x)
=\E_x\biggl(\int_0^\infty\e^{-(\lambda+r)t}\lambda(\kappa-X_t)\,\d t\biggr)\\
&=\frac{\lambda\kappa}{\lambda+r}-\int_0^\infty\lambda\e^{-(\lambda+r)t}\E_x(X_t)\,\d t
=\frac{\lambda\kappa}{\lambda+r}-\int_0^\infty\lambda\e^{-(\lambda+r)t}x\e^{\mu t}\,\d t\\
&=\frac{\lambda\kappa}{\lambda+r}-\frac{\lambda x}{\lambda+r-\mu},
\end{split}
\end{equation}
for all $x\in(0,\infty)$.

\subsection{Identifying the parameter regimes}
The previous analysis demonstrates that the existence of a solution $z_*\in(0,\kappa+c)$ to \eqref{eqSec3:FreeBndEqn} is necessary for Problem \eqref{eqSec3:FreeBndProb} to admit a solution. Conversely, the function $\widehat{V}\in\C{0}{0}(0,\infty)\cap\C{2}{0}(z_*,\kappa+c)$ defined by \eqref{eqSec3:FreeBndProbSol} satisfies the ordinary differential equation \eqref{eqSec3:FreeBndProb_a}, as well as the boundary conditions \eqref{eqSec3:FreeBndProb_b} and \eqref{eqSec3:FreeBndProb_c}, if $z_*\in(0,\kappa+c)$ satisfies \eqref{eqSec3:FreeBndEqn}. Hence, Problem~\eqref{eqSec3:FreeBndProb} admits a unique solution if and only if the free-boundary equation \eqref{eqSec3:FreeBndEqn} admits a unique solution $z_*\in(0,\kappa+c)$. We shall now investigate the solutions to this equation in detail.

To begin with, consider the function $H\in\C{2}{0}(0,\infty)$, given by
\begin{equation*}
H(z)\coloneqq\frac{F'(z)G(z)-F(z)G'(z)}{w_{\lambda+r}\scale'(z)}+F(\kappa+c),
\end{equation*}
for all $z\in(0,\infty)$, where $w_{\lambda+r}$ is the Wronskian \eqref{eqSec2:Wronskian} and $F,G\in\C{2}{0}(0,\infty)$ are given by
\begin{align*}
F(z)&\coloneqq\widehat{v}(z)-(\kappa-z)=\frac{r-\mu}{\lambda+r-\mu}z-\frac{r\kappa}{\lambda+r}\\
\intertext{and}
G(z)&\coloneqq\phi_{\lambda+r}(\kappa+c)\psi_{\lambda+r}(z)-\phi_{\lambda+r}(z)\psi_{\lambda+r}(\kappa+c),
\end{align*}
for all $z\in(0,\infty)$. It follows that $z_*\in(0,\kappa+c)$ satisfies \eqref{eqSec3:FreeBndEqn} if and only if $H(z_*)=0$, by virtue of the identity
\begin{equation*}
w_{\lambda+r}\scale'(z)=\phi_{\lambda+r}(z)\psi_{\lambda+r}'(z)-\phi_{\lambda+r}'(z)\psi_{\lambda+r}(z),	
\end{equation*}
for all $z\in(0,\infty)$. Now, observe that
\begin{equation*}
\begin{split}
H'(z)&=\frac{F''(z)G(z)-F(z)G''(z)+\frac{2\mu z}{\sigma^2z^2}\bigl(F'(z)G(z)-F(z)G'(z)\bigr)}{w_{\lambda+r}\scale'(z)}\\
&=\frac{F''(z)G(z)-\frac{2}{\sigma^2z^2}F(z)\bigl((\lambda+r)G(z)-\mu zG'(z)\bigr)+\frac{2\mu z}{\sigma^2z^2}\bigl(F'(z)G(z)-F(z)G'(z)\bigr)}{w_{\lambda+r}\scale'(z)}\\
&=\frac{2}{\sigma^2z^2}\frac{\L_XF(z)-(\lambda+r)F(z)}{w_{\lambda+r}\scale'(z)}G(z)\\
&=\frac{2}{\sigma^2z^2}\frac{\L_X\widehat{v}(z)-(\lambda+r)\widehat{v}(z)+\mu z+(\lambda+r)(\kappa-z)}{w_{\lambda+r}\scale'(z)}G(z)\\
&=\frac{2}{\sigma^2z^2}\frac{r\kappa+(\mu-r)z}{w_{\lambda+r}\scale'(z)}G(z),
\end{split}
\end{equation*}
for all $z\in(0,\infty)$, where the first equality follows from the identity $\scale''(z)=-\frac{2\mu}{\sigma^2z}\scale'(z)$ \citep[see][Section~II.9]{BS02}, the second equality follows from the fact that $G$ satisfies \eqref{eqSec2:ODE} with $\alpha=\lambda+r$, and the final equality follows since $\widehat{v}$ satisfies \eqref{eqSec3:FreeBndProb_a}. Note that $G(\kappa+c)=0$ implies that $H'(\kappa+c)=0$. Also note that $G(\kappa+c)=0$ and $G'(\kappa+c)=w_{\lambda+r}\scale'(\kappa+c)$ imply that $H(\kappa+c)=0$. That is to say, $H$ has a root at $\kappa+c$, which is also a stationary point. However, $\kappa+c$ is not a solution to \eqref{eqSec3:FreeBndEqn}, since the left-hand side of that equation is not well-defined if $z_*=\kappa+c$. Next, since $\phi_{\lambda+r}$ and $\psi_{\lambda+r}$ are strictly decreasing and increasing, respectively, it follows that $G$ is strictly increasing. Hence, $G(z)<0$, for all $z\in(0,\kappa+c)$; $G(\kappa+c)=0$; and $G(z)>0$, for all $z\in(\kappa+c,\infty)$. Finally, note that $\scale'(z)>0$, for all $z\in(0,\infty)$.

The next proposition uses the previous observations to give a complete description of the roots of $H$. In so doing, it identifies seven parameter regimes that will prove useful for organising the solution to Problem~\eqref{eqSec3:OptStopProb} in Section~\ref{Sec4}.

\begin{proposition}
The roots of $H$ depend on the model parameters as follows.
\label{propSec3:SolCond}
\begin{enumerate}[leftmargin=*,topsep=1ex,itemsep=1ex,itemindent=0ex,label=(\alph*)]
\begin{subequations}
\item
If the parameters satisfy
\begin{equation}
\label{eqpropSec3:Cond1}
\mu<\frac{rc}{\kappa+c}\qquad\text{and}\qquad r>0,
\end{equation}
then the root at $\kappa+c$ is a local maximum point of $H$, and the function has a second root at some point
\begin{equation*}
z_0\in\biggl(0,\frac{r\kappa}{r-\mu}\biggr)\subset(0,\kappa+c).
\end{equation*}
\item
If the parameters satisfy
\begin{equation}
\label{eqpropSec3:Cond2}
\mu=\frac{rc}{\kappa+c}\qquad\text{and}\qquad r>0,
\end{equation}
then the root at $\kappa+c$ is unique, and is also an inflection point of $H$.
\item
If the parameters satisfy
\begin{equation}
\label{eqpropSec3:Cond3}
\frac{rc}{\kappa+c}<\mu<r\qquad\text{and}\qquad r>0,
\end{equation}
then the root at $\kappa+c$ is a local minimum point of $H$, and the function has a second root at some point
\begin{equation*}
z_0\in\biggl(\frac{r\kappa}{r-\mu},\infty\biggr)\subset(\kappa+c,\infty).
\end{equation*}
\item
If the parameters satisfy
\begin{equation}
\label{eqpropSec3:Cond4}
\mu\geq r\qquad\text{and}\qquad r>0,
\end{equation}
then the root at $\kappa+c$ is unique, and is also the global minimum point of $H$.
\item
If the parameters satisfy
\begin{equation}
\label{eqpropSec3:Cond5}
\mu<0\qquad\text{and}\qquad r=0,
\end{equation}
then the root at $\kappa+c$ is unique, and is also the global maximum point of $H$.
\item
If the parameters satisfy
\begin{equation}
\label{eqpropSec3:Cond6}
\mu=0\qquad\text{and}\qquad r=0,
\end{equation}
then $H$ is identically zero.
\item
If the parameters satisfy
\begin{equation}
\label{eqpropSec3:Cond7}
\mu>0\qquad\text{and}\qquad r=0,
\end{equation}
then the root at $\kappa+c$ is unique, and is also the global minimum point of $H$.
\end{subequations}
\end{enumerate}
\end{proposition}
\begin{proof}
(a):~Suppose Condition~\eqref{eqpropSec3:Cond1} holds, in which case
\begin{equation*}
0<r-\frac{rc}{\kappa+c}=\frac{r\kappa}{\kappa+c}<r-\mu,
\end{equation*}
so that the point $\bar{z}\coloneqq\sfrac{r\kappa}{(r-\mu)}$ satisfies $\bar{z}\in(0,\kappa+c)$. Observe that $r\kappa+(\mu-r)z>0$, for all $z\in(0,\bar{z})$; $r\kappa+(\mu-r)\bar{z}=0$; and $r\kappa+(\mu-r)z<0$, for all $z\in(\bar{z},\infty)$. Combined with the properties of $G$ and $\scale'$ described earlier, this implies that $H'(z)<0$, for all $z\in(0,\bar{z})$; $H'(\bar{z})=0$; $H'(z)>0$, for all $z\in(\bar{z},\kappa+c)$; $H'(\kappa+c)=0$; and $H'(z)<0$, for all $z\in(\kappa+c,\infty)$. In particular, $H$ has stationary points at $\bar{z}$ and $\kappa+c$, with former being a local minimum and the latter being a local maximum. Since $H(\kappa+c)=0$ and $H$ is strictly increasing over $(\bar{z},\kappa+c)$, it follows that $H(\bar{z})<0$. Furthermore, $H(\kappa+c)=0$ rules out the existence of roots in the intervals $(\bar{z},\kappa+c)$ and $(\kappa+c,\infty)$, since $H$ is strictly increasing over the former interval and strictly decreasing over the latter interval. Finally, we use \eqref{eqSec2:PhiPrimePsiPrime} to write
\begin{equation*}
\begin{split}
H(z)&=\frac{r-\mu}{\lambda+r-\mu}\frac{\phi_{\lambda+r}(\kappa+c)\psi_{\lambda+r}(z)-\phi_{\lambda+r}(z)\psi_{\lambda+r}(\kappa+c)}{w_{\lambda+r}\scale'(z)}\\
&\hspace{0.5cm}+\biggl(\frac{r\kappa}{\lambda+r}-\frac{r-\mu}{\lambda+r-\mu}z\biggr)\frac{\phi_{\lambda+r}(\kappa+c)\psi_{\lambda+r}'(z)-\phi_{\lambda+r}'(z)\psi_{\lambda+r}(\kappa+c)}{w_{\lambda+r}\scale'(z)}+F(\kappa+c)\\
&=\frac{r-\mu}{\lambda+r-\mu}\frac{\phi_{\lambda+r}(\kappa+c)}{w_{\lambda+r}}\Biggl(\frac{z}{\sqrt{\nu^2+\frac{2(\lambda+r)}{\sigma^2}}-\nu}+\frac{r\kappa}{\lambda+r}\frac{\lambda+r-\mu}{r-\mu}-z\Biggr)\frac{\psi_{\lambda+r}'(z)}{\scale'(z)}\\
&\hspace{0.5cm}+\frac{r-\mu}{\lambda+r-\mu}\frac{\psi_{\lambda+r}(\kappa+c)}{w_{\lambda+r}}\Biggl(\frac{z}{\sqrt{\nu^2+\frac{2(\lambda+r)}{\sigma^2}}+\nu}-\frac{r\kappa}{\lambda+r}\frac{\lambda+r-\mu}{r-\mu}+z\Biggr)\frac{\phi_{\lambda+r}'(z)}{\scale'(z)}\\
&\hspace{0.51cm}+F(\kappa+c),
\end{split}
\end{equation*}
for all $z\in(0,\infty)$. Now, since
\begin{align*}
\lim_{z\downarrow 0}\frac{\phi_{\lambda+r}'(z)}{\scale'(z)}
&=\lim_{z\downarrow 0}\Biggl(-\sqrt{\nu^2+\frac{2(\lambda+r)}{\sigma^2}}-\nu\Biggr)z^{-\sqrt{\nu^2+\sfrac{2(\lambda+r)}{\sigma^2}}+\nu}=-\infty\\
\intertext{and}
\lim_{z\downarrow 0}\frac{\psi_{\lambda+r}'(z)}{\scale'(z)}
&=\lim_{z\downarrow 0}\Biggl(\sqrt{\nu^2+\frac{2(\lambda+r)}{\sigma^2}}-\nu\Biggr)z^{\sqrt{\nu^2+\sfrac{2(\lambda+r)}{\sigma^2}}+\nu}=0,
\end{align*}
by virtue of \eqref{eqSec2:ScaleSpeed} and \eqref{eqSec2:PhiPsi}, it follows that $H(0+)=\infty$.\footnote{Note that the values for the two limits above can also be inferred from the fact that the origin is a natural boundary for $X$ \citep[see][Section~II.10]{BS02}.} This, in turn, ensures the existence of a unique root $z_0\in(0,\bar{z})$, since $H$ is strictly decreasing over $(0,\bar{z})$, with $H(\bar{z})<0$.
\vspace{2mm}\newline\noindent
(b):~Suppose Condition~\eqref{eqpropSec3:Cond2} holds, in which case
$r-\mu=\sfrac{r\kappa}{(\kappa+c)}>0$. Consequently, $r\kappa+(\mu-r)z>0$, for all $z\in(0,\kappa+c)$; $r\kappa+(\mu-r)(\kappa+c)=0$; and $r\kappa+(\mu-r)z<0$, for all $z\in(\kappa+c,\infty)$. Combined with the properties of $G$ and $\scale'$ described earlier, this implies that $H'(z)<0$, for all $z\in(0,\kappa+c)$; $H'(\kappa+c)=0$; and $H'(z)<0$, for all $z\in(\kappa+c,\infty)$. That is to say, $H$ is strictly decreasing over $(0,\kappa+c)$ and $(\kappa+c,\infty)$, with an inflection point at $\kappa+c$. Since $H(\kappa+c)=0$, it follows that $\kappa+c$ is the only root of $H$.
\vspace{2mm}\newline\noindent
(c):~Suppose Condition~\eqref{eqpropSec3:Cond3} holds, in which case
\begin{equation*}
0<r-\mu<r-\frac{rc}{\kappa+c}=\frac{r\kappa}{\kappa+c},
\end{equation*}
so that the point $\bar{z}\coloneqq\sfrac{r\kappa}{(r-\mu)}$ satisfies $\bar{z}\in(\kappa+c,\infty)$. Observe that $r\kappa+(\mu-r)z>0$, for all $z\in(0,\bar{z})$; $r\kappa+(\mu-r)\bar{z}=0$; and $r\kappa+(\mu-r)z<0$, for all $z\in(\bar{z},\infty)$. Combined with the properties of $G$ and $\scale'$ described earlier, this implies that $H'(z)<0$, for all $z\in(0,\kappa+c)$; $H'(\kappa+c)=0$; $H'(z)>0$, for all $z\in(\kappa+c,\bar{z})$; $H'(\bar{z})=0$; and $H'(z)<0$, for all $z\in(\bar{z},\infty)$. In particular, $H$ has stationary points at $\kappa+c$ and $\bar{z}$, with former being a local minimum and the latter being a local maximum. Since $H(\kappa+c)=0$ and $H$ is strictly increasing over $(\kappa+c,\bar{z})$, it follows that $H(\bar{z})>0$. Furthermore, $H(\kappa+c)=0$ rules out the existence of roots in the intervals $(0,\kappa+c)$ and $(\kappa+c,\bar{z})$, since $H$ is strictly decreasing over the former interval and strictly increasing over the latter interval. Finally, we use \eqref{eqSec2:PhiPrimePsiPrime} to write
\begin{equation*}
\begin{split}
H(z)&=\frac{r-\mu}{\lambda+r-\mu}\frac{\phi_{\lambda+r}(\kappa+c)\psi_{\lambda+r}(z)-\phi_{\lambda+r}(z)\psi_{\lambda+r}(\kappa+c)}{w_{\lambda+r}\scale'(z)}\\
&\hspace{0.5cm}+\biggl(\frac{r\kappa}{\lambda+r}-\frac{r-\mu}{\lambda+r-\mu}z\biggr)\frac{\phi_{\lambda+r}(\kappa+c)\psi_{\lambda+r}'(z)-\phi_{\lambda+r}'(z)\psi_{\lambda+r}(\kappa+c)}{w_{\lambda+r}\scale'(z)}+F(\kappa+c)\\
&=\frac{r-\mu}{\lambda+r-\mu}\frac{\phi_{\lambda+r}(\kappa+c)}{w_{\lambda+r}}\\
&\hspace{2cm}\times\Biggl(1+\Biggl(\sqrt{\nu^2+\frac{2(\lambda+r)}{\sigma^2}}-\nu\Biggr)\biggl(\frac{r\kappa}{\lambda+r}\frac{\lambda+r-\mu}{r-\mu}\frac{1}{z}-1\biggr)\Biggr)\frac{\psi_{\lambda+r}(z)}{\scale'(z)}\\
&\hspace{0.5cm}-\frac{r-\mu}{\lambda+r-\mu}\frac{\psi_{\lambda+r}(\kappa+c)}{w_{\lambda+r}}\\
&\hspace{2cm}\times\Biggl(1-\Biggl(\sqrt{\nu^2+\frac{2(\lambda+r)}{\sigma^2}}+\nu\Biggr)\biggl(\frac{r\kappa}{\lambda+r}\frac{\lambda+r-\mu}{r-\mu}\frac{1}{z}-1\biggr)\Biggr)\frac{\phi_{\lambda+r}(z)}{\scale'(z)}\\
&\hspace{0.5cm}+F(\kappa+c),
\end{split}
\end{equation*}
for all $z\in(0,\infty)$. Note that
\begin{equation*}
(\nu+1)^2-\nu^2=2\nu+1=\frac{2\mu}{\sigma^2}<\frac{2(\lambda+r)}{\sigma^2}
\end{equation*}
since $\mu<r<\lambda+r$, from which it follows that
\begin{equation}
\label{eqpropSec3:ParamInequal}
\sqrt{\nu^2+\frac{2(\lambda+r)}{\sigma^2}}>\nu+1.	
\end{equation}
Consequently,
\begin{equation*}
\lim_{z\uparrow\infty}\Biggl(\sqrt{\nu^2+\frac{2(\lambda+r)}{\sigma^2}}-\nu\Biggr)\biggl(\frac{r\kappa}{\lambda+r}\frac{\lambda+r-\mu}{r-\mu}\frac{1}{z}-1\biggr)<-1.
\end{equation*}
Moreover, \eqref{eqSec2:ScaleSpeed} and \eqref{eqSec2:PhiPsi} give
\begin{align*}
\lim_{z\uparrow\infty}\frac{\phi_{\lambda+r}(z)}{\scale'(z)}
&=\lim_{z\uparrow\infty}z^{-\sqrt{\nu^2+\sfrac{2(\lambda+r)}{\sigma^2}}+\nu+1}=0\\
\intertext{and}
\lim_{z\uparrow\infty}\frac{\psi_{\lambda+r}(z)}{\scale'(z)}
&=\lim_{z\uparrow\infty}z^{\sqrt{\nu^2+\sfrac{2(\lambda+r)}{\sigma^2}}+\nu+1}=\infty,
\end{align*}
where the first limit follows from \eqref{eqpropSec3:ParamInequal}. Putting all of this together gives $H(\infty-)=-\infty$. This, in turn, ensures the existence of a unique root $z_0\in(\bar{z},\infty)$, since $H$ is strictly decreasing over $(\bar{z},\infty)$ with $H(\bar{z})>0$.
\vspace{2mm}\newline\noindent
(d)~Suppose Condition~\eqref{eqpropSec3:Cond4} holds. Then $r\kappa+(\mu-r)z\geq r\kappa>0$, for all $z\in(0,\infty)$. Combined with the properties of $G$ and $\scale'$ described earlier, this implies that $H'(z)<0$, for all $z\in(0,\kappa+c)$; $H'(\kappa+c)=0$; and $H'(z)>0$, for all $z\in(\kappa+c,\infty)$. That is to say, $H$ achieves a unique global minimum at $\kappa+c$. Moreover, since $H(\kappa+c)=0$, it follows that $H(z)>0$, for all $z\in(0,\kappa+c)\cup(\kappa+c,\infty)$. In other words, $H$ has a unique root at $\kappa+c$.
\vspace{2mm}\newline\noindent
(e)~Suppose Condition~\eqref{eqpropSec3:Cond5} holds. Then $r\kappa+(\mu-r)z=\mu z<0$, for all $z\in(0,\infty)$. Combined with the properties of $G$ and $\scale'$ described earlier, this implies that $H'(z)>0$, for all $z\in(0,\kappa+c)$; $H'(\kappa+c)=0$; and $H'(z)<0$, for all $z\in(\kappa+c,\infty)$. That is to say, $H$ achieves a unique global maximum at $\kappa+c$. Since $H(\kappa+c)=0$, it follows that $\kappa+c$ is the only root of $H$.
\vspace{2mm}\newline\noindent
(f):~Suppose Condition~\eqref{eqpropSec3:Cond6} holds. Then $F(z)=0$, for all $z\in(0,\infty)$. Consequently, $H(z)=0$, for all $z\in(0,\infty)$.
\vspace{2mm}\newline\noindent
(g)~Suppose Condition~\eqref{eqpropSec3:Cond7} holds. Then $r\kappa+(\mu-r)z=\mu z>0$, for all $z\in(0,\infty)$. Combined with the properties of $G$ and $\scale'$ described earlier, this implies that $H'(z)<0$, for all $z\in(0,\kappa+c)$; $H'(\kappa+c)=0$; and $H'(z)>0$, for all $z\in(\kappa+c,\infty)$. That is to say, $H$ achieves a unique global minimum at $\kappa+c$. Since $H(\kappa+c)=0$, it follows that $\kappa+c$ is the only root of $H$.
\end{proof}

Figure~\ref{figSec3:H} plots $H$ under the seven parameter regimes identified in Proposition~\ref{propSec3:SolCond}. Figure~\ref{figSec3:H}\subref{figSec3:H_1} shows that $H$ has a local minimum at $\sfrac{\kappa c}{(r-\mu)}\in(0,\kappa+c)$ and a local maximum at $\kappa+c$, if Condition~\eqref{eqpropSec3:Cond1} holds. The latter point is also a root, and there is a second root at some point $z_0\in(0,\sfrac{r\kappa}{(r-\mu)})$. Figure~\ref{figSec3:H}\subref{figSec3:H_2} shows that $H$ has a unique root at $\kappa+c$, which is also an inflection point, if Condition~\eqref{eqpropSec3:Cond2} holds. Figure~\ref{figSec3:H}\subref{figSec3:H_3} illustrates the situation when Condition~\eqref{eqpropSec3:Cond3} holds, in which case $H$ has a local minimum at $\kappa+c$ and a local maximum at $\sfrac{r\kappa}{(r-\mu)}\in(\kappa+c,\infty)$. The former point is also a root, and the function has a second root at some point $z_0\in(\sfrac{r\kappa}{(r-\mu)},\infty)$. Figure~\ref{figSec3:H}\subref{figSec3:H_4} shows that $H$ has a global minimum at $\kappa+c$, if Condition~\eqref{eqpropSec3:Cond4} holds, in which case that point is the only root. In Figure~\ref{figSec3:H}\subref{figSec3:H_5} we see that $H$ has a global maximum at $\kappa+c$, if Condition~\eqref{eqpropSec3:Cond5} holds, and that point is also the only root. Next, Figure~\ref{figSec3:H}\subref{figSec3:H_6} shows that $H$ is identically zero, if Condition~\eqref{eqpropSec3:Cond6} holds. Finally, Figure~\ref{figSec3:H}\subref{figSec3:H_7} illustrates the situation under Condition~\eqref{eqpropSec3:Cond7}, in which case $\kappa+c$ is the global minimum point of $H$, as well as the only root of the function.

\begin{figure}
\centering
\mbox{
\subfigure[]{\includegraphics[scale=0.6]{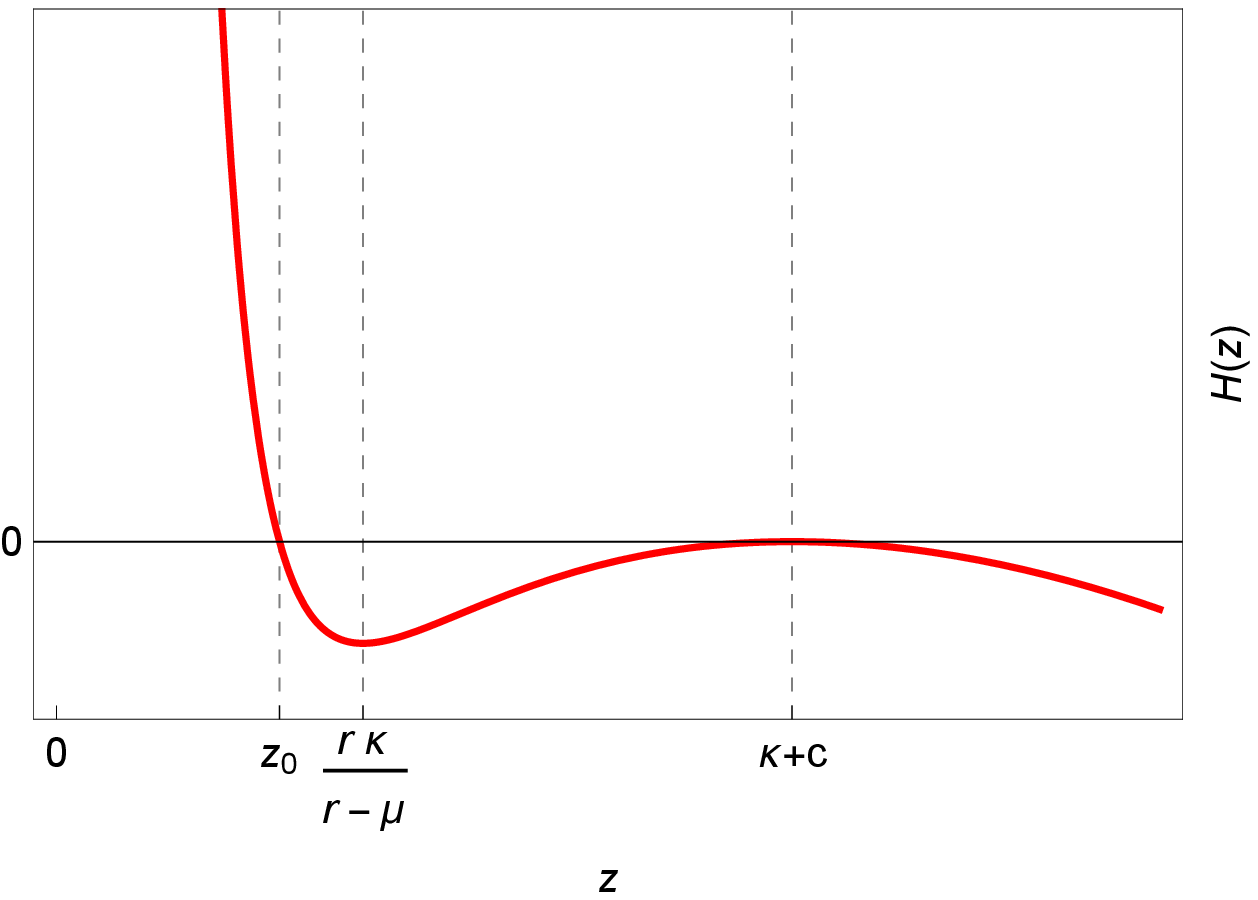}
\label{figSec3:H_1}}
\subfigure[]{\includegraphics[scale=0.6]{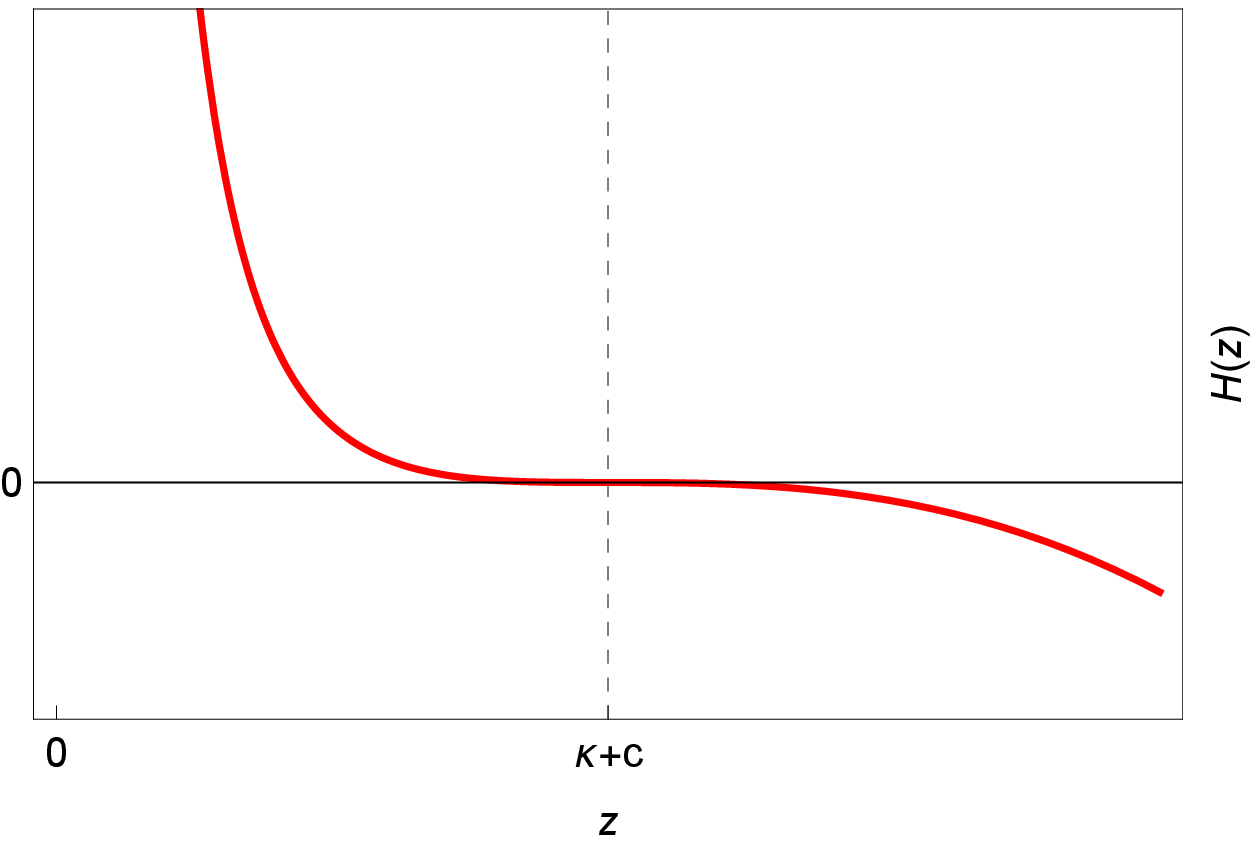}
\label{figSec3:H_2}}
}
\mbox{
\subfigure[]{\includegraphics[scale=0.6]{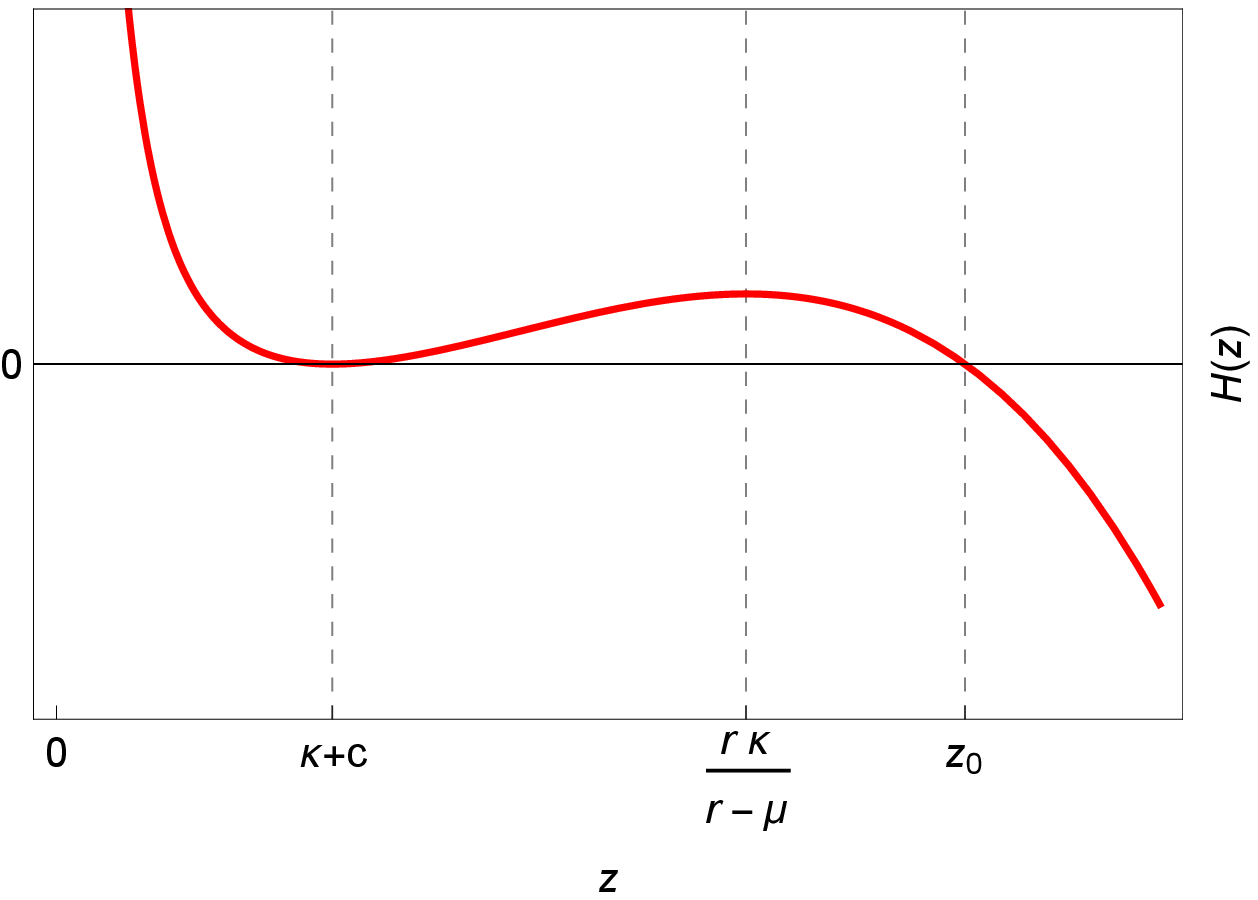}
\label{figSec3:H_3}}
\subfigure[]{\includegraphics[scale=0.6]{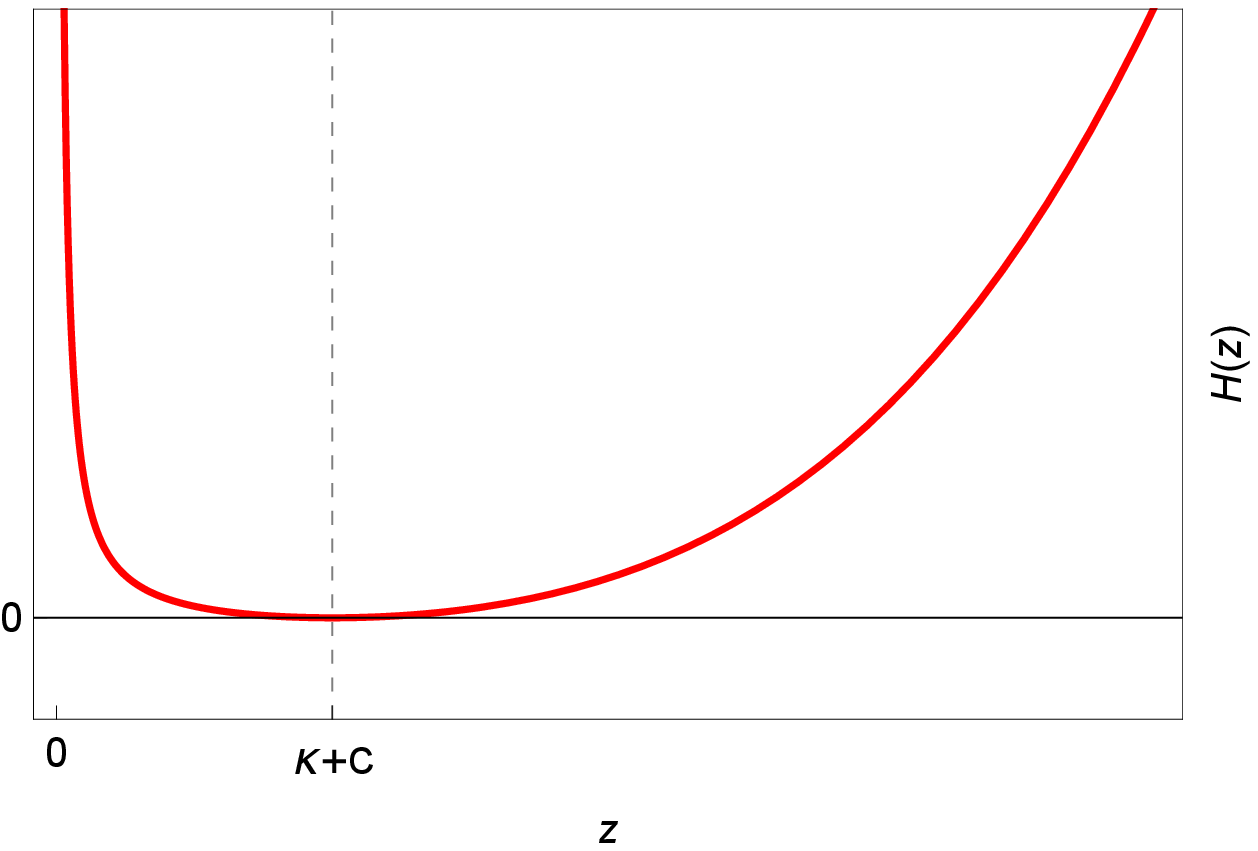}
\label{figSec3:H_4}}
}
\mbox{
\subfigure[]{\includegraphics[scale=0.6]{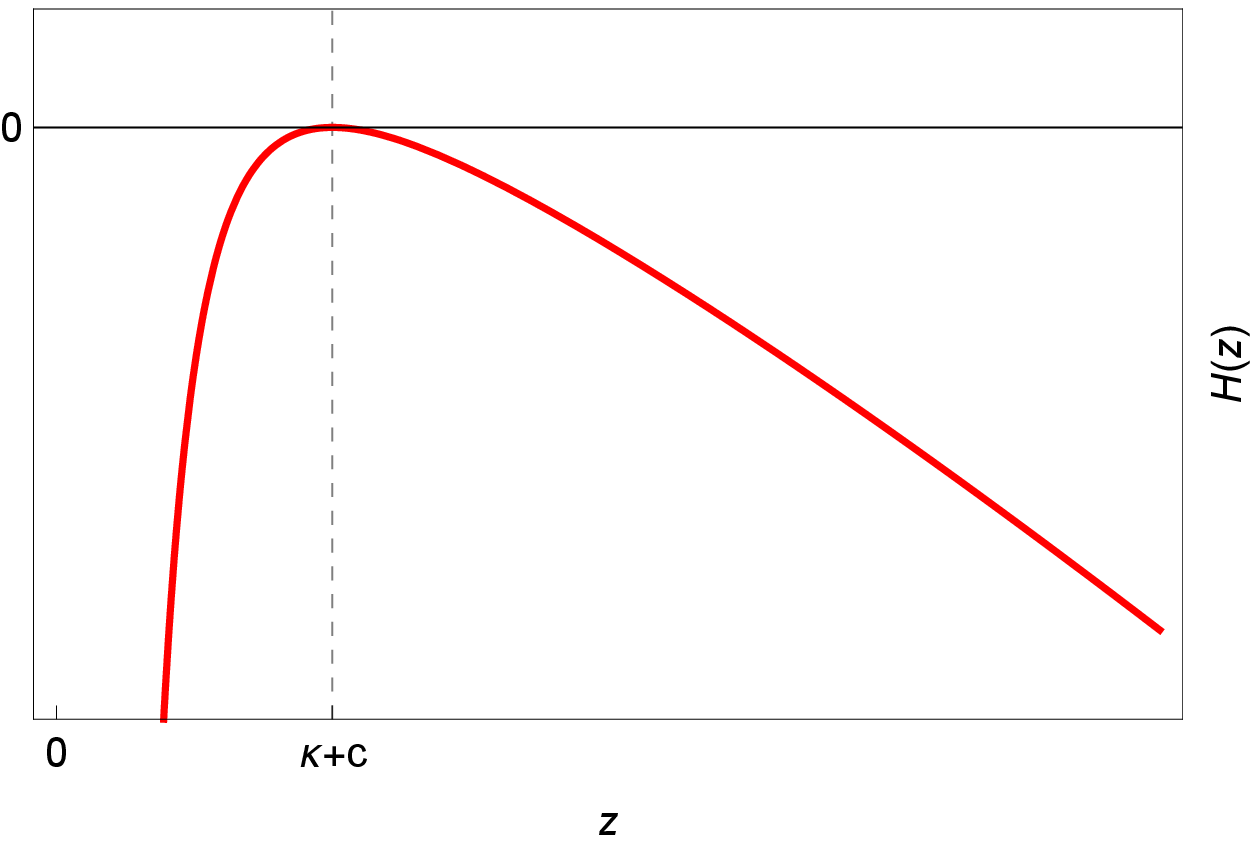}
\label{figSec3:H_5}}
\subfigure[]{\includegraphics[scale=0.6]{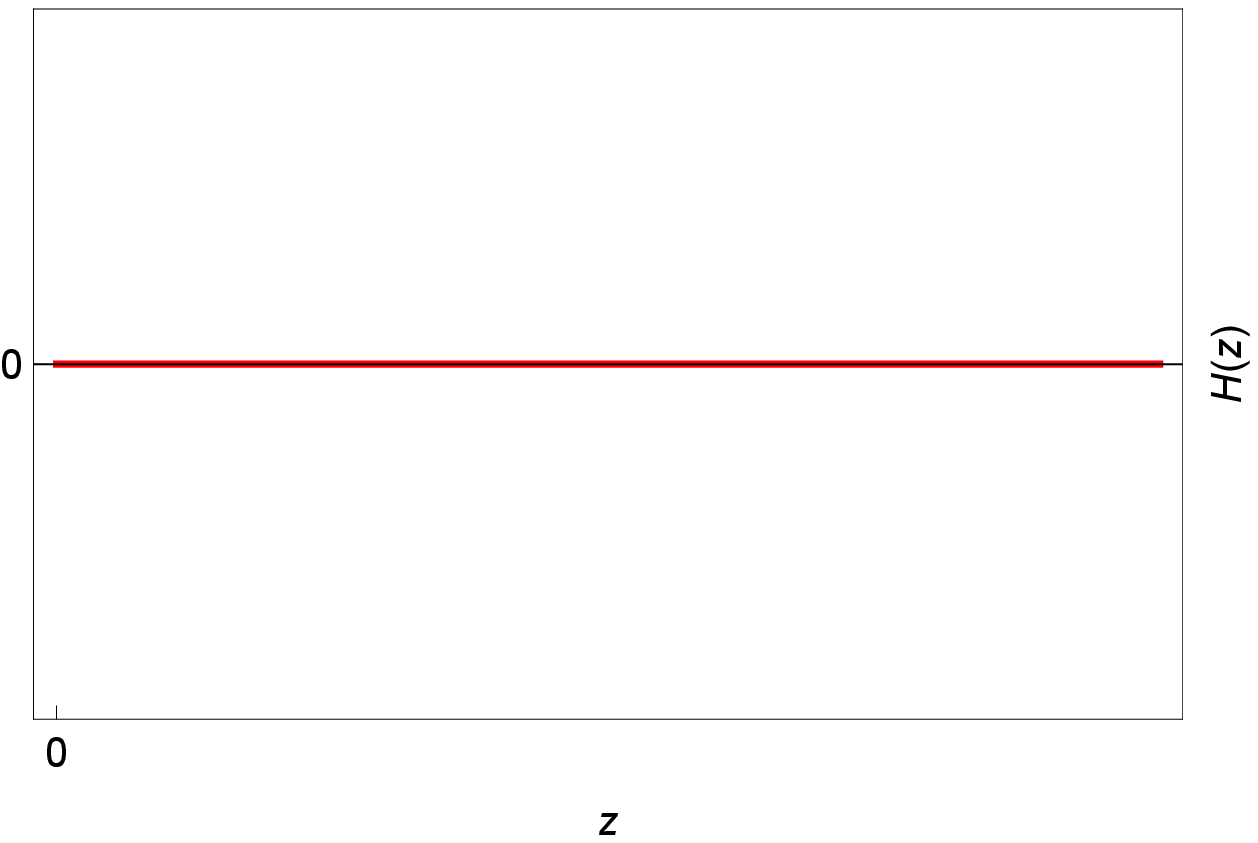}
\label{figSec3:H_6}}
}
\subfigure[]{\includegraphics[scale=0.6]{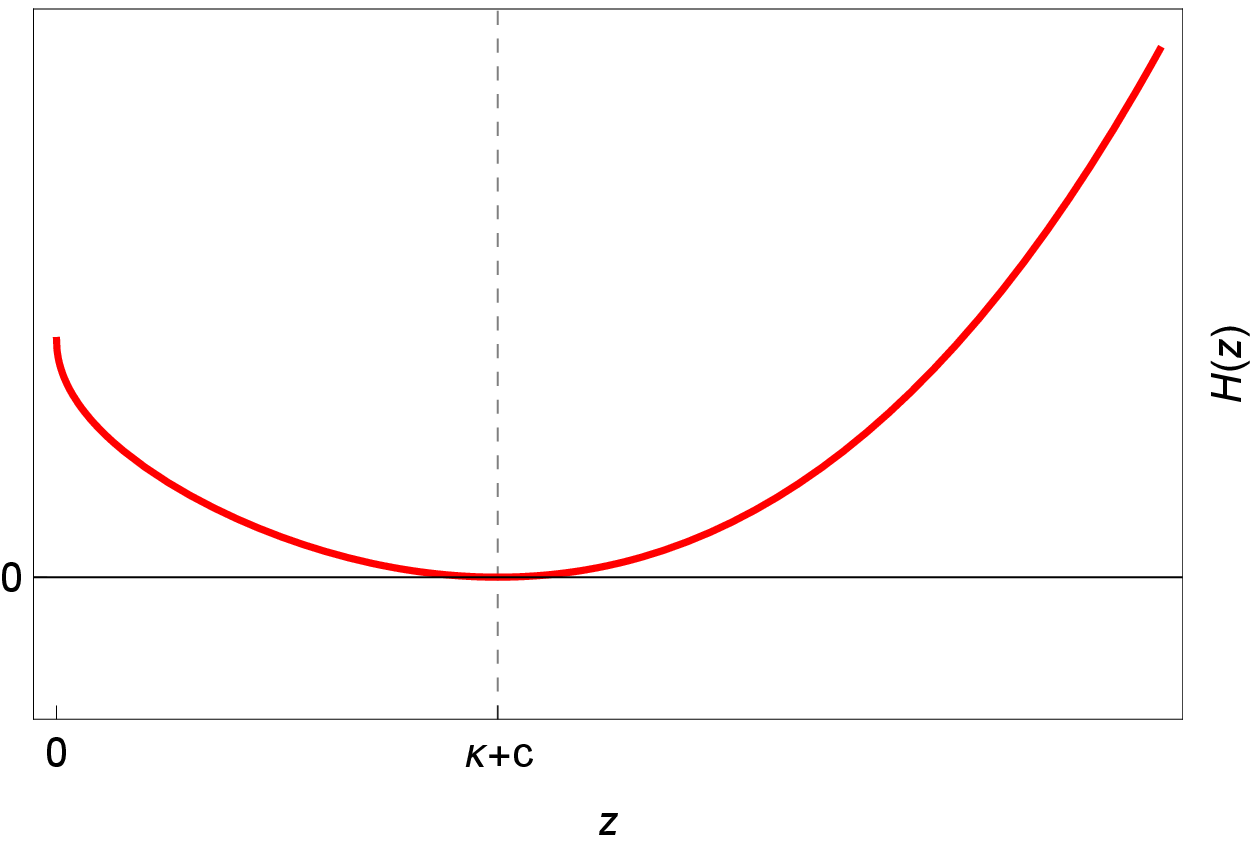}
\label{figSec3:H_7}}
\caption{Graphs of $H$ under the seven parameter regimes in Proposition~\ref{propSec3:SolCond}.}
\label{figSec3:H}
\end{figure}
\section{Solution and Verification for the Simpler Problem}
\label{Sec4}
\subsection{Constructing candidate solutions}
By analysing Proposition~\ref{propSec3:SolCond}, we can guess the optimal stopping policy for Problem~\eqref{eqSec3:OptStopProb}, under each of the parameter regimes identified there. Based on those guesses, we can derive expressions for the candidate value function under each regime. 

First, if Condition~\eqref{eqpropSec3:Cond1} holds, Problem~\eqref{eqSec3:FreeBndProb} admits a unique solution, comprising a free boundary $z_*\in(0,\kappa+c)$ and a function $\widehat{V}\in\C{0}{0}(0,\infty)\cap\C{2}{0}(z_*,\kappa+c)$. In detail, $z_*\coloneqq z_0$ satisfies the free boundary equation \eqref{eqSec3:FreeBndEqn}, where $z_0\in(0,\sfrac{r\kappa}{(r-\mu)})\subset(0,\kappa+c)$ is the root of $H$, whose existence was established by Proposition~\ref{propSec3:SolCond}(a), while $\widehat{V}$ is determined by \eqref{eqSec3:FreeBndProbSol} over $(z_*,\kappa+c)$ and $\widehat{V}(x)\coloneqq\kappa-x$, for all $x\in(0,z_*]\cup[\kappa+c,\infty)$. We speculate that $z_*$ is the optimal stopping threshold for Problem~\eqref{eqSec3:OptStopProb} and $\widehat{V}$ is the corresponding value function. This is economically plausible, since it implies that Problem~\eqref{eqSec3:OptStopProb} has a unique non-trivial solution when the drift rate of the stock price is less than the threshold $\sfrac{rc}{(\kappa+c)}$ and the discount rate is positive. A low drift rate ensures that there may be some value in waiting for the stock price to fall before closing out the short position (i.e. immediate stopping is not always optimal), while a positive discount rate ensures that the short seller should not wait forever.

The following lemma derives a lower bound for the candidate value function specified above. It will be used to prove Theorem~\ref{thmSec4:VerThm}(a), which establishes that the candidate value function and the associated candidate optimal stopping threshold do indeed solve Problem~\eqref{eqSec3:OptStopProb} under Condition~\eqref{eqpropSec3:Cond1}.

\begin{lemma}
\label{lemSec4:ValFun1LowBnd}
Suppose Condition~(\ref{eqpropSec3:Cond1}) holds. Let $z_*\coloneqq z_0\in(0,\sfrac{r\kappa}{(r-\mu)})\subset(0,\kappa+c)$ be the unique solution to (\ref{eqSec3:FreeBndEqn}), whose existence is established by Proposition~\ref{propSec3:SolCond}(a), and let $\widehat{V}\in\C{0}{0}(0,\infty)\cap \C{2}{0}(z_*,\kappa+c)$ be determined by (\ref{eqSec3:FreeBndProbSol}) over $(z_*,\kappa+c)$ and $\widehat{V}(x)\coloneqq\kappa-x$, for all $x\in(0,z_*]\cup[\kappa+c,\infty)$. Then $\widehat{V}(x)>\kappa-x$, for all $x\in(z_*,\kappa+c)$.
\end{lemma}
\begin{proof}
Define the function $U\in\C{0}{0}(0,\infty)\cap\C{2}{0}(z_*,\kappa+c)$, by setting $U(x)\coloneqq\widehat{V}(x)-(\kappa-x)$, for all $x\in(0,\infty)$. We shall demonstrate that $U(x)>0$, for all $x\in(z_*,\kappa+c)$. To begin with, observe that
\begin{equation}
\label{eqlemSec4:ValFun1LowBnd_1}
\begin{split}
\L_XU(x)-(\lambda+r)U(x)&=\L_X\widehat{V}(x)-(\lambda+r)\widehat{V}(x)+\mu x+(\lambda+r)(\kappa-x)\\
&=-\lambda(\kappa-x)+\mu x+(\lambda+r)(\kappa-x)
=r\kappa-(r-\mu)x,
\end{split}
\end{equation}
for all $x\in(z_*,\kappa+c)$, by virtue of \eqref{eqSec3:FreeBndProb_a}. In particular,
\begin{equation}
\label{eqlemSec4:ValFun1LowBnd_2}
\begin{split}
\frac{1}{2}\sigma^2z_*^2U''(z_*+)
&=\frac{1}{2}\sigma^2z_*^2U''(z_*+)+\mu z_*U'(z_*+)-(\lambda+r)U(z_*+)\\
&=\L_XU(z_*+)-(\lambda+r)U(z_*+)
=r\kappa-(r-\mu)z_*>0,
\end{split}
\end{equation}
since $U(z_*+)=U(z_*)=0$ and $U'(z_*+)=0$, due to \eqref{eqSec3:FreeBndProb_b}, \eqref{eqSec3:FreeBndProb_c} and the continuity of $\widehat{V}$, and since $z_*<\sfrac{r\kappa}{(r-\mu)}$. So, $U''(z_*+)>0$,  which together with $U'(z_*+)=U(z_*+)=0$, ensures the existence of some $\varepsilon>0$, such that $U(x)>U(z_*)=0$, for all $x\in(z_*,z_*+\varepsilon)$. In particular, given any $x\in(z_*,\sfrac{r\kappa}{(r-\mu)}]$, it follows that
\begin{equation*}
\max_{\xi\in[z_*,x]}U(\xi)>U(z_*)=0,
\end{equation*}
whence $U$ has a positive maximum over $[z_*,x]$, which is realised either at an interior point of the interval or at the right end-point. However, given any $x\in(z_*,\sfrac{r\kappa}{(r-\mu)}]$, \eqref{eqlemSec4:ValFun1LowBnd_1} yields
\begin{equation*}
\L_XU(\xi)-(\lambda+r)U(\xi)=r\kappa-(r-\mu)\xi>0,
\end{equation*}
for all $\xi\in(z_*,x)\subseteq(z_*,\sfrac{r\kappa}{(r-\mu)})$. Based on this differential inequality, the maximum principle \citep[see][Theorem~1.3]{PW67} asserts that $U$ cannot realise its maximum in the interior of $[z_*,x]$, for any $x\in(z_*,\sfrac{r\kappa}{(r-\mu)}]$, since it is a non-constant function with a non-negative maximum over the interval. Consequently,
\begin{equation*}
U(x)=\max_{\xi\in[z_*,x]}U(\xi)>U(z_*)=0,
\end{equation*}
for all $x\in(z_*,\sfrac{r\kappa}{(r-\mu)}]$. Next, observe that
\begin{equation*}
\max_{x\in[\frac{r\kappa}{r-\mu},\kappa+c]}-U(x)\geq-U(\kappa+c)=0,
\end{equation*}
by virtue of \eqref{eqSec3:FreeBndProb_b}. That is to say, $-U$ has a non-negative maximum over $[\frac{r\kappa}{r-\mu},\kappa+c]$, which is realised either at an interior point of the interval or at the right end-point, since we have already established that $-U(\sfrac{r\kappa}{(r-\mu)})<0$. Another application of \eqref{eqlemSec4:ValFun1LowBnd_1} gives
\begin{equation*}
\L_X(-U)(x)-(\lambda+r)(-U)(x)
=-\bigl(\L_XU(x)-(\lambda+r)U(x)\bigr)
=(r-\mu)x-r\kappa>0,
\end{equation*}
for all $x\in(\sfrac{r\kappa}{(r-\mu)},\kappa+c)$. Once again, the maximum principle ensures that $-U$ cannot achieve its maximum in the interior of $[\sfrac{r\kappa}{(r-\mu)},\kappa+c]$, since it is a non-constant function with a non-negative maximum over the interval. It follows that $-U$ must have a unique maximum at the right end-point of the interval, which implies that $U(x)>U(\kappa+c)=0$, for all $x\in[\sfrac{r\kappa}{(r-\mu)},\kappa+c)$.
\end{proof}

\begin{figure}
\centering
\includegraphics[scale=0.6]{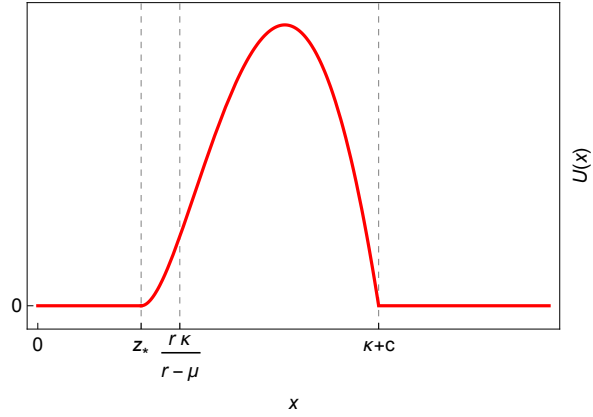}
\caption{The function $U$ defined in Lemma~\ref{lemSec4:ValFun1LowBnd}, under Condition~\eqref{eqpropSec3:Cond1}.}
\label{figSec4:U_1}
\end{figure}

Figure~\ref{figSec4:U_1} plots the function $U$ defined in Lemma~\ref{lemSec4:ValFun1LowBnd}, using parameter values that satisfy Condition~\eqref{eqpropSec3:Cond1}. We observe that $U(x)\geq 0$, for all $x\in(0,\infty)$, which is to say that $\widehat{V}(x)\geq\kappa-x$. Moreover, we see that $\widehat{V}(x)>\kappa-x$, for all $x\in(z_*,\kappa+c)$, as established by Lemma~\ref{lemSec4:ValFun1LowBnd}. Figure~\ref{figSec4:U_1} also illustrates the continuity of $U'$ at $z_*$, which follows from the fact that \eqref{eqSec3:FreeBndProb_c} ensures the continuity of $\widehat{V}'$ there. By contrast, Problem~\eqref{eqSec3:FreeBndProb} does not specify continuity of $\widehat{V}'$ at $\kappa+c$, as evident from the kink in $U$ at that point. Specifically, $U'((\kappa+c)-)<U'((\kappa+c)+)=0$, which implies that $\widehat{V}'((\kappa+c)-)<\widehat{V}((\kappa+c)+)=-1$. Finally, note that $U''$ is not continuous at $z_*$, since \eqref{eqlemSec4:ValFun1LowBnd_2} implies that $U''(z_*+)>0=U''(z_*-)$. This translates to $\widehat{V}''(z_*+)>0=\widehat{V}''(z_*-)$, which is not unexpected, since Problem~\eqref{eqSec3:FreeBndProb} does not impose a continuity requirement on $\widehat{V}''$ at $z_*$.

If Condition~\eqref{eqpropSec3:Cond2} holds, $H$ has a single root at $\kappa+c$, according to Proposition~\ref{propSec3:SolCond}(b). This implies that the free boundary equation \eqref{eqSec3:FreeBndEqn} does not admit a solution, whence Problem~\eqref{eqSec3:FreeBndProb} does not admit a solution either. However, the upper bound for the root $z_0\in(0,\sfrac{r\kappa}{(r-\mu)})\subset(0,\kappa+c)$ of $H$, which exists under Condition~\eqref{eqpropSec3:Cond1}, satisfies $\sfrac{r\kappa}{(r-\mu)}\uparrow\kappa+c$ as $\mu\uparrow\sfrac{rc}{(\kappa+c)}$, for any given $r>0$. This suggests that the root itself may satisfy $z_0\uparrow\kappa+c$ as $\mu\uparrow\sfrac{rc}{(\kappa+c)}$. Since that root is the candidate optimal stopping threshold for Problem~\eqref{eqSec3:OptStopProb}, under Condition~\eqref{eqpropSec3:Cond1}, $z_*\coloneqq\kappa+c$ is the natural candidate optimal stopping threshold, under Condition~\eqref{eqpropSec3:Cond2}. The candidate value function $\widehat{V}\in\C{2}{0}(0,\infty)$ is then determined by $\widehat{V}(x)\coloneqq\kappa-x$, for all $x\in(0,\infty)$. The economic interpretation is that the short seller should close out his position immediately, if the drift rate of the stock price is equals the threshold $\sfrac{rc}{(\kappa+c)}$ and the discount rate is positive. In that case, waiting for the stock price to fall is suboptimal, since it is expected to increase too quickly.

Next, suppose Condition~\eqref{eqpropSec3:Cond3} holds. In that case, Proposition~\ref{propSec3:SolCond}(c) ensures that $H$ possesses a root $z_0\in(\sfrac{r\kappa}{(r-\mu)},\infty)\subset(\kappa+c,\infty)$, which is also the unique solution to the free boundary equation \eqref{eqSec3:FreeBndEqn}. Since $\check{\tau}_{z_0}\wedge\hat{\tau}_{\kappa+c}=0$ $\P_x$-a.s., for all $x\in(0,\infty)$, we speculate that Problem~\eqref{eqSec3:OptStopProb} is solved by stopping immediately, in which case $z_*\coloneqq\kappa+c$ is the natural candidate for the optimal stopping threshold. The candidate value function $\widehat{V}\in\C{2}{0}(0,\infty)$ is then determined by $\widehat{V}(x)\coloneqq\kappa-x$, for all $x\in(0,\infty)$. This seems economically reasonable, since it suggests that the short seller should close out his position immediately if the drift rate of the stock price is large enough, relative to the discount rate, and the discount rate is positive. In other words, waiting for a fall in the stock price destroys value if the stock price is expected to appreciate at a high enough rate.

The same economic logic applies to the situation when Condition~\eqref{eqpropSec3:Cond4} holds, in which case the drift rate of the stock price is even higher relative to the discount rate. Once again, we surmise that the optimal stopping threshold is $z_*\coloneqq\kappa+c$ and the value function $\widehat{V}\in\C{2}{0}(0,\infty)$ is given by $\widehat{V}(x)\coloneqq\kappa-x$, for all $x\in(0,\infty)$.

Next, fix $\mu<0$, and observe that $\sfrac{r\kappa}{(r-\mu)}\downarrow 0$ as $r\downarrow 0$, which implies that $z_0\downarrow 0$ as $r\downarrow 0$, where $z_0\in(0,\sfrac{r\kappa}{(r-\mu)})$ is the root of $H$ in Proposition~\ref{propSec3:SolCond}(a). Based on the discussion of the situation when Condition~\eqref{eqpropSec3:Cond1} holds, this suggests that the optimal stopping threshold is $z_*\coloneqq 0$ when Condition~\eqref{eqpropSec3:Cond5} holds, in which case $\check{\tau}_{z_*}=\infty$, since the origin is a natural boundary for $X$. Economically, this captures the intuition that it is never optimal for the short seller to close out voluntarily when the drift rate of the stock price is negative and the discount rate is zero, since the stock price is expected to decline over time and the opportunity cost of not realising an early profit is zero.

Based on the argument above, we obtain the following expression for the candidate value function $\widehat{V}\in\C{0}{0}(0,\infty)\cap\C{2}{0}(0,\kappa+c)$ under Condition~\eqref{eqpropSec3:Cond5}:
\begin{equation}
\label{eqSec4:ValFun5}
\begin{split}
&\widehat{V}(x)\coloneqq J(x,\check{\tau}_0)
=\E_x(\kappa-X_{\hat{\tau}_{\kappa+c}\wedge\rho})
=\kappa-\E_x\bigl(\ind{\{\hat{\tau}_{\kappa+c}\leq\rho\}}X_{\hat{\tau}_{\kappa+c}}\bigr)-\E_x\bigl(\ind{\{\hat{\tau}_{\kappa+c}>\rho\}}X_\rho\bigr)\\
&=\kappa-(\kappa+c)\E_x\bigl(\P_x(\hat{\tau}_{\kappa+c}\leq\rho\,|\,\SigAlg{F}_{\hat{\tau}_{\kappa+c}})\bigr)-\E_x(X_\rho)+\E_x\Bigl(\E_x\bigl(\ind{\{\hat{\tau}_{\kappa+c}\leq\rho\}}X_\rho\,|\,\SigAlg{F}_{\hat{\tau}_{\kappa+c}}\bigr)\Bigr)\\
&=\kappa-(\kappa+c)\E_x\biggl(\int_{\hat{\tau}_{\kappa+c}}^\infty\lambda\e^{-\lambda t}\,\d t\biggr)-\int_0^\infty\lambda\e^{-\lambda t}\E_x(X_t)\,\d t+\E_x\biggl(\int_{\hat{\tau}_{\kappa+c}}^\infty\lambda\e^{-\lambda t}X_t\,\d t\biggr)\\
&=\kappa-(\kappa+c)\E_x\bigl(\e^{-\lambda\hat{\tau}_{\kappa+c}}\bigr)-\int_0^\infty\lambda x\e^{-(\lambda-\mu)t}\,\d t+\E_x\biggl(\int_0^\infty\lambda\e^{-\lambda(\hat{\tau}_{\kappa+c}+s)}X_{\hat{\tau}_{\kappa+c}+s}\,\d s\biggr)\\
&=\kappa-(\kappa+c)\frac{\psi_\lambda(x)}{\psi_\lambda(\kappa+c)}-\frac{\lambda x}{\lambda-\mu}+\E_x\biggl(\e^{-\lambda\hat{\tau}_{\kappa+c}}\int_0^\infty\lambda\e^{-\lambda s}\E_x\bigl(X_{\hat{\tau}_{\kappa+c}+s}\,|\,\SigAlg{F}_{\hat{\tau}_{\kappa+c}}\bigr)\,\d s\biggr)\\
&=\kappa-(\kappa+c)\frac{\psi_\lambda(x)}{\psi_\lambda(\kappa+c)}-\frac{\lambda x}{\lambda-\mu}+\E_x\biggl(\e^{-\lambda\hat{\tau}_{\kappa+c}}\int_0^\infty\lambda\e^{-\lambda s}\E_{X_{\hat{\tau}_{\kappa+c}}}(X_s)\,\d s\biggr)\\
&=\kappa-(\kappa+c)\frac{\psi_\lambda(x)}{\psi_\lambda(\kappa+c)}-\frac{\lambda x}{\lambda-\mu}+\biggl(\int_0^\infty\lambda(\kappa+c)\e^{-(\lambda-\mu)s}\,\d s\biggr)\E_x\bigl(\e^{-\lambda\hat{\tau}_{\kappa+c}}\bigr)\\
&=\kappa-(\kappa+c)\frac{\psi_\lambda(x)}{\psi_\lambda(\kappa+c)}-\frac{\lambda x}{\lambda-\mu}+\frac{\lambda(\kappa+c)}{\lambda-\mu}\frac{\psi_\lambda(x)}{\psi_\lambda(\kappa+c)}\\
&=\kappa-\frac{\lambda x}{\lambda-\mu}+\frac{\mu(\kappa+c)}{\lambda-\mu}\frac{\psi_\lambda(x)}{\psi_\lambda(\kappa+c)},
\end{split}
\end{equation}
for all $x\in(0,\kappa+c)$, and $\widehat{V}(x)=\kappa-x$, for all $x\in[\kappa+c,\infty)$. The third equality above follows from the fact that exponential random variables are almost surely finite, which implies that $\ind{\{\hat{\tau}_{\kappa+c}\leq\rho\}}X_{\hat{\tau}_{\kappa+c}}=\ind{\{\hat{\tau}_{\kappa+c}\leq\rho\}}(\kappa+c)$, for all $x\in(0,\kappa+c)$. Also note that the expressions $X_{\hat{\tau}_{\kappa+c}+s}$, for $s\geq 0$, in the fourth to sixth lines are well-defined, since $\e^{-\lambda\hat{\tau}_{\kappa+c}}=\ind{\{\hat{\tau}_{\kappa+c}<\infty\}}\e^{-\lambda\hat{\tau}_{\kappa+c}}$. Finally, the sixth and ninth equalities employ the Laplace transform identity \eqref{eqSec2:LaplaceTransforms_2}.

The following lemma derives a lower bound for the candidate value function derived above. It is used in the proof of Theorem~\ref{thmSec4:VerThm}(e) to show that the candidate value function and the associated candidate optimal stopping threshold solve Problem~\eqref{eqSec3:OptStopProb} under Condition~\eqref{eqpropSec3:Cond5}.

\begin{lemma}
\label{lemSec4:ValFun5LowBnd}
Suppose Condition~(\ref{eqpropSec3:Cond5}) holds, and let $\widehat{V}\in\C{0}{0}(0,\infty)\cap\C{2}{0}(0,\kappa+c)$ be determined by (\ref{eqSec4:ValFun5}) over $(0,\kappa+c)$ and $\widehat{V}(x)\coloneqq\kappa-x$, for all $x\in[\kappa+c,\infty)$. Then $\widehat{V}(x)>\kappa-x$, for all $x\in(0,\kappa+c)$.
\end{lemma}
\begin{proof}
Define the function $U\in\C{0}{0}(0,\infty)\cap\C{2}{0}(0,\kappa+c)$, by setting $U(x)\coloneqq\widehat{V}(x)-(\kappa-x)$, for all $x\in(0,\infty)$. We shall demonstrate that $U(x)>0$, for all $x\in(0,\kappa+c)$. To begin with, note that $\mu<0$ ensures that $\nu<-\sfrac{1}{2}$. Consequently,
\begin{equation}
\label{eqlemSec4:ValFun5LowBnd}
\sqrt{\nu^2+\frac{2\lambda}{\sigma^2}}-\nu-1>|\nu|-\nu-1>0.
\end{equation}
Hence,
\begin{equation*}
\frac{\d}{\d x}\frac{\psi_\lambda(x)}{x}
=\frac{x\psi_\lambda'(x)-\psi_\lambda(x)}{x^2}
=\biggl(\sqrt{\nu^2+\frac{2\lambda}{\sigma^2}}-\nu-1\biggr)\frac{\psi_\lambda(x)}{x^2}>0,
\end{equation*}
for all $x\in(0,\infty)$, by virtue of \eqref{eqSec2:PhiPrimePsiPrime}. That is to say, the function $(0,\infty)\ni x\mapsto\sfrac{\psi_\lambda(x)}{x}$ is monotonically increasing. Finally, an application of \eqref{eqSec4:ValFun5}  gives
\begin{equation*}
U(x)=\frac{\mu}{\lambda-\mu}\biggl((\kappa+c)\frac{\psi_\lambda(x)}{\psi_\lambda(\kappa+c)}-x\biggr)
=\frac{\mu x}{\lambda-\mu}\biggl(\frac{\psi_\lambda(x)}{x}\frac{\kappa+c}{\psi_\lambda(\kappa+c)}-1\biggr)
>0,
\end{equation*}
for all $x\in(0,\kappa+c)$, since $\mu<0$ and $\sfrac{\psi_\lambda(x)}{x}<\sfrac{\psi_\lambda(\kappa+c)}{(\kappa+c)}$.
\end{proof}

\begin{figure}
\centering
\includegraphics[scale=0.6]{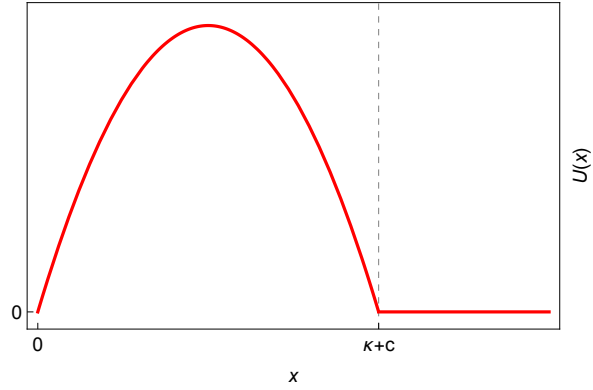}
\caption{The function $U$ defined in Lemma~\ref{lemSec4:ValFun5LowBnd}, under Condition~\eqref{eqpropSec3:Cond5}.}
\label{figSec4:U_4}
\end{figure}

Figure~\ref{figSec4:U_4} plots the function $U$ defined in Lemma~\ref{lemSec4:ValFun5LowBnd}, using parameter values that satisfy Condition~\eqref{eqpropSec3:Cond5}. We observe that $U(x)\geq 0$, for all $x\in(0,\infty)$, which is to say that the candidate value function satisfies $\widehat{V}(x)\geq\kappa-x$, for all $x\in(0,\infty)$. Moreover, as established by Lemma~\ref{lemSec4:ValFun5LowBnd}, we see that $\widehat{V}(x)>\kappa-x$, for all $x\in(0,\kappa+c)$. The kink in $U$ at $\kappa+c$ indicates that the $\widehat{V}$ is not differentiable at that point. This can be seen directly by differentiating \eqref{eqSec4:ValFun5} to get
\begin{equation*}
\begin{split}
\widehat{V}'((\kappa+c)-)&=-\frac{\lambda}{\lambda-\mu}+\frac{\mu(\kappa+c)}{\lambda-\mu}\frac{\psi_\lambda'(\kappa+c)}{\psi_\lambda(\kappa+c)}\\
&=-\frac{\lambda}{\lambda-\mu}+\frac{\mu}{\lambda-\mu}\biggl(\sqrt{\nu^2+\frac{2\lambda}{\sigma^2}}-\nu\biggr)\\
&<-\frac{\lambda}{\lambda-\mu}+\frac{\mu}{\lambda-\mu}
=-1=\widehat{V}'((\kappa+c)+),
\end{split}
\end{equation*}
by virtue of \eqref{eqlemSec4:ValFun5LowBnd} and since $\mu<0$.

According to Proposition~\ref{propSec3:SolCond}(f), $H$ is identically zero if Condition~\eqref{eqpropSec3:Cond6} holds, which intimates that Problem~\ref{eqSec3:OptStopProb} does not admit a uniquely determined optimal stopping threshold in that case. Indeed, it follows from \eqref{eqSec3:ODEPartSol} that $\widehat{v}(x)=\kappa-x$, for all $x\in(0,\infty)$, under Condition~\eqref{eqpropSec3:Cond6}, whence every point in $(0,\kappa+c)$ satisfies the free boundary equation \eqref{eqSec3:FreeBndEqn}. This suggests that any stopping time solves Problem~\ref{eqSec3:OptStopProb}, in which case it is most convenient to choose immediate stopping. Hence, the candidate optimal stopping threshold is $z_*=\kappa+c$, under Condition~\eqref{eqpropSec3:Cond6}, and the candidate value function $\widehat{V}\in\C{2}{0}(0,\infty)$ is defined by $\widehat{V}(x)\coloneqq\kappa-x$, for all $x\in(0,\infty)$. The economic interpretation is that the short seller is indifferent about when to close out his position when the drift rate of the stock price and the discount rate are both zero. In particular, waiting offers no advantage, since the stock price is not expected to decline, but waiting also imposes no penalty, due to the zero discount rate.

Finally, note that Condition~\eqref{eqpropSec3:Cond7} is the limiting case of Condition~\eqref{eqpropSec3:Cond4}, as $r\downarrow 0$. Consequently, since the candidate optimal stopping threshold under Condition~\eqref{eqpropSec3:Cond4} is $z_*\coloneqq\kappa+c$, and the candidate value function $\widehat{V}\in\C{2}{0}(0,\infty)$ is given by $\widehat{V}(x)\coloneqq\kappa-x$, for all $x\in(0,\infty)$, we surmise that the same is true if Condition~\eqref{eqpropSec3:Cond7} holds. The economic interpretation is that immediate close out is optimal if the drift rate of the stock price is positive and the discount rate is zero, since the stock price is expected to appreciate over time.

\subsection{Verifying the candidate solutions}
The next result confirms that the candidate optimal stopping policies and the candidate value functions proposed above do indeed solve Problem~\eqref{eqSec3:OptStopProb}, under their associated parameter regimes.

\begin{theorem}
\label{thmSec4:VerThm}
The optimal stopping time $\tau_*\in\StopTimes$ and the value function $\widetilde{V}\in\C{0}{0}(0,\infty)$ for Problem~(\ref{eqSec3:OptStopProb}) are given by $\tau_*=\check{\tau}_{z_*}$ and $\widetilde{V}=\widehat{V}$, where
\begin{enumerate}[leftmargin=*,topsep=1ex,itemsep=1ex,itemindent=0ex,label=(\alph*)]
\item
$z_*\in(0,\sfrac{r\kappa}{(r-\mu)})\subset(0,\kappa+c)$ is the solution to (\ref{eqSec3:FreeBndEqn}) and $\widehat{V}\in\C{0}{0}(0,\infty)\cap\C{2}{0}(z_*,\kappa+c)$ is determined by (\ref{eqSec3:FreeBndProbSol}) over $(z_*,\kappa+c)$ and $\widehat{V}(x)=\kappa-x$, for all $x\in(0,z_*]\cup[\kappa+c,\infty)$, if Condition~(\ref{eqpropSec3:Cond1}) holds;
\item
$z_*=\kappa+c$ and $\widehat{V}\in\C{2}{0}(0,\infty)$ is determined by $\widehat{V}(x)=\kappa-x$, for all $x\in(0,\infty)$, if Condition~(\ref{eqpropSec3:Cond2}) holds; 
\item
$z_*=\kappa+c$ and $\widehat{V}\in\C{2}{0}(0,\infty)$ is determined by $\widehat{V}(x)=\kappa-x$, for all $x\in(0,\infty)$, if Condition~(\ref{eqpropSec3:Cond3}) holds; 
\item
$z_*=\kappa+c$ and $\widehat{V}\in\C{2}{0}(0,\infty)$ is determined by $\widehat{V}(x)=\kappa-x$, for all $x\in(0,\infty)$, if Condition~(\ref{eqpropSec3:Cond4}) holds; 
\item
$z_*=0$ and $\widehat{V}\in\C{0}{0}(0,\infty)\cap\C{2}{0}(0,\kappa+c)$ is determined by (\ref{eqSec4:ValFun5}) over $(0,\kappa+c)$ and $\widehat{V}(x)=\kappa-x$, for all $x\in[\kappa+c,\infty)$, if Condition~(\ref{eqpropSec3:Cond5}) holds.
\item
$z_*=\kappa+c$ and $\widehat{V}\in\C{2}{0}(0,\infty)$ is determined by $\widehat{V}(x)=\kappa-x$, for all $x\in(0,\infty)$, if Condition~(\ref{eqpropSec3:Cond6}) holds; and 
\item
$z_*=\kappa+c$ and $\widehat{V}\in\C{2}{0}(0,\infty)$ is determined by $\widehat{V}(x)=\kappa-x$, for all $x\in(0,\infty)$, if Condition~(\ref{eqpropSec3:Cond7}) holds. 
\end{enumerate}
\end{theorem}
\begin{proof}
(a):~Suppose Condition~(\ref{eqpropSec3:Cond1}) holds, in which case $z_*\in(0,\kappa+c)$ is the unique solution to (\ref{eqSec3:FreeBndEqn}) and $\widehat{V}\in\C{0}{0}(0,\infty)\cap\C{2}{0}(z_*,\kappa+c)$ is determined by (\ref{eqSec3:FreeBndProbSol}) over $(z_*,\kappa+c)$ and $\widehat{V}(x)=\kappa-x$, for all $x\in(0,z_*]\cup[\kappa+c,\infty)$. Note that $\widehat{V}\in\C{1}{0}(0,\kappa+c]$, since \eqref{eqSec3:FreeBndProb_c} ensures that $\widehat{V}'$ is continuous at $z_*$ and $|\widehat{V}'((\kappa+c)-)|<\infty$, by inspection of \eqref{eqSec3:FreeBndProbSol}. On the other hand, $\widehat{V}\notin\C{2}{0}(0,\kappa+c]$, since $\widehat{V}''(z_*+)>0=\widehat{V}''(z_*-)$, as remarked in the discussion following Lemma~\ref{lemSec4:ValFun1LowBnd}. Consequently, the standard It\^o formula cannot be applied to the process
\begin{equation*}
\R_+\ni t\mapsto\e^{-(\lambda+r)(t\wedge\hat{\tau}_{\kappa+c})}\widehat{V}(X_{t\wedge\hat{\tau}_{\kappa+c}}).
\end{equation*}
However, $\widehat{V}\in\C{2}{0}(0,z_*]\cap\C{2}{0}[z_*,\kappa+c]$, since $|\widehat{V}''(z_*+)|<\infty$ and $|\widehat{V}''((\kappa+c)-)|<\infty$, by inspection of \eqref{eqSec3:FreeBndProbSol}. Hence, we may apply the local time-space formula of \citet{Pes05a} to the above-mentioned process, to get
\begin{equation*}
\begin{split}
\e^{-(\lambda+r)(t\wedge\hat{\tau}_{\kappa+c})}&\widehat{V}(X_{t\wedge\hat{\tau}_{\kappa+c}})=\widehat{V}(X_0)-\int_0^{t\wedge\hat{\tau}_{\kappa+c}}(\lambda+r)\e^{-(\lambda+r)s}\widehat{V}(X_s)\,\d s\\
&+\int_0^{t\wedge\hat{\tau}_{\kappa+c}}\e^{-(\lambda+r)s}\widehat{V}'(X_s)\,\d X_s+\int_0^{t\wedge\hat{\tau}_{\kappa+c}}\ind{\{X_s\neq z_*\}}\e^{-(\lambda+r)s}\frac{1}{2}\widehat{V}''(X_s)\,\d\<X\>_s\\
&+\int_0^{t\wedge\hat{\tau}_{\kappa+c}}\ind{\{X_s= z_*\}}\e^{-(\lambda+r)s}\frac{1}{2}\bigl(\widehat{V}'(X_s+)-\widehat{V}'(X_s-)\bigr)\,\d\ell^{z_*}_s(X),
\end{split}
\end{equation*}
for all $t\geq 0$, where the local time process $\ell^{z_*}(X)$ quantifies the time $X$ spends in the immediate vicinity of $z_*$ \citep[see][Section~3.5]{PS06}. Note that the final term above is zero, since \eqref{eqSec3:FreeBndProb_c} ensures that $\widehat{V}'(z_*+)=\widehat{V}'(z_*-)$, while the identities
\begin{align*}
\widehat{V}(X_s)&=\ind{\{X_s\leq z_*\}}(\kappa-X_s)+\ind{\{X_s>z_*\}}\widehat{V}(X_s),\\
\widehat{V}'(X_s)&=-\ind{\{X_s\leq z_*\}}+\ind{\{X_s>z_*\}}\widehat{V}'(X_s),\\
\intertext{and}
\ind{\{X_s\neq z_*\}}\widehat{V}''(X_s)&=\ind{\{X_s>z_*\}}\widehat{V}''(X_s),
\end{align*}
for all $s\geq 0$, follow from the definition of $\widehat{V}$. We therefore obtain
\begin{equation*}
\begin{split}
&\e^{-(\lambda+r)(t\wedge\hat{\tau}_{\kappa+c})}\widehat{V}(X_{t\wedge\hat{\tau}_{\kappa+c}})\\
&=\widehat{V}(X_0)+\int_0^{t\wedge\hat{\tau}_{\kappa+c}}\ind{\{X_s\leq z_*\}}\e^{-(\lambda+r)s}\bigl(-\mu X_s-(\lambda+r)(\kappa-X_s)\bigr)\,\d s\\
&\hspace{2cm}+\int_0^{t\wedge\hat{\tau}_{\kappa+c}}\ind{\{X_s>z_*\}}\e^{-(\lambda+r)s}\biggl(\mu X_s\widehat{V}'(X_s)+\frac{1}{2}\sigma^2X_s^2\widehat{V}''(X_s)-(\lambda+r)\widehat{V}(X_s)\biggr)\,\d s\\
&\hspace{2cm}+\int_0^{t\wedge\hat{\tau}_{\kappa+c}}\e^{-(\lambda+r)s}\sigma X_s\widehat{V}'(X_s)\,\d B_s\\
&=\widehat{V}(X_0)+\int_0^{t\wedge\hat{\tau}_{\kappa+c}}\ind{\{X_s\leq z_*\}}\e^{-(\lambda+r)s}\bigl((r-\mu)X_s-r\kappa-\lambda(\kappa-X_s)\bigr)\,\d s\\
&\hspace{2cm}+\int_0^{t\wedge\hat{\tau}_{\kappa+c}}\ind{\{X_s>z_*\}}\e^{-(\lambda+r)s}\bigl(\L_X\widehat{V}(X_s)-(\lambda+r)\widehat{V}(X_s)\bigr)\,\d s\\
&\hspace{2cm}+\int_0^{t\wedge\hat{\tau}_{\kappa+c}}\e^{-(\lambda+r)s}\sigma X_s\widehat{V}'(X_s)\,\d B_s\\
&\leq\widehat{V}(X_0)-\int_0^{t\wedge\hat{\tau}_{\kappa+c}}\lambda\e^{-(\lambda+r)s}(\kappa-X_s)\,\d s+\int_0^{t\wedge\hat{\tau}_{\kappa+c}}\e^{-(\lambda+r)s}\sigma X_s\widehat{V}'(X_s)\,\d B_s,
\end{split}
\end{equation*}
for all $t\geq 0$, where the inequality follows from
\begin{equation*}
\ind{\{X_s\leq z_*\}}\bigl((r-\mu)X_s-r\kappa-\lambda(\kappa-X_s)\bigr)\leq-\ind{\{X_s\leq z_*\}}\lambda(\kappa-X_s),
\end{equation*}
for all $s\geq 0$, since Proposition~\ref{propSec3:SolCond}(a) established that $z_*<\sfrac{r\kappa}{(r-\mu)}$, together with
\begin{equation*}
\ind{\{X_s>z_*\}}\bigl(\L_X\widehat{V}(X_s)-(\lambda+r)\widehat{V}(X_s)\bigr)=-\ind{\{X_s>z_*\}}\lambda(\kappa-X_s),
\end{equation*}
for all $s\geq 0$, since $\widehat{V}$ satisfies \eqref{eqSec3:FreeBndProb_a} over $(z_*,\kappa+c)$. Next, observe that $\widehat{V}'$ is bounded over $(0,\kappa+c]$, since $\widehat{V}\in\C{1}{0}(0,\kappa+c]$ and $\widehat{V}'(0+)=-1$. This is sufficient to ensure that the local martingale
\begin{equation*}
\R_+\ni t\mapsto\int_0^{t\wedge\hat{\tau}_{\kappa+c}}\e^{-(\lambda+r)s}\sigma X_s\widehat{V}'(X_s)\,\d B_s
\end{equation*}
is in fact a uniformly integrable martingale. Given an arbitrary stopping time $\tau\in\StopTimes$, an application of the optional sampling theorem then yields
\begin{equation*}
\begin{split}
\widehat{V}(x)&\geq\E_x\biggl(\int_0^{\tau\wedge\hat{\tau}_{\kappa+c}}\lambda\e^{-(\lambda+r)s}(\kappa-X_s)\,\d s+\e^{-(\lambda+r)(\tau\wedge\hat{\tau}_{\kappa+c})}\widehat{V}\bigl(X_{\tau\wedge\hat{\tau}_{\kappa+c}}\bigr)\biggr)\\
&\geq\E_x\biggl(\int_0^{\tau\wedge\hat{\tau}_{\kappa+c}}\lambda\e^{-(\lambda+r)t}(\kappa-X_t)\,\d t+\e^{-(\lambda+r)(\tau\wedge\hat{\tau}_{\kappa+c})}\bigl(\kappa-X_{\tau\wedge\hat{\tau}_{\kappa+c}}\bigr)\biggr)
=J(x,\tau),
\end{split}
\end{equation*}
for all $x\in(0,\infty)$, since Lemma~\ref{lemSec4:ValFun1LowBnd} established that $\widehat{V}(x)\geq\kappa-x$. This implies that $\widehat{V}(x)\geq\widetilde{V}(x)$, for all $x\in(0,\infty)$. On the other hand, the function $(0,\infty)\ni x\mapsto J(x,\check{\tau}_{z_*})$ is the unique solution to the Dirichlet problem with data \eqref{eqSec3:FreeBndProb_a} and \eqref{eqSec3:FreeBndProb_b}, due to the stochastic representation theorem for the solutions to Dirichlet problems \citep[for a precise formulation of the relevant result applicable to our setting, see][Theorem~1]{VAW05}. Consequently, $\widehat{V}(x)=J(x,\check{\tau}_{z_*})\leq\widetilde{V}(x)$, for all $x\in(0,\infty)$, since $\widehat{V}$ satisfies \eqref{eqSec3:FreeBndProb_a} and \eqref{eqSec3:FreeBndProb_b} by construction.
\vspace{2mm}\newline\noindent
(b):~Suppose Condition~(\ref{eqpropSec3:Cond2}) holds, in which case $z_*=\kappa+c$ and $\widehat{V}\in\C{2}{0}(0,\infty)$ is determined by $\widehat{V}(x)=\kappa-x$, for all $x\in(0,\infty)$. It\^o's formula then gives
\begin{equation}
\label{eqthmSec4:VerThm}
\begin{split}
\e^{-(\lambda+r)(t\wedge\hat{\tau}_{\kappa+c})}&(\kappa-X_{t\wedge\hat{\tau}_{\kappa+c}})=\e^{-(\lambda+r)(t\wedge\hat{\tau}_{\kappa+c})}\widehat{V}(X_{t\wedge\hat{\tau}_{\kappa+c}})\\
&=\widehat{V}(X_0)-\int_0^{t\wedge\hat{\tau}_{\kappa+c}}(\lambda+r)\e^{-(\lambda+r)s}\widehat{V}(X_s)\,\d s\\
&\hspace{1cm}+\int_0^{t\wedge\hat{\tau}_{\kappa+c}}\e^{-(\lambda+r)s}\widehat{V}'(X_s)\,\d X_s+\int_0^{t\wedge\hat{\tau}_{\kappa+c}}\e^{-(\lambda+r)s}\frac{1}{2}\widehat{V}''(X_s)\,\d\<X\>_s\\
&=\widehat{V}(X_0)-\int_0^{t\wedge\hat{\tau}_{\kappa+c}}\e^{-(\lambda+r)s}\bigl(\L_X\widehat{V}(X_s)-(\lambda+r)\widehat{V}(X_s)\bigr)\,\d s\\
&\hspace{7cm}-\int_0^{t\wedge\hat{\tau}_{\kappa+c}}\e^{-(\lambda+r)s}\sigma X_s\,\d B_s\\
&\leq\widehat{V}(X_0)-\int_0^{t\wedge\hat{\tau}_{\kappa+c}}\lambda\e^{-(\lambda+r)s}(\kappa-X_s)\,\d s-\int_0^{t\wedge\hat{\tau}_{\kappa+c}}\e^{-(\lambda+r)s}\sigma X_s\,\d B_s,
\end{split}
\end{equation}
for all $t\geq 0$. To justify the inequality above, note that
\begin{equation*}
0<r-\mu=r-\frac{rc}{\kappa+c}=\frac{r\kappa}{\kappa+c},
\end{equation*}
since $0<\sfrac{rc}{(\kappa+c)}=\mu<r$, by assumption. Consequently,
\begin{equation*}
\begin{split}
\L_X\widehat{V}(x)-(\lambda+r)\widehat{V}(x)=-\mu x-(\lambda+r)(\kappa-x)
&=(r-\mu)x-r\kappa-\lambda(\kappa-x)\\
&\leq(r-\mu)(\kappa+c)-r\kappa-\lambda(\kappa-x)\\
&=-\lambda(\kappa-x),
\end{split}
\end{equation*}
for all $x\in(0,\kappa+c]$. Note that the process
\begin{equation*}
\R_+\ni t\mapsto\int_0^{t\wedge\hat{\tau}_{\kappa+c}}\e^{-(\lambda+r)s}\sigma X_s\,\d B_s
\end{equation*}
is a uniformly integrable martingale. Given an arbitrary stopping time $\tau\in\StopTimes$, an application of the optional stopping theorem then gives
\begin{equation*}
\widehat{V}(x)\geq\E_x\biggl(\int_0^{\tau\wedge\hat{\tau}_{\kappa+c}}\lambda\e^{-(\lambda+r)s}(\kappa-X_s)\,\d s+\e^{-(\lambda+r)(\tau\wedge\hat{\tau}_{\kappa+c})}(\kappa-X_{\tau\wedge\hat{\tau}_{\kappa+c}})\biggr)
=J(x,\tau),
\end{equation*}
for all $x\in(0,\infty)$. This implies that $\widehat{V}(x)\geq\widetilde{V}(x)$, for all $x\in(0,\infty)$. On the other hand, $\widehat{V}(x)=J(x,\hat{\tau}_{\kappa+c})\leq\widetilde{V}(x)$, for all $x\in(0,\infty)$.
\vspace{2mm}\newline\noindent
(c):~Suppose Condition~(\ref{eqpropSec3:Cond3}) holds, in which case $z_*=\kappa+c$ and $\widehat{V}\in\C{2}{0}(0,\infty)$ is determined by $\widehat{V}(x)=\kappa-x$, for all $x\in(0,\infty)$. The proof is identical to that of Part~(b), except that in this case the inequality in \eqref{eqthmSec4:VerThm} follows from
\begin{equation*}
\begin{split}
\L_X\widehat{V}(x)-(\lambda+r)\widehat{V}(x)=-\mu x-(\lambda+r)(\kappa-x)
&=(r-\mu)x-r\kappa-\lambda(\kappa-x)\\
&\leq(r-\mu)(\kappa+c)-r\kappa-\lambda(\kappa-x)\\
&<-\lambda(\kappa-x),
\end{split}
\end{equation*}
for all $x\in(0,\kappa+c]$, since
\begin{equation*}
0<r-\mu<r-\frac{rc}{\kappa+c}=\frac{r\kappa}{\kappa+c},
\end{equation*}
by virtue of the assumption that $0<\sfrac{rc}{(\kappa+c)}<\mu<r$.
\vspace{2mm}\newline\noindent
(d):~Suppose Condition~(\ref{eqpropSec3:Cond4}) holds, in which case $z_*=\kappa+c$ and $\widehat{V}\in\C{2}{0}(0,\infty)$ is determined by $\widehat{V}(x)=\kappa-x$, for all $x\in(0,\infty)$. The proof is identical to that of Part~(b), except that in this case the inequality in \eqref{eqthmSec4:VerThm} follows from
\begin{equation*}
\begin{split}
\L_X\widehat{V}(x)-(\lambda+r)\widehat{V}(x)=-\mu x-(\lambda+r)(\kappa-x)
&=(r-\mu)x-r\kappa-\lambda(\kappa-x)\\
&\leq-r\kappa-\lambda(\kappa-x)<-\lambda(\kappa-x),
\end{split}
\end{equation*}
for all $x\in(0,\infty)$, by virtue of the assumption that $\mu\geq r$.
\vspace{2mm}\newline\noindent
(e):~Suppose Condition~(\ref{eqpropSec3:Cond5}) holds, in which case $z_*=0$ and $\widehat{V}\in\C{0}{0}(0,\infty)\cap\C{2}{0}(0,\kappa+c)$ is determined by \eqref{eqSec4:ValFun5} over $(0,\kappa+c)$ and $\widehat{V}(x)=\kappa-x$, for all $x\in[\kappa+c,\infty)$. Note that $\widehat{V}\in\C{2}{0}(0,\kappa+c]$ by inspection of \eqref{eqSec4:ValFun5}, since $\psi_\lambda\in\C{2}{0}(0,\kappa+c]$. We may therefore apply It\^o's formula to get
\begin{equation*}
\begin{split}
\e^{-\lambda(t\wedge\hat{\tau}_{\kappa+c})}&\widehat{V}(X_{t\wedge\hat{\tau}_{\kappa+c}})=\widehat{V}(X_0)-\int_0^{t\wedge\hat{\tau}_{\kappa+c}}\lambda\e^{-\lambda s}\widehat{V}(X_s)\,\d s+\int_0^{t\wedge\hat{\tau}_{\kappa+c}}\e^{-\lambda s}\widehat{V}'(X_s)\,\d X_s\\
&\hspace{8cm}+\int_0^{t\wedge\hat{\tau}_{\kappa+c}}\e^{-\lambda s}\frac{1}{2}\widehat{V}''(X_s)\,\d\<X\>_s\\
&=\widehat{V}(X_0)+\int_0^{t\wedge\hat{\tau}_{\kappa+c}}\e^{-\lambda s}\bigl(\L_X\widehat{V}(X_s)-\lambda\widehat{V}(X_s)\bigr)\,\d s+\int_0^{t\wedge\hat{\tau}_{\kappa+c}}\e^{-\lambda s}\sigma X_s\widehat{V}'(X_s)\d B_s\\
&=\widehat{V}(X_0)-\int_0^{t\wedge\hat{\tau}_{\kappa+c}}\lambda\e^{-\lambda s}(\kappa-X_s)\,\d s+\int_0^{t\wedge\hat{\tau}_{\kappa+c}}\e^{-\lambda s}\sigma X_s\widehat{V}'(X_s)\d B_s,
\end{split}
\end{equation*}
for all $t\geq 0$, where the final equality follows from
\begin{equation*}
\begin{split}
\L_X\widehat{V}(x)&=\L_X\biggl(\kappa-\frac{\lambda x}{\lambda-\mu}\biggr)+\frac{\mu(\kappa+c)}{\lambda-\mu}\frac{\L_X\psi_\lambda(x)}{\psi_\lambda(\kappa+c)}\\
&=-\frac{\lambda\mu x}{\lambda-\mu}+\frac{\mu(\kappa+c)}{\lambda-\mu}\frac{\lambda\psi_\lambda(x)}{\psi_\lambda(\kappa+c)}\\
&=\lambda\biggl(\kappa-\frac{\lambda x}{\lambda-\mu}+\frac{\mu(\kappa+c)}{\lambda-\mu}\frac{\psi_\lambda(x)}{\psi_\lambda(\kappa+c)}\biggr)-\lambda\biggl(\kappa-\frac{\lambda x}{\lambda-\mu}+\frac{\mu x}{\lambda-\mu}\biggr)\\
&=\lambda\widehat{V}(x)-\lambda(\kappa-x),
\end{split}
\end{equation*}
for all $x\in(0,\kappa+c)$, by virtue of \eqref{eqSec4:ValFun5} and the fact that $\psi_\lambda$ satisfies \eqref{eqSec2:ODE}. Next, observe that \eqref{eqSec4:ValFun5} gives $\widehat{V}'(0+)=-\sfrac{\lambda}{(\lambda-\mu)}$, since $\psi_\lambda'(0+)=0$ follows from \eqref{eqSec2:PhiPsi} and \eqref{eqSec2:PhiPrimePsiPrime}. Together with the previously established fact that $\widehat{V}\in\C{2}{0}(0,\kappa+c]$, this ensures that $\widehat{V}'$ is bounded on $(0,\kappa+c]$, from which it follows that the local martingale
\begin{equation*}
\R_+\ni t\mapsto\int_0^{t\wedge\hat{\tau}_{\kappa+c}}\e^{-\lambda s}\sigma X_s\widehat{V}'(X_s)\d B_s
\end{equation*}
is in fact a uniformly integrable martingale. Given an arbitrary stopping time $\tau\in\StopTimes$, an application of the optional sampling theorem then yields
\begin{equation*}
\begin{split}
\widehat{V}(x)&=\E_x\biggl(\int_0^{\tau\wedge\hat{\tau}_{\kappa+c}}\lambda\e^{-\lambda s}(\kappa-X_s)\,\d s+\e^{-\lambda(\tau\wedge\hat{\tau}_{\kappa+c})}\widehat{V}\bigl(X_{\tau\wedge\hat{\tau}_{\kappa+c}}\bigr)\biggr)\\
&\geq\E_x\biggl(\int_0^{\tau\wedge\hat{\tau}_{\kappa+c}}\lambda\e^{-\lambda s}(\kappa-X_s)\,\d s+\e^{-\lambda(\tau\wedge\hat{\tau}_{\kappa+c})}\bigl(\kappa-X_{\tau\wedge\hat{\tau}_{\kappa+c}}\bigr)\biggr)
=J(x,\tau),
\end{split}
\end{equation*}
for all $x\in(0,\infty)$, since Lemma~\ref{lemSec4:ValFun5LowBnd} established that $\widehat{V}(x)\geq\kappa-x$. This implies that $\widehat{V}(x)\geq\widetilde{V}(x)$, for all $x\in(0,\infty)$. On the other hand, $\widehat{V}(x)\coloneqq J(x,\check{\tau}_0)\leq\widetilde{V}(x)$, for all $x\in(0,\infty)$.
\vspace{2mm}\newline\noindent
(f):~Suppose Condition~(\ref{eqpropSec3:Cond6}) holds, in which case $z_*=\kappa+c$ and $\widehat{V}\in\C{2}{0}(0,\infty)$ is determined by $\widehat{V}(x)=\kappa-x$, for all $x\in(0,\infty)$. It\^o's formula then gives
\begin{equation*}
\begin{split}
\e^{-\lambda(t\wedge\hat{\tau}_{\kappa+c})}&(\kappa-X_{t\wedge\hat{\tau}_{\kappa+c}})=\e^{-\lambda(t\wedge\hat{\tau}_{\kappa+c})}\widehat{V}(X_{t\wedge\hat{\tau}_{\kappa+c}})\\
&=\widehat{V}(X_0)-\int_0^{t\wedge\hat{\tau}_{\kappa+c}}\lambda\e^{-\lambda s}\widehat{V}(X_s)\,\d s+\int_0^{t\wedge\hat{\tau}_{\kappa+c}}\e^{-\lambda s}\widehat{V}'(X_s)\,\d X_s\\
&\hspace{8cm}+\int_0^{t\wedge\hat{\tau}_{\kappa+c}}\e^{-\lambda s}\frac{1}{2}\widehat{V}''(X_s)\,\d\<X\>_s\\
&=\widehat{V}(X_0)-\int_0^{t\wedge\hat{\tau}_{\kappa+c}}\e^{-\lambda s}\bigl(\L_X\widehat{V}(X_s)-\lambda\widehat{V}(X_s)\bigr)\,\d s-\int_0^{t\wedge\hat{\tau}_{\kappa+c}}\e^{-\lambda s}\sigma X_s\,\d B_s\\
&=\widehat{V}(X_0)-\int_0^{t\wedge\hat{\tau}_{\kappa+c}}\lambda\e^{-\lambda s}(\kappa-X_s)\,\d s-\int_0^{t\wedge\hat{\tau}_{\kappa+c}}\e^{-\lambda s}\sigma X_s\,\d B_s,
\end{split}
\end{equation*}
for all $t\geq 0$, since $\L_X\widehat{V}(x)=-\mu x=0$, for all $x\in(0,\infty)$, by virtue of the assumption that $\mu=0$. Note that the local martingale
\begin{equation*}
\R_+\ni t\mapsto\int_0^{t\wedge\hat{\tau}_{\kappa+c}}\e^{-\lambda s}\sigma X_s\,\d B_s
\end{equation*}
is in fact a uniformly integrable martingale. Given an arbitrary stopping time $\tau\in\StopTimes$, an application of the optional sampling theorem then yields
\begin{equation*}
\widehat{V}(x)=\E_x\biggl(\int_0^{t\wedge\hat{\tau}_{\kappa+c}}\lambda\e^{-\lambda s}(\kappa-X_s)\,\d s+\e^{-\lambda(t\wedge\hat{\tau}_{\kappa+c})}(\kappa-X_{t\wedge\hat{\tau}_{\kappa+c}})\biggr)
=J(x,\tau),
\end{equation*}
for all $x\in(0,\infty)$. This ensures that $\widehat{V}(x)=\widetilde{V}(x)$, for all $x\in(0,\infty)$.
\vspace{2mm}\newline\noindent
(g):~Suppose Condition~(\ref{eqpropSec3:Cond7}) holds, in which case $z_*=\kappa+c$ and $\widehat{V}\in\C{2}{0}(0,\infty)$ is determined by $\widehat{V}(x)=\kappa-x$, for all $x\in(0,\infty)$. It\^o's formula then gives
\begin{equation*}
\begin{split}
\e^{-\lambda(t\wedge\hat{\tau}_{\kappa+c})}&(\kappa-X_{t\wedge\hat{\tau}_{\kappa+c}})=\e^{-\lambda(t\wedge\hat{\tau}_{\kappa+c})}\widehat{V}(X_{t\wedge\hat{\tau}_{\kappa+c}})\\
&=\widehat{V}(X_0)-\int_0^{t\wedge\hat{\tau}_{\kappa+c}}\lambda\e^{-\lambda s}\widehat{V}(X_s)\,\d s+\int_0^{t\wedge\hat{\tau}_{\kappa+c}}\e^{-\lambda s}\widehat{V}'(X_s)\,\d X_s\\
&\hspace{8cm}+\int_0^{t\wedge\hat{\tau}_{\kappa+c}}\e^{-\lambda s}\frac{1}{2}\widehat{V}''(X_s)\,\d\<X\>_s\\
&=\widehat{V}(X_0)-\int_0^{t\wedge\hat{\tau}_{\kappa+c}}\e^{-\lambda s}\bigl(\L_X\widehat{V}(X_s)-\lambda\widehat{V}(X_s)\bigr)\,\d s-\int_0^{t\wedge\hat{\tau}_{\kappa+c}}\e^{-\lambda s}\sigma X_s\,\d B_s\\
&\leq\widehat{V}(X_0)-\int_0^{t\wedge\hat{\tau}_{\kappa+c}}\lambda\e^{-\lambda s}(\kappa-X_s)\,\d s-\int_0^{t\wedge\hat{\tau}_{\kappa+c}}\e^{-\lambda s}\sigma X_s\,\d B_s,
\end{split}
\end{equation*}
for all $t\geq 0$, since $\L_X\widehat{V}(x)=-\mu x<0$, for all $x\in(0,\infty)$, by virtue of the assumption that $\mu>0$. Note that the local martingale
\begin{equation*}
\R_+\ni t\mapsto\int_0^{t\wedge\hat{\tau}_{\kappa+c}}\e^{-\lambda s}\sigma X_s\,\d B_s
\end{equation*}
is in fact a uniformly integrable martingale. Given an arbitrary stopping time $\tau\in\StopTimes$, an application of the optional sampling theorem then yields
\begin{equation*}
\widehat{V}(x)\leq\E_x\biggl(\int_0^{t\wedge\hat{\tau}_{\kappa+c}}\lambda\e^{-\lambda s}(\kappa-X_s)\,\d s+\e^{-\lambda(t\wedge\hat{\tau}_{\kappa+c})}(\kappa-X_{t\wedge\hat{\tau}_{\kappa+c}})\biggr)
=J(x,\tau),
\end{equation*}
for all $x\in(0,\infty)$. This implies that $\widehat{V}(x)\geq\widetilde{V}(x)$, for all $x\in(0,\infty)$. On the other hand, $\widehat{V}(x)\coloneqq J(x,\check{\tau}_{\kappa+c})\leq\widetilde{V}(x)$, for all $x\in(0,\infty)$.
\end{proof}
\section{The Solution to the Short Selling Problem}
\label{Sec5}
\subsection{The case when the discount rate is non-zero}
As observed in Section~\ref{Sec3}, the optimal close-out time $\tau_*\in\StopTimes$ and the value function $V\in\C{0}{0}(0,\infty)$ for the original short selling problem \eqref{eqSec2:ShortSellProb} are determined by the solution to Problem~\eqref{eqSec3:OptStopProb}, with $\kappa$ set equal to the the initial stock price. The solution to that problem is presented in Theorem~\ref{thmSec4:VerThm}.

We begin by considering the case when $r>0$. Given $x\in(0,\infty)$, suppose \eqref{eqSec3:FreeBndEqn}, with $\kappa\coloneqq x$, possesses a solution in the interval $(0,x)$. That is to say, the free-boundary equation
\begin{equation}
\label{eqSec5:FreeBndEqn}
\begin{split}
&\bigl(x-z_*-\widehat{v}(z_*)|_{\kappa=x}\bigr)\frac{\phi_{\lambda+r}(x+c)\psi_{\lambda+r}'(z_*)-\phi_{\lambda+r}'(z_*)\psi_{\lambda+r}(x+c)}{\phi_{\lambda+r}(x+c)\psi_{\lambda+r}(z_*)-\phi_{\lambda+r}(z_*)\psi_{\lambda+r}(x+c)}\\
&\hspace{2cm}+\bigl(c+\widehat{v}(x+c)|_{\kappa=x}\bigr)\frac{\phi_{\lambda+r}(z_*)\psi_{\lambda+r}'(z_*)-\phi_{\lambda+r}'(z_*)\psi_{\lambda+r}(z_*)}{\phi_{\lambda+r}(x+c)\psi_{\lambda+r}(z_*)-\phi_{\lambda+r}(z_*)\psi_{\lambda+r}(x+c)}\\
&\hspace{2cm}+\widehat{v}\,'(z_*)|_{\kappa=x}=-1.
\end{split}
\end{equation}
admits a solution $z_*\in(0,x)$. The optimal close-out time is then $\tau_*=\check{\tau}_{z_*}$ and the value function satisfies \eqref{eqSec3:FreeBndProbSol}, with $\kappa\coloneqq x$. In other words,
\begin{equation}
\label{eqSec5:ValFunc1}
\begin{split}
V(x)&=\widehat{V}(x)|_{\kappa=x}\\
&=\bigl(x-z_*-\widehat{v}(z_*)|_{\kappa=x}\bigr)\frac{\phi_{\lambda+r}(x+c)\psi_{\lambda+r}(x)-\phi_{\lambda+r}(x)\psi_{\lambda+r}(x+c)}{\phi_{\lambda+r}(x+c)\psi_{\lambda+r}(z_*)-\phi_{\lambda+r}(z_*)\psi_{\lambda+r}(x+c)}\\
&\hspace{2cm}+\bigl(\widehat{v}(x+c)|_{\kappa=x}+c\bigr)\frac{\phi_{\lambda+r}(z_*)\psi_{\lambda+r}(x)-\phi_{\lambda+r}(x)\psi_{\lambda+r}(z_*)}{\phi_{\lambda+r}(x+c)\psi_{\lambda+r}(z_*)-\phi_{\lambda+r}(z_*)\psi_{\lambda+r}(x+c)}\\
&\hspace{2cm}+\widehat{v}(x)|_{\kappa=x}.
\end{split}
\end{equation}
On the other hand, if \eqref{eqSec5:FreeBndEqn} does not possess a solution in the interval $(0,x)$, the optimal repurchase price is $z_*\coloneqq x$, which implies that the optimal close-out time is $\tau_*=\check{\tau}_{z_*}=0$ and the value function satisfies $V(x)=0$.

\subsection{The case when the discount rate is zero}
Next, we consider the situation when $r=0$. If $\mu<0$, the optimal repurchase price is $z_*\coloneqq 0$, which implies that the optimal close-out time is $\tau_*=\check{\tau}_{z_*}=\infty$ and the value function is obtained from \eqref{eqSec4:ValFun5}, as follows:
\begin{equation}
\label{eqSec5:ValFunc2}
V(x)=\widehat{V}(x)|_{\kappa=x}
=\frac{\mu(x+c)}{\lambda-\mu}\frac{\psi_\lambda(x)}{\psi_\lambda(x+c)}-\frac{\mu x}{\lambda-\mu},
\end{equation}
for all $x\in(0,\infty)$. On the other hand, if $\mu\geq 0$, the optimal repurchase price is $z_*\coloneqq x$, for all $x\in(0,\infty)$, in which case the optimal close-out time is $\tau_*=\check{\tau}_{z_*}=0$ and the value function satisfies $V(x)=0$.

\subsection{The unconstrained problem when the discount rate is non-zero}
To assess the impact of the collateral constraint and the possibility of recall on Problem~\eqref{eqSec2:ShortSellProb}, it is useful to compare the optimal close-out time and value function for that problem with the corresponding data for the unconstrained short selling problem. That problem is obtained by letting $c\uparrow\infty$ and $\lambda\downarrow 0$, so that $\zeta=\infty$ and $\rho=\infty$ $\P_x$-a.s., for all $x\in(0,\infty)$. Problem~\eqref{eqSec2:ShortSellProb} then reduces to
\begin{equation}
\label{eqSec5:ShortSellProb}
V(x)\coloneqq\sup_{\tau\in\StopTimes}\E_x\bigl(\e^{-r\tau}(x-X_\tau)\bigr),	
\end{equation}
for all $x\in(0,\infty)$.

We begin by considering the situation when $r>0$. Letting $c\uparrow\infty$ and $\lambda\downarrow 0$ in \eqref{eqSec5:FreeBndEqn} gives
\begin{equation*}
(x-z_*)\frac{\phi_r'(z_*)}{\phi(z_*)}=1,
\end{equation*}
for all $x\in(0,\infty)$, by virtue of the limits
\begin{equation}
\label{eqSec5:Limits}
\lim_{\lambda\downarrow 0}\hat{v}(x)|_{\kappa=x}=0,
\qquad
\lim_{c\uparrow\infty}\lim_{\lambda\downarrow 0}\phi_{\lambda+r}(x+c)=0
\qquad\text{and}\qquad
\lim_{c\uparrow\infty}\lim_{\lambda\downarrow 0}\frac{c}{\psi_{\lambda+r}(x+c)}=0,
\end{equation}
which follow by inspection of \eqref{eqSec3:ODEPartSol} and \eqref{eqSec2:PhiPsi}. Using those limits again, the previous identity provides the following explicit formula for the optimal close-out price for the unconstrained short selling problem:
\begin{equation}
\label{eqSec5:NoConsFreeBnd}
z_*=\frac{\sqrt{\nu^2+\frac{2r}{\sigma^2}}+\nu}{1+\sqrt{\nu^2+\frac{2r}{\sigma^2}}+\nu}x,
\end{equation}
for all $x\in(0,\infty)$. Note that $z_*\in(0,x)$, for all $x\in(0,\infty)$, which implies the optimal close-out time $\tau_*=\check{\tau}_{z_*}$ is both strictly positive and finite. In other words, immediate close-out and waiting forever are both suboptimal for the unconstrained short selling problem, if the discount rate is non-zero. Finally, letting $c\uparrow\infty$ and $\lambda\downarrow 0$ in \eqref{eqSec5:ValFunc1} gives the following expression for the value function for the unconstrained short selling problem:
\begin{equation}
\label{eqSec5:NoConsValFunc1}
V(x)=(x-z_*)\frac{\phi_r(x)}{\phi_r(z_*)},
\end{equation}
for all $x\in(0,\infty)$, by virtue of the limits in \eqref{eqSec5:Limits}, where $z_*$ is determined by \eqref{eqSec5:NoConsFreeBnd}.

\subsection{The unconstrained problem when the discount rate is zero}
To analyse the unconstrained short selling problem \eqref{eqSec5:ShortSellProb} in the case when $r=0$, we first recall the following facts about the geometric Brownian motion determined by \eqref{eqSec2:SDE}:
\begin{align*}
\text{if}\qquad\mu<\frac{1}{2}\sigma^2\qquad\text{then}&\qquad X_{\infty-}=0\qquad\text{and}\qquad\sup_{t\geq 0}X_t<\infty,\\
\text{if}\qquad\mu=\frac{1}{2}\sigma^2\qquad\text{then}&\qquad\inf_{t\geq 0}X_t=0\qquad\text{and}\qquad\sup_{t\geq 0}X_t=\infty,\\
\intertext{and}
\text{if}\qquad\mu>\frac{1}{2}\sigma^2\qquad\text{then}&\qquad X_{\infty-}=\infty\qquad\text{and}\qquad\inf_{t\geq 0}X_t>0
\end{align*}
\citep[see][Exercise~5.5.31]{KS91}. Hence, if $\mu<\sfrac{\sigma^2}{2}$, it follows that the optimal repurchase price is $z_*=0$, which implies that the optimal close-out time is $\tau_*=\check{\tau}_{z_*}=\infty$ and the value function is given by $V(x)=x$, for all $x\in(0,\infty)$. In other words, since the stock price eventually converges to zero in this case, and there is no penalty for waiting, it is optimal to wait forever. On the other hand, if $\mu=\sfrac{\sigma^2}{2}$ then $\check{\tau}_\varepsilon<\infty$, for all $\varepsilon>0$. In that case, the value function is given by $V(x)=x$, for all $x\in(0,\infty)$, since
\begin{equation*}
V(x)\geq\sup_{\varepsilon>0}\E_x(x-X_{\check{\tau}_\varepsilon})
=\sup_{\varepsilon>0}(x-\varepsilon)=x,
\end{equation*}
while $V(x)\leq x$ follows immediately from \eqref{eqSec5:ShortSellProb}. There is, however, no optimal close-out policy in this case. That is to say, since the stock price gets arbitrarily close to zero, but does not converge to zero, every close-out policy can be improved upon. Finally, we consider the situation when $\mu>\sfrac{\sigma^2}{2}$, in which case $X$ is a submartingale with terminal value $X_\infty\coloneqq X_{\infty-}=\infty$. It follows from the optional sampling theorem that $\E_x(X_\tau)\geq x$, for all $x\in(0,\infty)$ and each $\tau\in\StopTimes$. Hence, given an initial stock price $x\in(0,\infty)$, the optimal repurchase price is $z_*=x$ and the optimal close-out time is $\tau_*=\check{\tau}_{z_*}=0$, while the value function is given by $V(x)=0$, for all $x\in(0,\infty)$. In this case, the drift rate of the stock price is high enough for immediate close-out to be optimal.
\section{Comparative Static Analysis}
\label{Sec6}
\subsection{Methdology and parameter values}
This section analyses the dependence of the optimal repurchase price and value function for the constrained short selling problem \eqref{eqSec2:ShortSellProb} on the model parameters. To assess the impact of margin risk and recall risk, we compare the solution to the constrained problem with the solution to the unconstrained problem \eqref{eqSec5:ShortSellProb}. The solid red curves in the figures below illustrate the dependence of the optimal repurchase price and value function for the constrained short selling problem on one particular parameter, with the remaining parameters assigned fixed default values. The dashed blue curves plot the dependence of the optimal repurchase price and value function for the unconstrained problem on the same parameter.

The default parameter values are $\mu=\pm 0.02$, $\sigma=0.3$, $r=0.05$, $\lambda=0.01$, $c=\$50.00$ and $x=\$1.00$, and time is measured in years.\footnote{Note the values of $c$ and $\lambda$ are not relevant in the unconstrained case.} As we shall see, the solution to the constrained short selling problem is very sensitive to the sign of the drift rate of the stock price, which is why we consider the modest negative and positive drift scenarios $\mu=-0.02$ and $\mu=0.02$. The default volatility $\sigma=0.3$ is a reasonable proxy for observed equity implied volatilities, while the default discount rate $r=0.05$ corresponds roughly to the cost of borrowing in a developed market. The default recall intensity $\lambda=0.01$ implies that 1\% of stock loans are recalled per year, on average, which is much lower than the observed frequency.\footnote{In the broker data studied by \citet{D'A02}, around 2\% of the stocks on loan were recalled per month.}  Finally, the default values $c=\$50.00$ and $x=\$1.00$ for the collateral budget and the initial stock price mean that after selling the stock for $\$1.00$, the trader eventually runs out of collateral when the stock price reaches $\$51.00$. The default recall intensity and collateral budget were chosen to be conservative, to avoid overstating the impact of margin risk and recall risk on the short seller's problem.

Before we analyse the figures below in detail, we observe that the optimal repurchase price for the constrained short selling problem exceeds the optimal repurchase price for the unconstrained problem, in each case. In other words, the trader always chooses to close out earlier, when confronted with the possibility of involuntary close-out due to collateral exhaustion or early recall, than he would otherwise. This is because a more conservative strategy for the constrained problem reduces the likelihood of forced close-out, which usually results in a loss. Similarly, we observe that the value function for the unconstrained short selling problem dominates the value function for the constrained problem, in each case. The vertical distance between the two curves represents the loss in value due to the short selling constraints.

\subsection{The impact of a change in the drift rate of the stock price}
Figure~\ref{figSec6:mu} plots the optimal repurchase price and value function for the constrained and unconstrained short selling problems, as functions of the drift rate of the stock price. In Figure~\ref{figSec6:mu}\subref{figSec6:z*(mu)} we observe that the constrained and unconstrained optimal repurchase prices are very similar when the drift rate is negative, since forced close-out due to collateral exhaustion is unlikely if the stock price is expected to decline over time. However, the optimal close-out policies for the constrained and unconstrained short selling problems diverge rapidly as the drift rate of the stock price increases beyond zero. In particular, immediate repurchase is optimal ($z_*=\$1.00=x$) when $\mu>0.03$, in the case of the constrained problem (which means the trader will not sell the stock short in the first place), while immediate close-out is never optimal for the unconstrained problem. This is because there is a high likelihood that the constrained trader will run out of collateral and be forced to close-out in a loss-making position, if the drift rate of the stock price is that high, while the unconstrained trader has the luxury of waiting for the stock price to fall. To illustrate how dramatic the difference between the constrained and unconstrained close-out strategies can be when the drift rate of the stock price is positive, we observe that the optimal repurchase price for the constrained problem is $z_*=\$1.00$ when $\mu=0.05$, while the optimal repurchase price for the unconstrained problem is approximately $z_*=\$0.53$.

\begin{figure}
\centering
\mbox{
\subfigure[]{\includegraphics[scale=0.6]{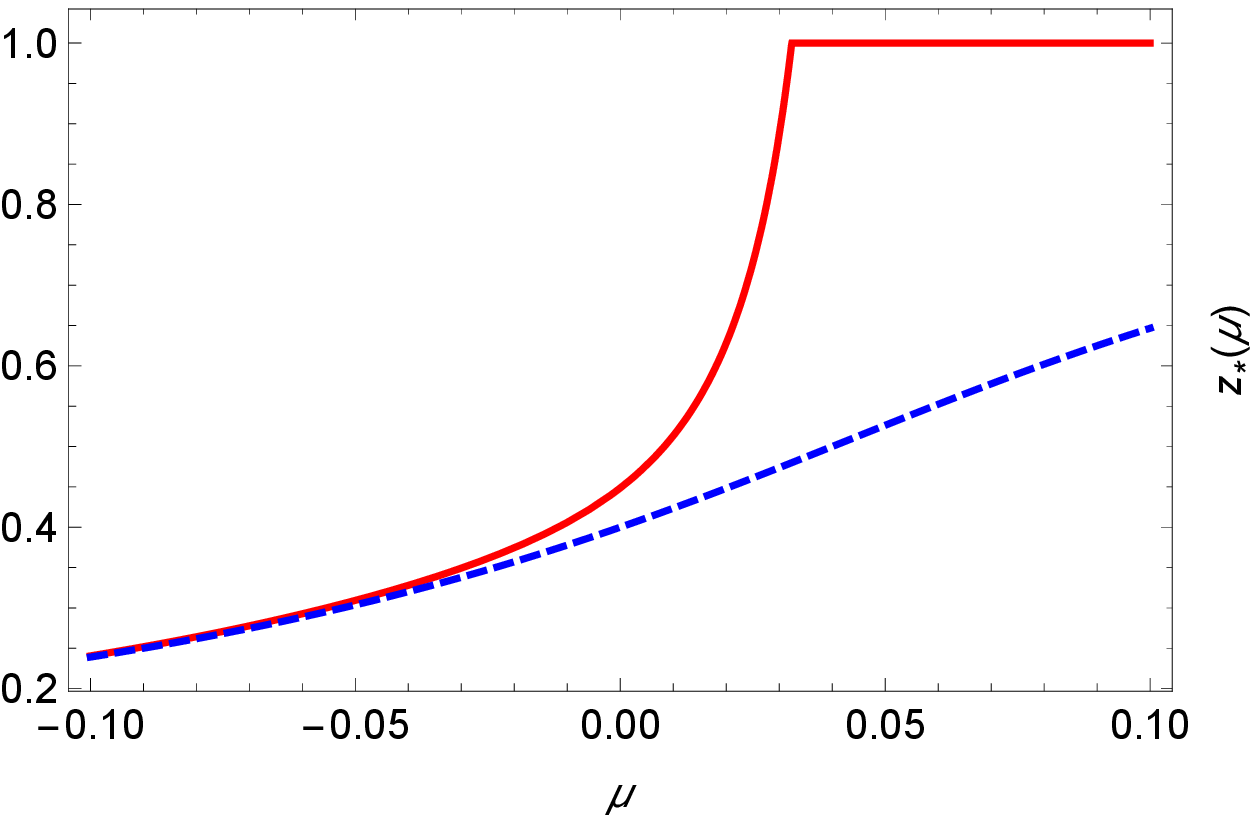}
\label{figSec6:z*(mu)}}
\subfigure[]{\includegraphics[scale=0.6]{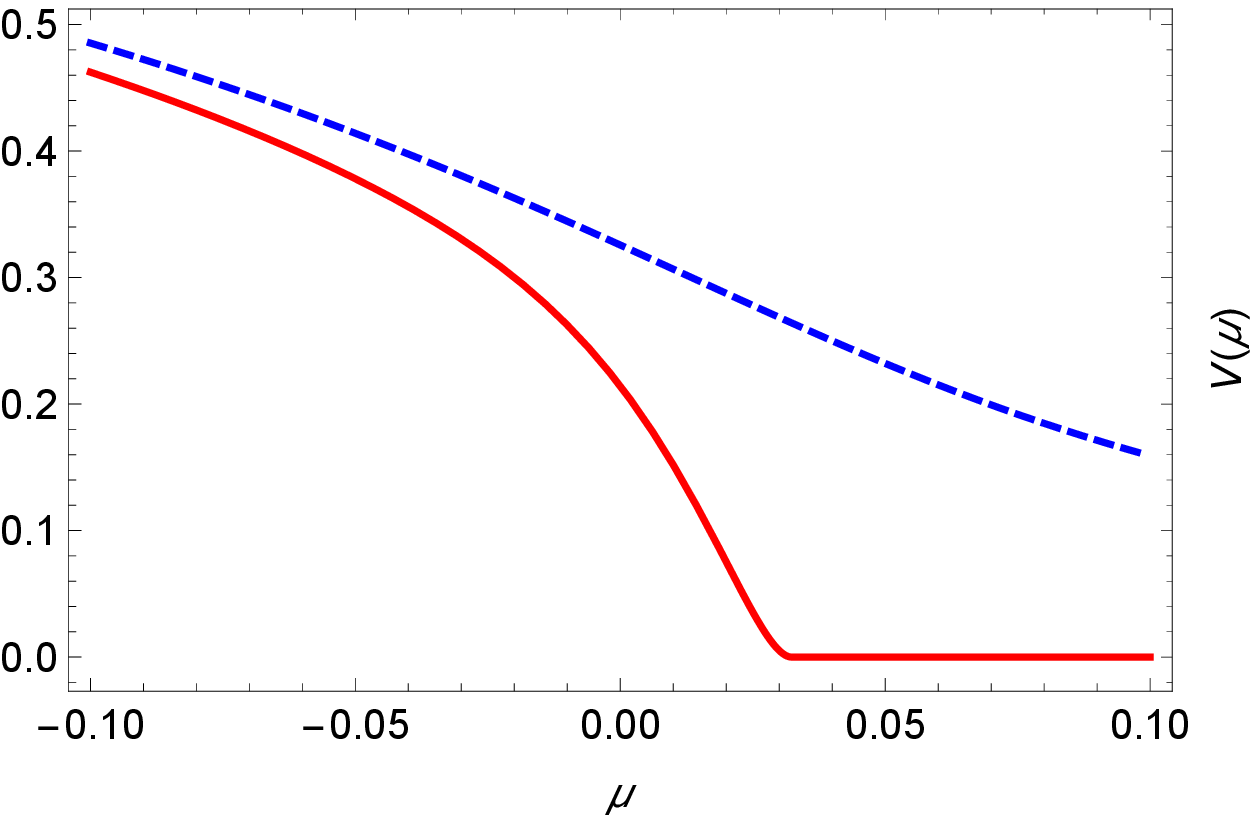}
\label{figSec6:V(mu)}}
}
\caption{The dependence of the optimal close-out price and the value function on the drift rate of the stock price, for the constrained (solid red curve) and unconstrained (dashed blue curve) short selling problems.}
\label{figSec6:mu}
\end{figure}

As expected, Figure~\ref{figSec6:mu}\subref{figSec6:V(mu)} shows that the loss in value due to margin risk and recall risk is relatively small when the drift rate of the stock price is negative. Since the stock price is expected decline over time in that case, the collateral constraint is unlikely to bind. Hence, the only real effect is due to the likelihood of recall (which does not depend on the drift rate). By contrast, the likelihood of forced close-out due to collateral exhaustion has a substantial impact on the constrained value function as the drift rate of the stock price increases beyond zero, resulting in a large loss of value relative to the unconstrained problem. In particular, the value of the constrained short sale is zero when $\mu>0.03$, since the trader will not sell the stock short in the first place, while the value of the unconstrained short sale is always strictly positive. For example, the constrained short sale is valueless when $\mu=0.05$, while the value of the unconstrained short sale is approximately $\$0.23$.

Figure~\ref{figSec6:mu} reveals that the constrained short selling problem is much more sensitive to the drift rate of the stock price than the unconstrained problem. For example, the optimal repurchase price for the constrained problem ranges between $\$0.24$ and $\$1.00$ in Figure~\ref{figSec6:mu}\subref{figSec6:z*(mu)}, while the optimal repurchase price for the unconstrained problem ranges between $\$0.24$ and $\$0.65$. Similarly, the value of the constrained short sale in Figure~\ref{figSec6:mu}\subref{figSec6:V(mu)} ranges between $\$0.00$ and $\$0.46$, while the value of the unconstrained short sale ranges between $\$0.16$ and $\$0.49$.

The enhanced sensitivity of the solution to the constrained short selling problem to the drift rate of the stock price is a source of fragility. Since estimates of the drift rate are accompanied by large standard errors, in practice, the trader is likely to misestimate it substantially, causing him to pursue a materially suboptimal close-out strategy. In detail, suppose the trader has observed the stock price at $n\in\N$ regular intervals over the period $[0,T]$, for some $T>0$. The sequence of observed prices is thus $(S_{i\Delta t})_{0\leq i\leq n}$, where $\Delta t\coloneqq\sfrac{T}{n}$. Based on those observations, the maximum likelihood estimator $\hat{\mu}$ of the drift rate is determined by
\begin{align*}
\hat{\mu}-\frac{1}{2}\sigma^2=\frac{1}{n\Delta t}\sum_{i=1}^n\ln\frac{S_{i\Delta t}}{S_{(i-1)\Delta t}}
&=\frac{1}{T}\sum_{i=1}^n\biggl(\biggl(\mu-\frac{1}{2}\sigma^2\biggr)\Delta t+\sigma\bigl(B_{i\Delta t}-B_{(i-1)\Delta t}\bigr)\biggr)\\
&=\mu-\frac{1}{2}\sigma^2+\frac{\sigma}{T}B_T
\end{align*}
\cite[see][Section~9.3.2]{CLM97}, whence
\begin{equation*}
\hat{\mu}=\mu+\frac{\sigma}{T}B_T\sim\mathcal{N}\biggl(\mu,\frac{\sigma^2}{T}\biggr).
\end{equation*}
Hence, if the drift rate and volatility of the stock price are $\mu=0.04$ and $\sigma=0.3$, respectively, and if the trader has 100 years of price data, the probability that he will estimate a non-positive drift rate is
\begin{equation*}
\P(\hat{\mu}\leq 0)=\P\biggl(\mu+\frac{\sigma}{\sqrt{T}}Z\leq 0\biggr)
=\P\biggl(Z\leq-\frac{\mu}{\sigma}\sqrt{T}\biggr)=\P\biggl(Z\leq-\frac{4}{3}\biggr)\approx 0.0912,
\end{equation*}
where $Z\sim\mathcal{N}(0,1)$. With reference to Figure~\ref{figSec6:mu}\subref{figSec6:z*(mu)}, this implies that there is a $9\%$ chance the trader will sell the stock short and maintain the position until its price reaches some level below $\$0.45$ (the optimal repurchase price for the constrained short sale when $\mu=0$), instead of correctly recognising that the drift rate is too high to sell it short in the first place.\footnote{A possible refinement to our model would be to incorporate parameter uncertainty and learning into the trader's close-out decision, along the lines followed by \citet{EL11}. They considered an investor whose choice of when to liquidate a pre-existing stock holding is confounded by uncertainty about the drift rate of the stock price. In their model, the investor updates his prior belief about the drift rate by observing the stock price over time. They showed that the investor's optimal strategy is to liquidate as soon as the stock price falls below a certain time-dependent boundary.}

\subsection{The impact of a change in the volatility of the stock price}
The dependence of the optimal repurchase price and value function on the volatility of the stock price, for the constrained and unconstrained short selling problems, is illustrated by Figure~\ref{figSec6:sigma_muneg}, in the case when the stock price has negative drift. Figure~\ref{figSec6:sigma_muneg}\subref{figSec6:z*(sigma)_muneg} shows that the optimal repurchase price for both problems is inversely related to volatility. In essence, a higher volatility increases the trader's incentive to wait for the stock price to reach a lower level before closing out, since it increases the likelihood that a lower level will be reached quickly. This is analogous to the situation with American puts, where higher volatilities correspond to a lower early exercise boundaries.\footnote{To justify this analogy, note that the constrained short selling problem \eqref{eqSec2:ShortSellProb} is similar to the pricing problem for an up-and-out at-the-money perpetual American put, while the unconstrained short selling problem \eqref{eqSec5:ShortSellProb} is similar to the pricing problem for a vanilla at-the-money perpetual American put.} However, the optimal repurchase prices for the two problems diverge as the stock price volatility increases, since a higher volatility increases the likelihood that the collateral boundary will be breached before the short sale can be closed out profitably, in the constrained case.

\begin{figure}
\centering
\mbox{
\subfigure[]{\includegraphics[scale=0.6]{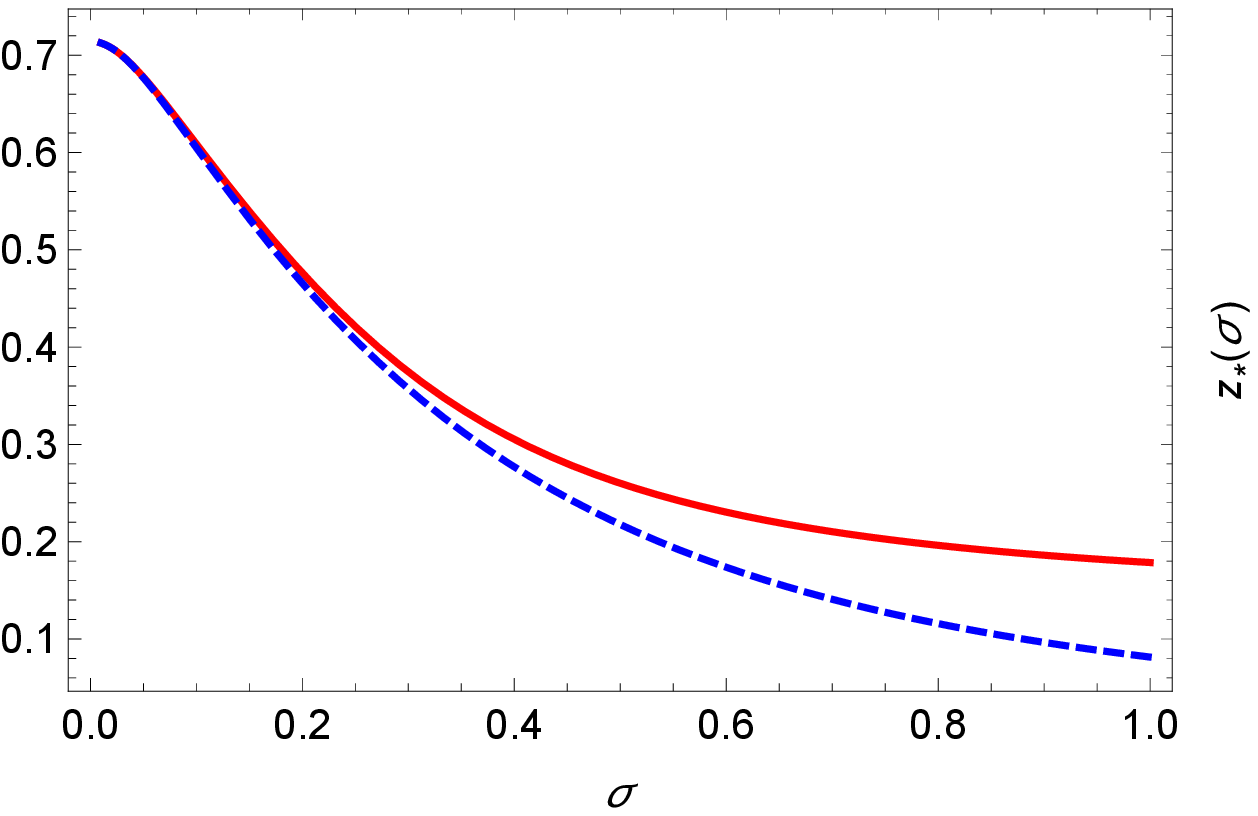}
\label{figSec6:z*(sigma)_muneg}}
\subfigure[]{\includegraphics[scale=0.6]{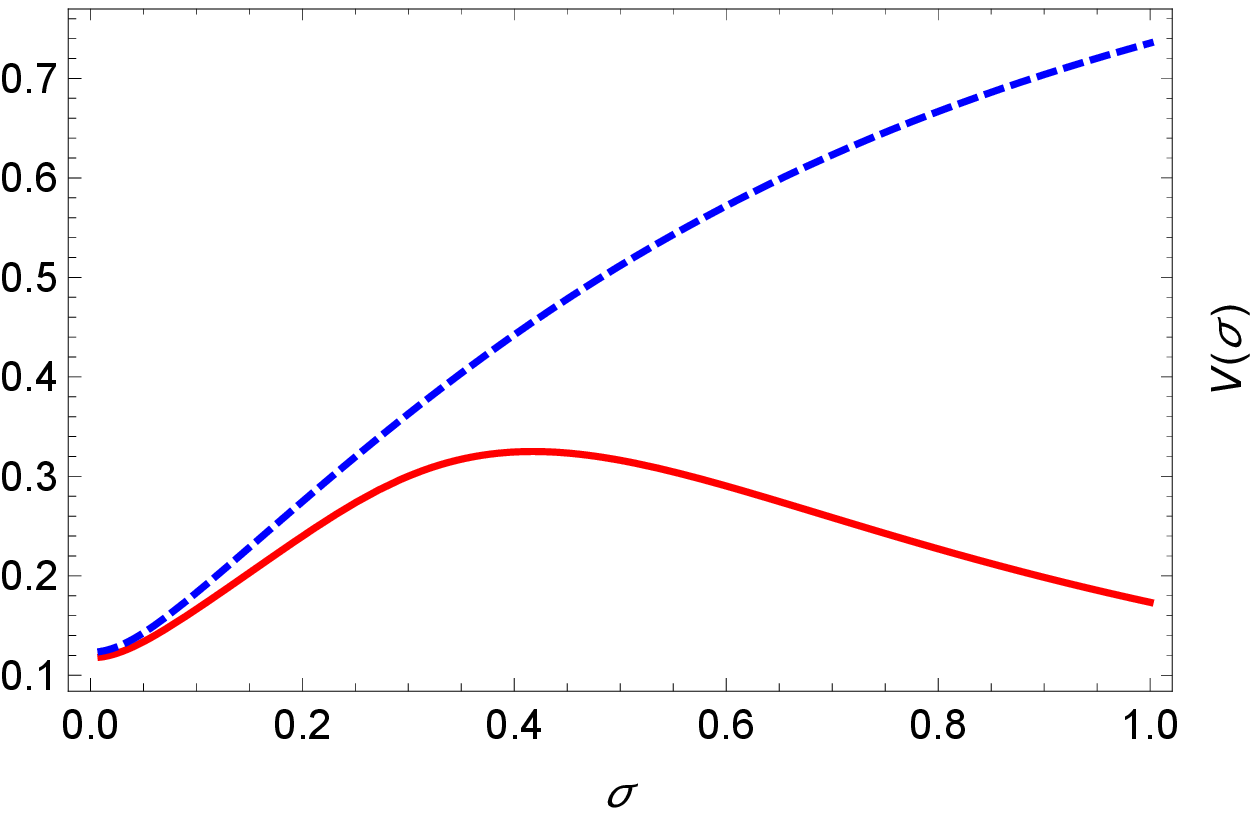}
\label{figSec6:V(sigma)_muneg}}
}
\caption{The dependence of the optimal close-out price and the value function on the volatility of the stock price, for the constrained (solid red curve) and unconstrained (dashed blue curve) short selling problems, when the drift rate of the stock price is negative.}
\label{figSec6:sigma_muneg}
\end{figure}

In Figure~\ref{figSec6:sigma_muneg}\subref{figSec6:V(sigma)_muneg} we see that the value function for the unconstrained short selling problem increases monotonically as the volatility of the stock price increases. This is analogous to the positive dependence of American option prices on the volatilities of their underlying assets \citep[see][]{Eks04}. By contrast, the value function for the constrained short selling problem initially increases as the volatility of the stock price increases, before subsequently decreasing. This reflects a tradeoff, where a higher volatility increases the likelihood of the optimal repurchase price being reached quickly, while simultaneously increasing the likelihood that the trader will run out of collateral. The former effect dominates when the volatility of the stock price is low, while the latter effect dominates when it is high. The prices of knock-out barrier options exhibit a similar non-monotonic dependence on the the volatilities of their underlying assets, for the same reason \citep[see][]{DK96}.\footnote{\citet{DK96} express this nicely by observing that ``the owner of a barrier option is long volatility at the strike [\,\ldots] and short volatility at an out barrier.''}

Figure~\ref{figSec6:sigma_mupos} illustrates the dependence of the optimal repurchase price and value function on the volatility of the stock price, for the constrained and unconstrained short selling problems, when the drift rate of the stock price is positive. Once again, the optimal repurchase price for both problems is decreasing with respect to volatility, as is evident from Figure~\ref{figSec6:sigma_mupos}\subref{figSec6:z*(sigma)_mupos}. However, volatility drives a larger wedge between the optimal close-out strategies for the two problems, than is the case when the drift rate is negative. For example, the difference between the constrained and unconstrained optimal repurchase prices is approximately $\$0.26-\$0.22=\$0.04$ when $\sigma=0.5$, in the negative drift scenario, while it is around $\$0.52-\$0.25=\$0.27$ in the positive drift scenario. This indicates that high volatility levels exacerbate the effect of positive drift, by making it more likely that the constrained trader will run out of collateral, relative to the situation when the drift rate is negative. As a result, the optimal close-out strategy for the constrained short selling problem is more conservative relative to the optimal close-out strategy for the unconstrained problem, when the drift rate is positive.

\begin{figure}
\centering
\mbox{
\subfigure[]{\includegraphics[scale=0.6]{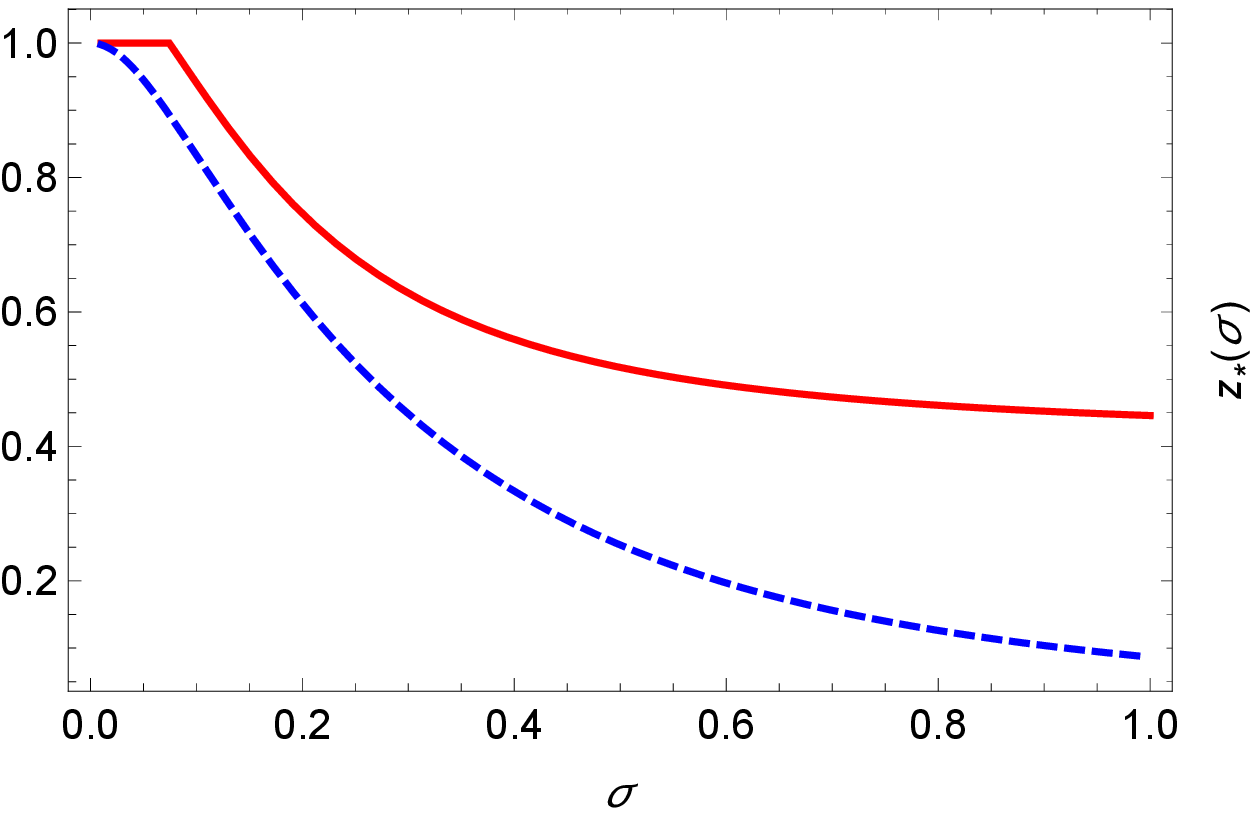}
\label{figSec6:z*(sigma)_mupos}}
\subfigure[]{\includegraphics[scale=0.6]{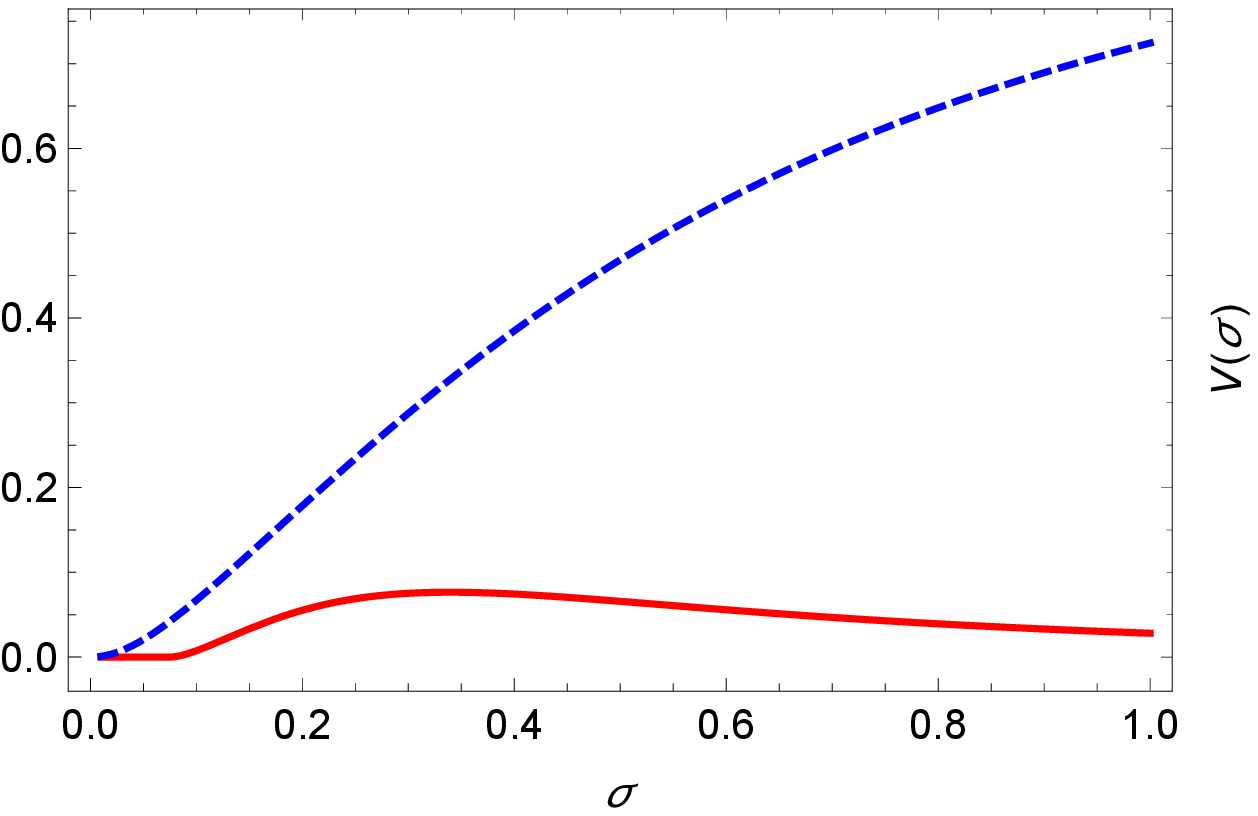}
\label{figSec6:V(sigma)_mupos}}
}
\caption{The dependence of the optimal close-out price and the value function on the volatility of the stock price, for the constrained (solid red curve) and unconstrained (dashed blue curve) short selling problems, when the drift rate of the stock price is positive.}
\label{figSec6:sigma_mupos}
\end{figure}

In the case of the constrained short selling problem, Figure~\ref{figSec6:sigma_mupos}\subref{figSec6:z*(sigma)_mupos} indicates that the position should be closed out immediately ($z_*=\$1.00=x$) when the drift rate of the stock price is positive and its volatility is small. If the stock is sold short under those circumstances, the positive drift rate will dominate the small volatility, causing a large proportion of stock price paths to reach the collateral barrier before the position can be closed out profitably. On the other hand, it is always optimal to wait before closing out the short position in the unconstrained case, even when the drift rate is positive and the volatility is small. However, the optimal repurchase price converges to the initial stock price as the volatility decreases, with immediate close-out being optimal in the zero-volatility limit. This is because the stock price will appreciate deterministically in that case, ruling out any possibility of profitable close-out.

In Figure~\ref{figSec6:sigma_mupos}\subref{figSec6:V(sigma)_mupos} we see that the value function for the constrained short selling problem continues to exhibit a non-monotonic dependence on volatility when the drift rate of the stock price is positive. However, we observe that the value of the constrained short sale is small relative to its value when the drift rate is negative, due to the fact that collateral exhaustion is much more likely to occur in the positive drift scenario. By contrast, the unconstrained short selling problem displays a modest decline in value across all volatilities, when the negative drift scenario is compared with the positive drift scenario. As expected from our analysis of the optimal close-out strategies for the constrained and unconstrained short selling problems, we see that volatility drives a larger wedge between their value functions when the drift rate is positive, than when it is negative. For example, the difference in values for the constrained and unconstrained problems is approximately $\$0.51-\$0.32=\$0.19$ when $\sigma=0.5$, in the negative drift case, while it is around $\$0.47-\$0.07=\$0.40$ in the positive drift case.

\subsection{The impact of a change in the discount rate}
In Figure~\ref{figSec6:r_muneg} we see how the optimal repurchase price and value function depend on the discount rate, for the constrained and unconstrained short selling problems, in the case when the drift rate of the stock price is negative. Figure~\ref{figSec6:r_muneg}\subref{figSec6:z*(r)_muneg} reveals that there is little difference between the optimal close-out strategies for the two problems, irrespective of the discount rate. In both cases, the optimal repurchase price increases monotonically with respect to the discount rate, since a higher discount rate imposes a larger penalty for waiting. Figure~\ref{figSec6:r_muneg}\subref{figSec6:V(r)_muneg} indicates that the unconstrained short sale is considerably more valuable than the constrained short sale when the discount rate is small, but the difference becomes smaller as the discount rate increases. For example, the difference between the value functions is approximately $\$1.00-\$0.60=\$0.40$ when $r=0$, but decreases to around $\$0.26-\$0.23=\$0.03$ when $r=0.1$. The reason is that losses due to collateral exhaustion are more costly, in present value terms, when the discount rate is low than when it is high. Given that it generally takes a long time for the stock price to reach the collateral barrier, a higher discount rate means the loss incurred when the collateral barrier is eventually reached is less significant, in present value terms.

\begin{figure}
\centering
\mbox{
\subfigure[]{\includegraphics[scale=0.6]{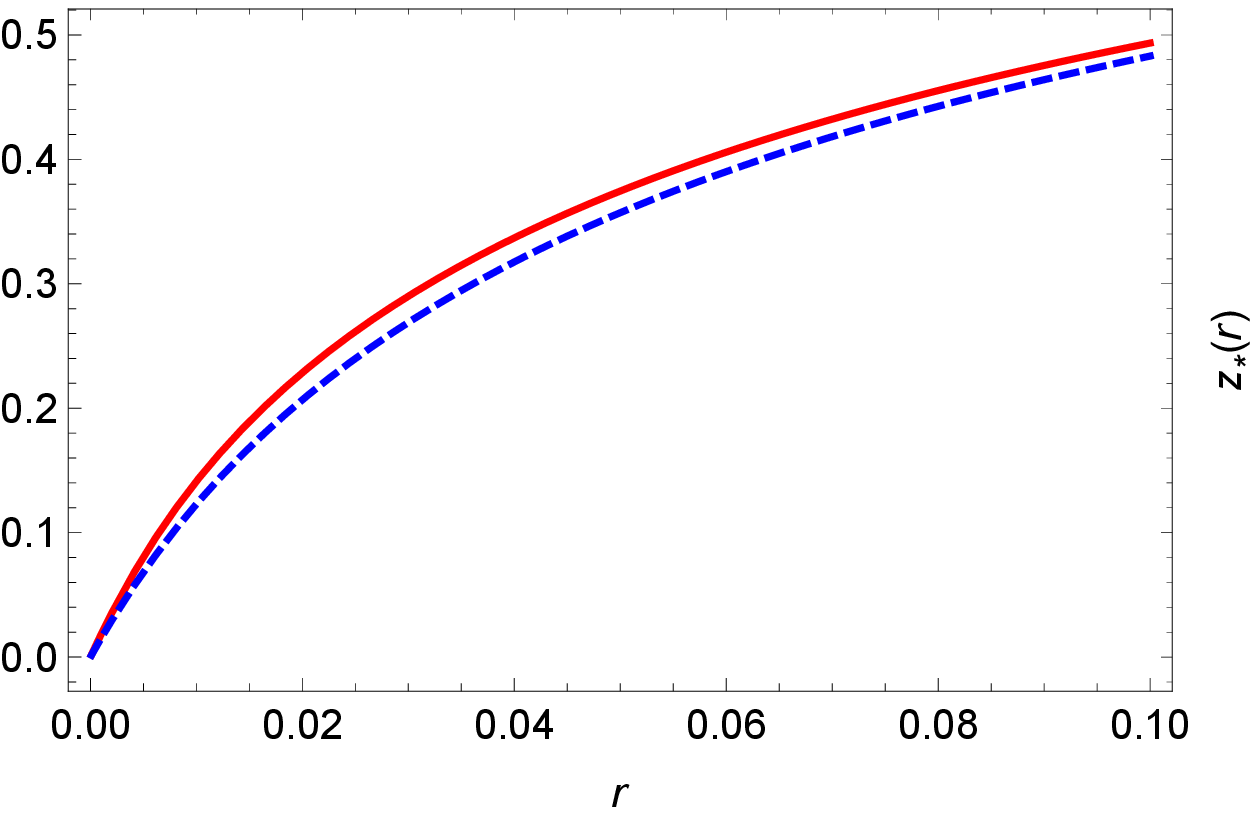}
\label{figSec6:z*(r)_muneg}}
\subfigure[]{\includegraphics[scale=0.6]{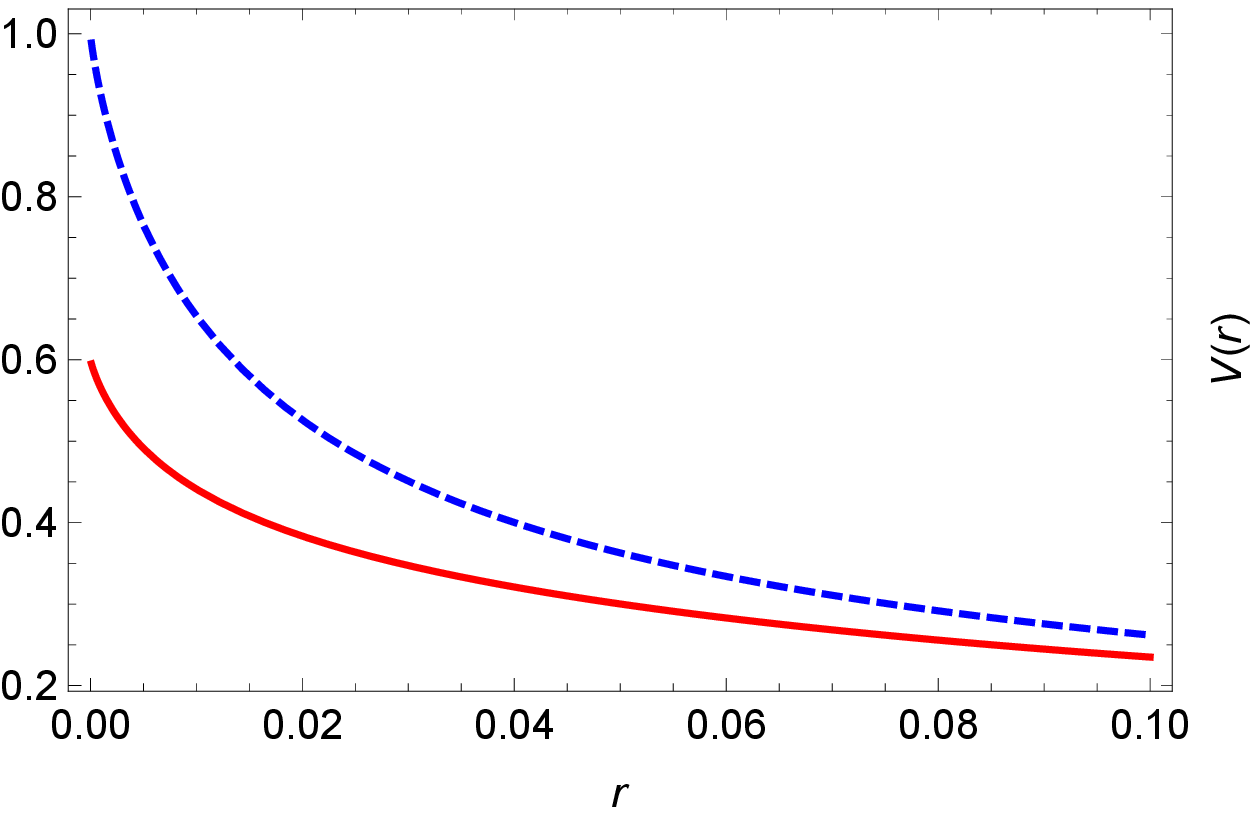}
\label{figSec6:V(r)_muneg}}
}
\caption{The dependence of the optimal close-out price and the value function on the discount rate, for the constrained (solid red curve) and unconstrained (dashed blue curve) short selling problems, when the drift rate of the stock price is negative.}
\label{figSec6:r_muneg}
\end{figure}

Figure~\ref{figSec6:r_mupos} illustrates the dependence of the optimal repurchase price and value function on the discount rate, for the constrained and unconstrained short selling problems, in the case when the drift rate of the stock price is positive. Figure~\ref{figSec6:r_mupos}\subref{figSec6:z*(r)_mupos} shows that the optimal repurchase price for the unconstrained problem increases monotonically with the discount rate, as it does when the stock price has negative drift. However, the behaviour of the optimal repurchase price for the constrained problem is more interesting when the drift rate is positive. When $r\leq 0.03$, immediate close-out is optimal ($z_*=\$1.00=x$), but waiting is optimal when $r>0.03$. Moreover, as the discount rate increases further, the optimal repurchase prices for the constrained and unconstrained short selling problems appear to converge. For example, the difference between the optimal repurchase prices for the two problems is $\$1.00-\$0.00=\$1.00$ when $r=0$, but it is only about $\$0.60-\$0.55=\$0.05$ when $r=0.1$. To explain these features, note that a positive drift rate ensures a high likelihood of running out of collateral before the constrained short sale can be closed out profitably, resulting in a substantial loss. When the discount rate is low, that loss is very significant, in present value terms. However, as the discount rate increases, its present value becomes less significant, since the stock price generally takes a long time to reach the collateral barrier. Ultimately, the loss due to collateral exhaustion becomes so small, in present value terms, that it plays no role in determining the optimal close-out strategy for the constrained short selling problem.

\begin{figure}
\centering
\mbox{
\subfigure[]{\includegraphics[scale=0.6]{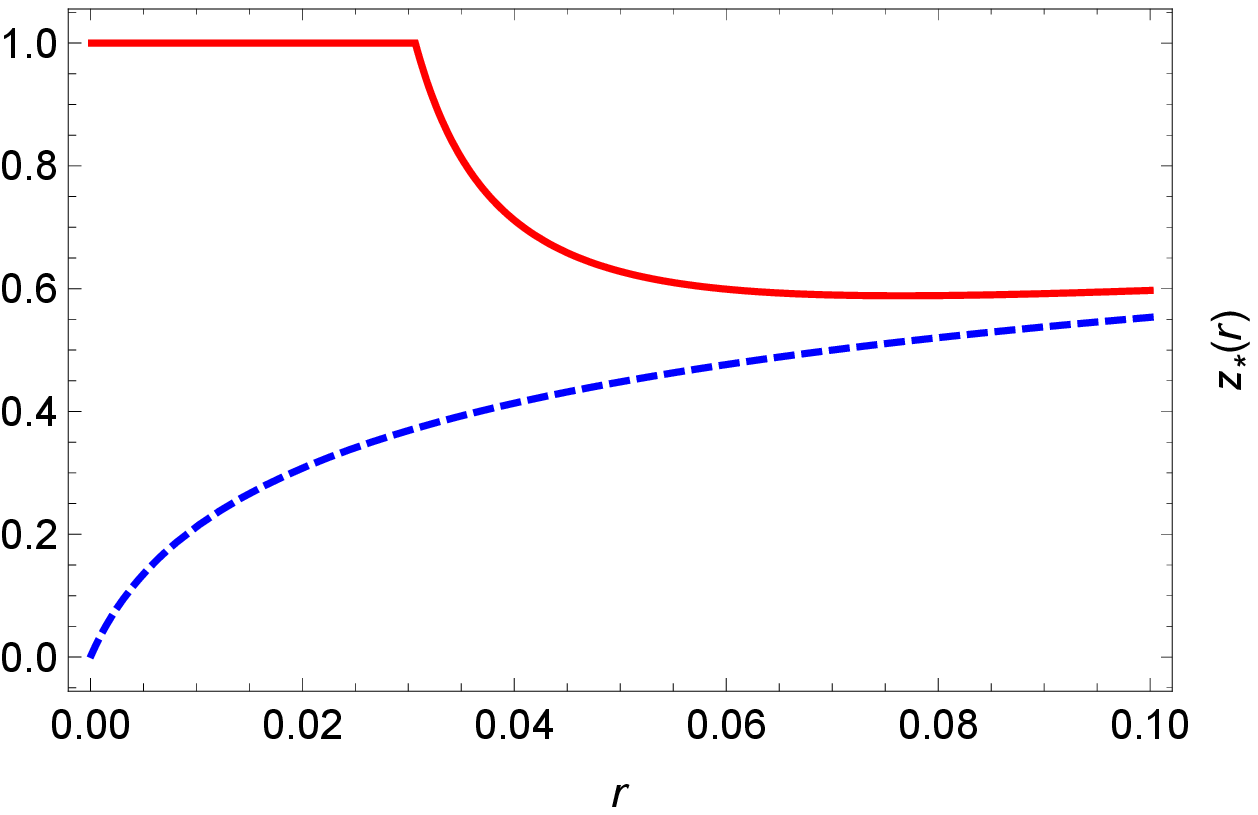}
\label{figSec6:z*(r)_mupos}}
\subfigure[]{\includegraphics[scale=0.6]{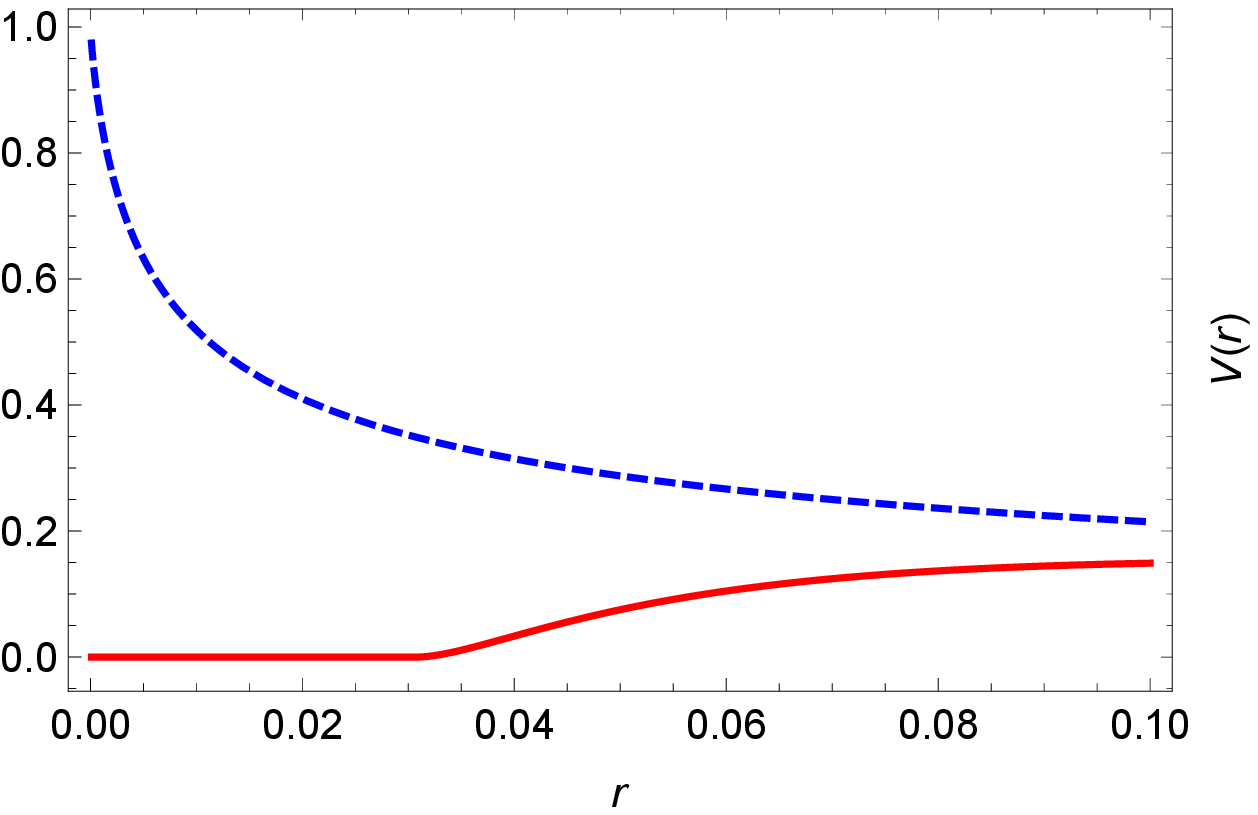}
\label{figSec6:V(r)_mupos}}
}
\caption{The dependence of the optimal close-out price and the value function on the discount rate, for the constrained (solid red curve) and unconstrained (dashed blue curve) short selling problems, when the drift rate of the stock price is positive.}
\label{figSec6:r_mupos}
\end{figure}

In the case of the unconstrained short selling problem, Figure~\ref{figSec6:r_mupos}\subref{figSec6:V(r)_mupos} reveals that the dependence of the value function on the discount rate does not change much when the drift rate of the stock price is positive. However, the value function for the constrained short selling problem behaves very differently with respect to the discount rate, in the positive drift case. In particular, since immediate close-out is optimal when the discount rate is small, the value functions for the constrained and unconstrained problems are significantly different then. But since the optimal close-out policies for the two problems apparently converge as the discount rate increases, their value functions appear to converge too. For example, the difference between the values for the constrained and unconstrained short selling problems is $\$1.00-\$0.00=\$1.00$ when $r=0$, but it is only around $\$0.21-\$0.15=\$0.06$ when $r=0.1$.

\subsection{The impact of a change in the recall intensity}
The optimal repurchase price and value function for the constrained and unconstrained short selling problems are displayed as functions of the recall intensity in Figure~\ref{figSec6:lambda_muneg}, in the case when the drift rate of the stock price is negative. The optimal repurchase price for the unconstrained problem is naturally unaffected by the recall intensity, as is evident in Figure~\ref{figSec6:lambda_muneg}\subref{figSec6:z*(lambda)_muneg}. However, the optimal repurchase price for the constrained problem increases as the recall intensity increases. Essentially, a higher recall intensity increases the likelihood of recall at an inopportune time. Consequently, the optimal repurchase price for the constrained short selling problem becomes more conservative as the recall intensity increases, which drives a wedge between the optimal close-out strategies for the two problems. For example, the difference between the optimal repurchase prices for the two problems is approximately $\$0.36-\$0.36=\$0.00$ when $\lambda=0$, but it increases to around $\$0.45-\$0.36=\$0.09$ when $\lambda=0.1$.

\begin{figure}
\centering
\mbox{
\subfigure[]{\includegraphics[scale=0.6]{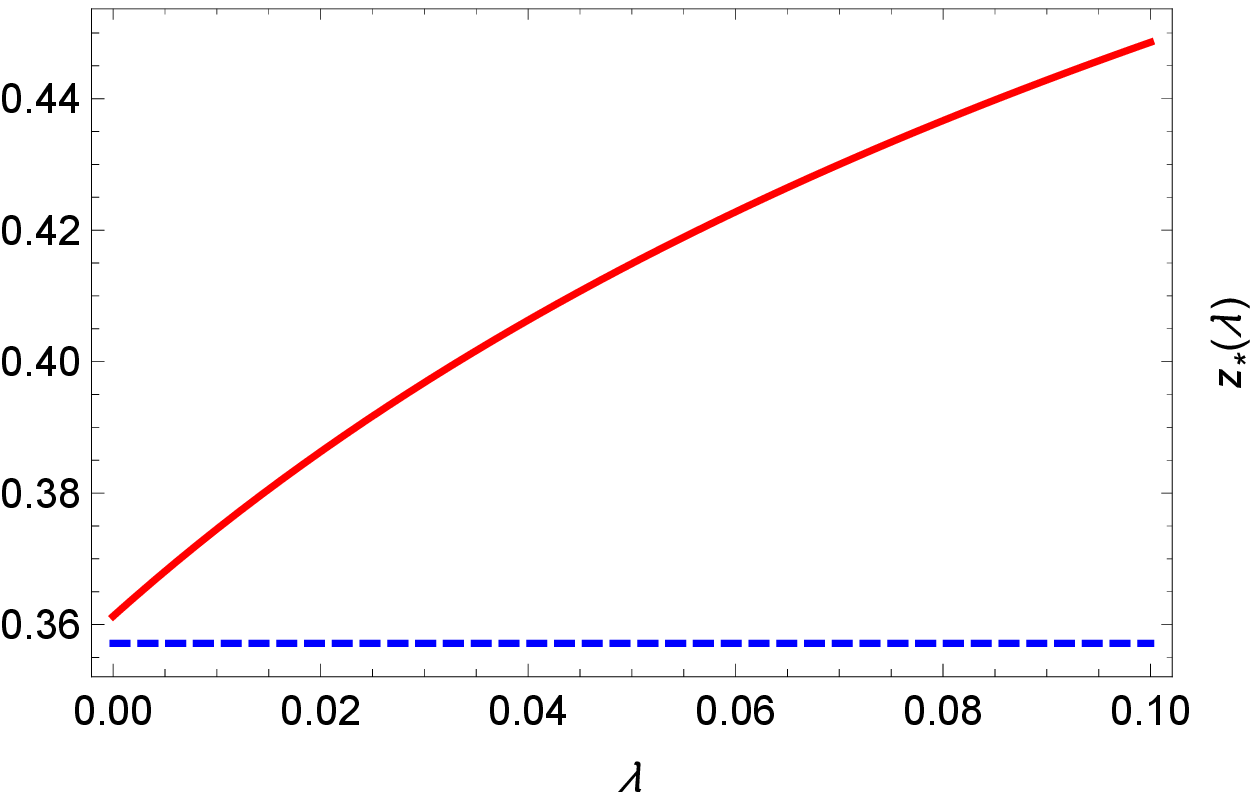}
\label{figSec6:z*(lambda)_muneg}}
\subfigure[]{\includegraphics[scale=0.6]{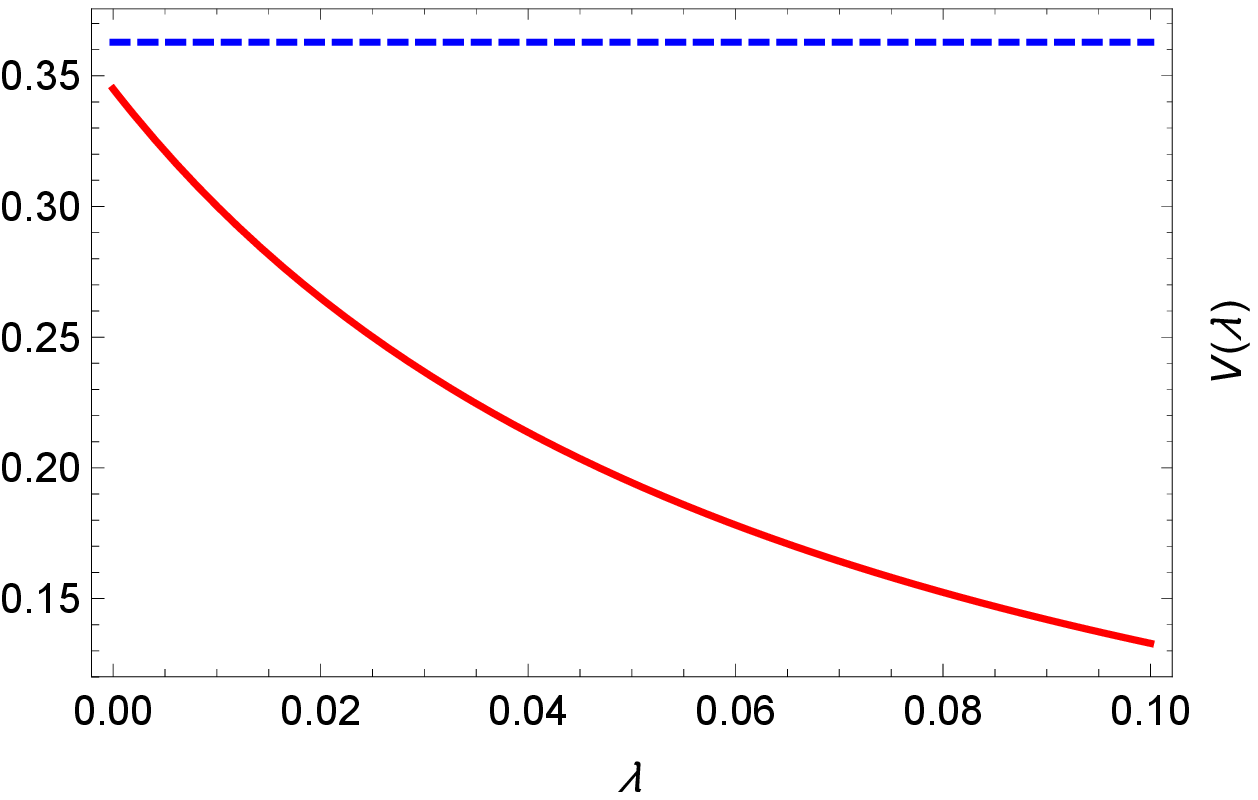}
\label{figSec6:V(lambda)_muneg}}
}
\caption{The dependence of the optimal close-out price and the value function on the recall intensity, for the constrained (solid red curve) and unconstrained (dashed blue curve) short selling problems, when the drift rate of the stock price is negative.}
\label{figSec6:lambda_muneg}
\end{figure}

In Figure~\ref{figSec6:lambda_muneg}\subref{figSec6:V(lambda)_muneg} we see that the value function for the unconstrained short selling problem is unaffected by the recall intensity, as expected, in the negative drift case. However, the value function for the constrained problem is monotonically decreasing with respect to the recall intensity. This reflects the fact that a higher recall intensity forces the trader to adopt a less optimal repurchase price, relative to the unconstrained problem, with a corresponding loss in value. To quantify the loss due to the possibility of recall, when the drift rate of the stock price is negative, we observe that the difference between the value functions for the two problems is approximately $\$0.36-\$0.35=\$0.01$ when $\lambda=0$, and increases to around $\$0.36-\$0.13=\$0.23$ when $\lambda=0.1$.

Figure~\ref{figSec6:lambda_mupos} plots the optimal repurchase price and value function for the constrained and unconstrained short selling problems as functions of the recall intensity, in the case when the drift rate of the stock price is positive. As expected, Figure~\ref{figSec6:lambda_mupos}\subref{figSec6:z*(lambda)_mupos} shows that the optimal repurchase price for the unconstrained problem is invariant with respect to changes in the recall intensity, while the optimal repurchase price for the constrained problem increases monotonically as the recall intensity increases. The difference between the optimal repurchase prices for the two problems is approximately $\$0.58-\$0.45=\$0.13$ when $\lambda=0$, and increases to around $\$0.90-\$0.45=\$0.45$ when $\lambda=0.1$. Since recall is impossible when $\lambda=0$, the difference in that case is entirely attributable to the impact of margin risk.

\begin{figure}
\centering
\mbox{
\subfigure[]{\includegraphics[scale=0.6]{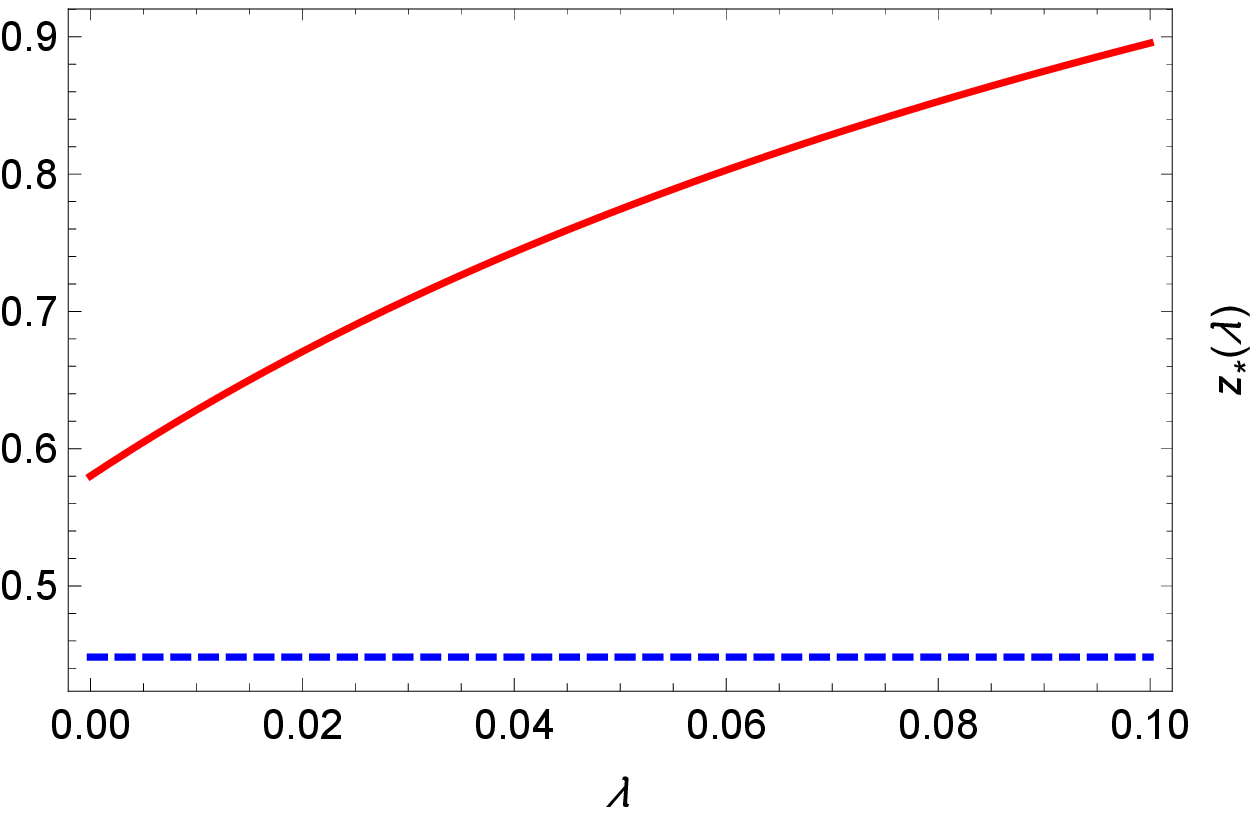}
\label{figSec6:z*(lambda)_mupos}}
\subfigure[]{\includegraphics[scale=0.6]{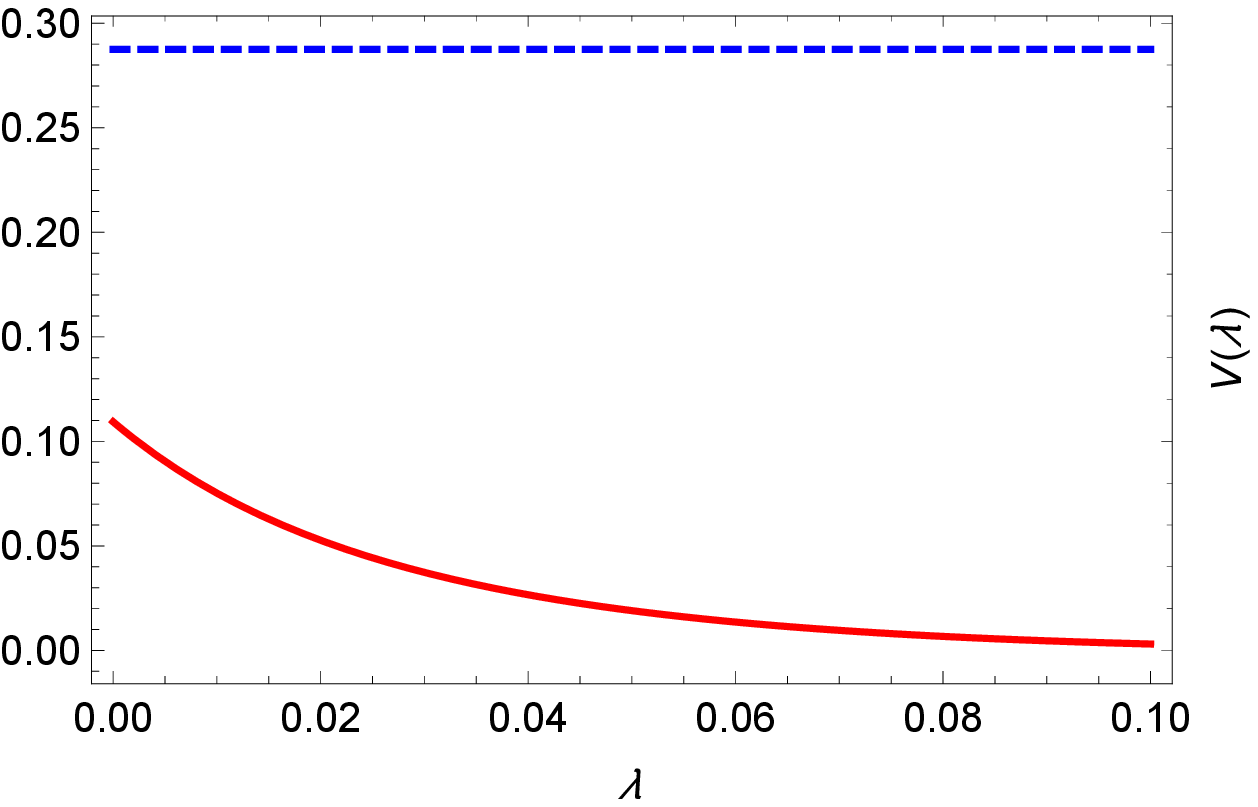}
\label{figSec6:V(lambda)_mupos}}
}
\caption{The dependence of the optimal close-out price and the value function on the recall intensity, for the constrained (solid red curve) and unconstrained (dashed blue curve) short selling problems, when the drift rate of the stock price is positive.}
\label{figSec6:lambda_mupos}
\end{figure}

Figure~\ref{figSec6:lambda_mupos}\subref{figSec6:V(lambda)_mupos} confirms that the value function for the unconstrained problem does not depend on the recall intensity, when the drift rate is positive, while the value function for the constrained problem is monotonically decreasing in that parameter. To quantify the loss of value due to the possibility of recall, we observe that the difference between the two value functions is approximately $\$0.29-\$0.11=\$0.18$ when $\lambda=0$, and increases to about $\$0.29-\$0.00=\$0.29$ when $\lambda=0.1$. Since recall is impossible in the former case, the loss of value is entirely due to the possibility of running out of collateral.

\subsection{The impact of a change in the collateral budget}
Figure~\ref{figSec6:c_muneg} illustrates the dependence of the optimal repurchase price and value function for the constrained and unconstrained short selling problems on collateral availability, if the drift rate of the stock price is negative. Figure~\ref{figSec6:c_muneg}\subref{figSec6:z*(c)_muneg} confirms that the collateral budget has no impact on the optimal close-out strategy for the unconstrained problem, but we observe a dramatic impact on the optimal close-out strategy for the constrained problem. In particular, it is optimal to close the constrained short sale out as soon as the stock price reaches $z_*=\$0.61$, if the margin calls arising from potential future stock price increases cannot be funded at all, but the optimal repurchase price declines rapidly as the collateral budget increases. The difference between the optimal repurchase prices for the two problems is around $\$0.61-\$0.36=\$0.25$ when $c=\$0.00$, but it decreases quickly and ultimately stabilises at $\$0.37-\$0.36=\$0.01$ when $c=\$100.00$. With that much collateral available, margin risk is unimportant, and the difference between the optimal repurchase prices for the two problems can be attributed purely to the possibility of early recall.

\begin{figure}
\centering
\mbox{
\subfigure[]{\includegraphics[scale=0.6]{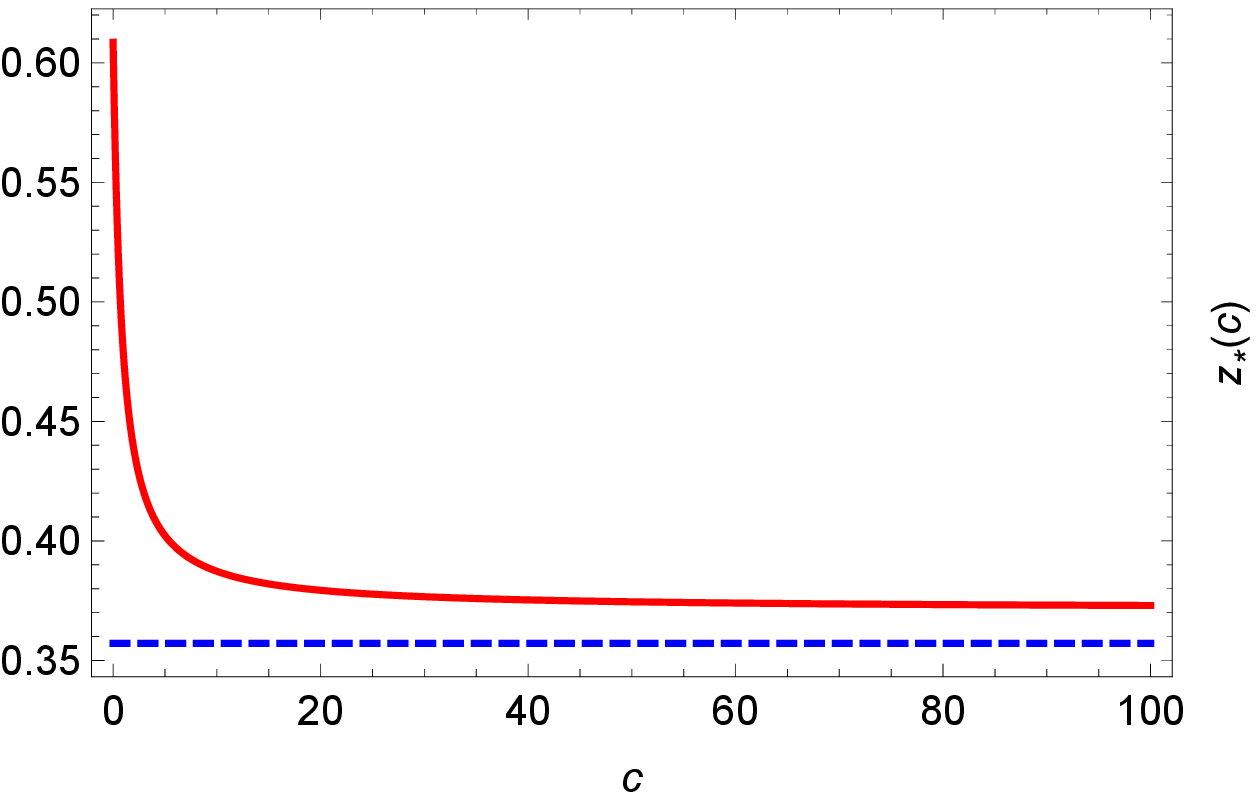}
\label{figSec6:z*(c)_muneg}}
\subfigure[]{\includegraphics[scale=0.6]{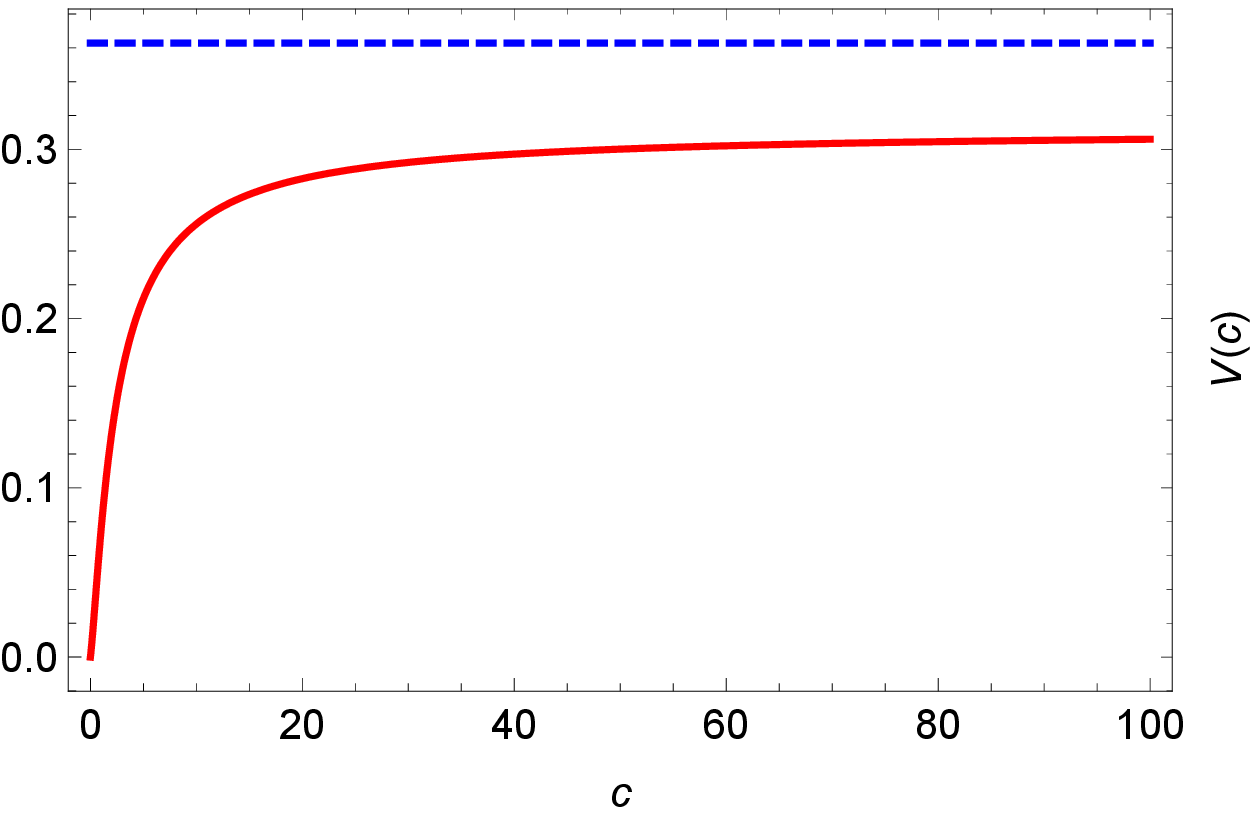}
\label{figSec6:V(c)_muneg}}
}
\caption{The dependence of the optimal close-out price and the value function on the collateral budget, for the constrained (solid red curve) and unconstrained (dashed blue curve) short selling problems, when the drift rate of the stock price is negative.}
\label{figSec6:c_muneg}
\end{figure}

As expected, Figure~\ref{figSec6:c_muneg}\subref{figSec6:V(c)_muneg} reveals that margin risk has no impact on the value of the unconstrained short sale, while its impact on the value of the constrained short sale is dramatic. In particular, the constrained short position is valueless when there is no collateral to fund margin calls, but its value increases rapidly as the collateral budget increases. We can quantify the impact of collateral on the loss of value due to the short selling constraints, by examining the difference between the constrained and unconstrained value functions for different levels of the collateral budget. For example, the difference between the value functions is approximately $\$0.36-\$0.00=\$0.36$ when $c=\$0.00$, but it decreases quickly and ultimately stabilises at $\$0.36-\$0.31=\$0.05$ when $c=\$100$. Since margin risk is essentially irrelevant in that case, the difference between the value functions for the two problems is exclusively due to the impact of recall risk.

In Figure~\ref{figSec6:c_mupos} we observe how the optimal repurchase price and value function for the constrained and unconstrained short selling problems depend on the collateral budget, in the case when the drift rate of the stock price is positive. Once again, the optimal repurchase price for the unconstrained problem is unaffected by collateral availability, as evident from Figure~\ref{figSec6:c_mupos}\subref{figSec6:z*(c)_mupos}, while the optimal repurchase price for the constrained problem is very sensitive to the collateral budget. If $c\leq\$4.00$, we see that immediate close-out ($z_*=\$1.00=x$) of the constrained short sale is optimal, since the collateral budget is insufficient to fund the margin calls that are likely to occur before the position can be closed-out profitably. However, as the amount of available collateral increases beyond $c=\$4.00$, the optimal repurchase price for the constrained short selling problem decreases rapidly, before levelling off. To assess the impact of collateral availability on the optimal close-out strategies for the constrained and unconstrained short selling problems, we observe that the difference between the optimal repurchase prices is initially $\$1.00-\$0.45=\$0.55$ when $c=\$0.00$. However, it quickly declines and achieves an ultimate value of around $\$0.60-\$0.45=\$0.15$ when $c=\$100.00$. With access to that much collateral, margin risk is virtually irrelevant, implying that the difference between the optimal close-out strategies is almost entirely attributable to the impact of recall risk.

\begin{figure}
\centering
\mbox{
\subfigure[]{\includegraphics[scale=0.6]{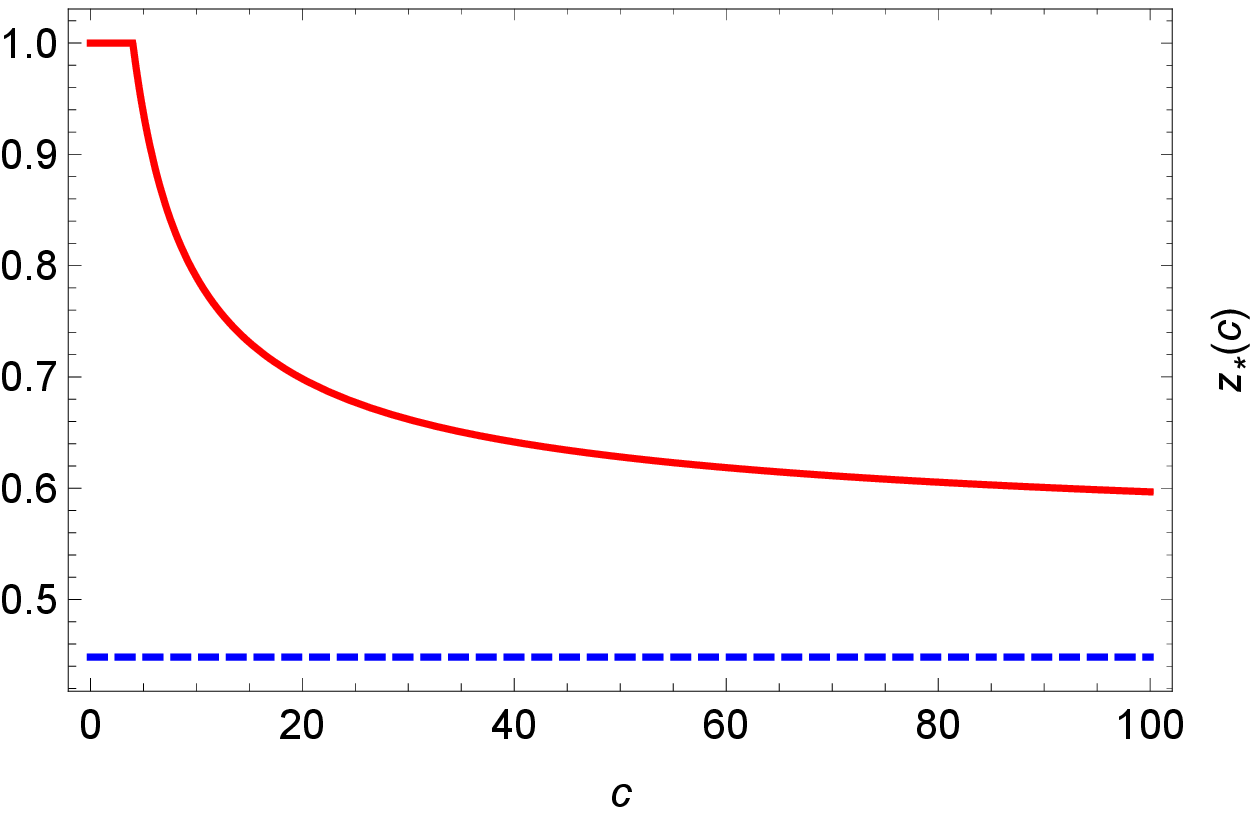}
\label{figSec6:z*(c)_mupos}}
\subfigure[]{\includegraphics[scale=0.6]{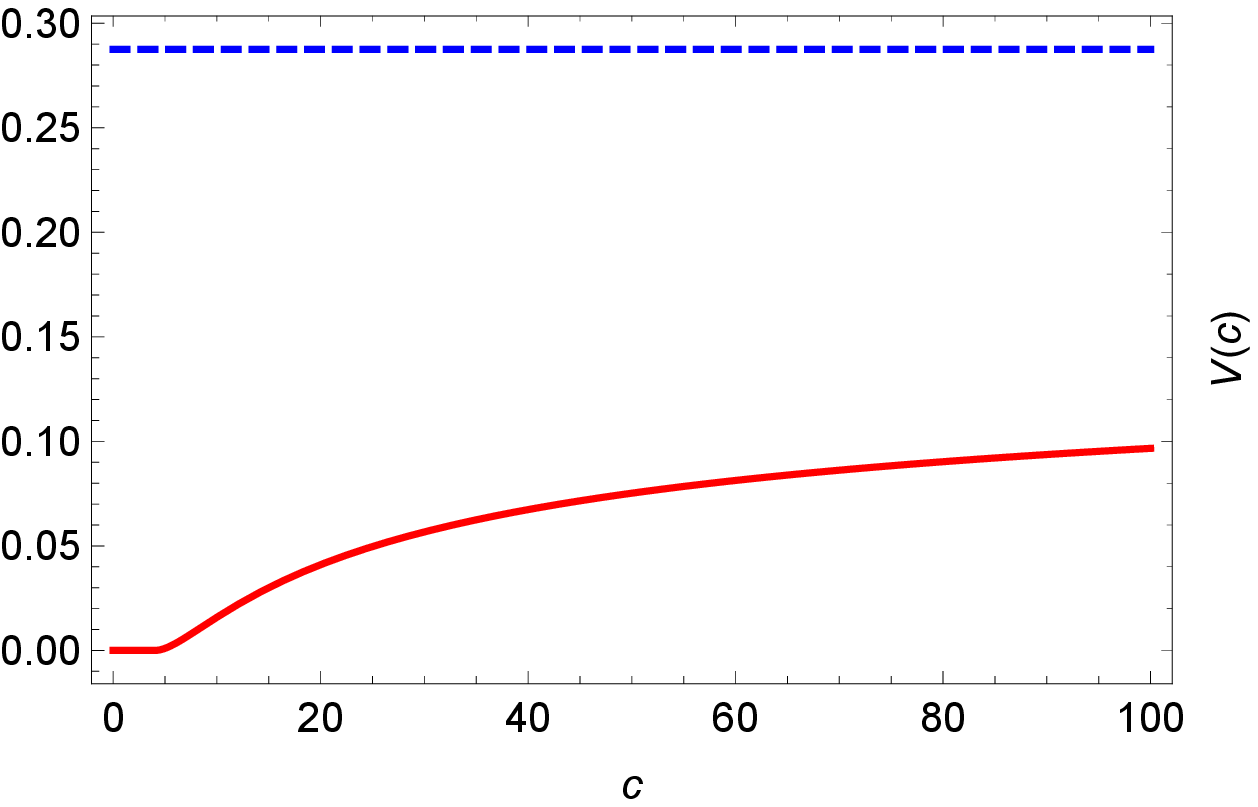}
\label{figSec6:V(c)_mupos}}
}
\caption{The dependence of the optimal close-out price and the value function on the collateral budget, for the constrained (solid red curve) and unconstrained (dashed blue curve) short selling problems, when the drift rate of the stock price is positive.}
\label{figSec6:c_mupos}
\end{figure}

Figure~\ref{figSec6:c_mupos}\subref{figSec6:V(c)_mupos} shows that the value function for the constrained short selling problem increases monotonically with respect to the collateral budget, when the drift rate of the stock price is positive, as it does in the negative drift scenario. We observe that the value of the constrained short sale is zero when $c\leq\$4.00$, since immediate close-out is optimal in that case, but it gradually increases as the amount of available collateral increases beyond $c=\$4.00$. By contrast, the value function for the unconstrained short selling problem is invariant with respect to changes in the collateral budget. To understand how the amount of available collateral affects the loss due to margin risk and recall risk in the positive drift case, we calculate the difference between the value functions for the constrained and unconstrained short selling problems, using different values for the collateral budget. For example, the difference between the value functions is $\$0.29-\$0.00=\$0.29$ when $c=\$0.00$, after which it slowly declines and reaches an ultimate value of around $\$0.29-\$0.10=\$0.19$ when $c=\$100.00$. In that case, the collateral budget is large enough for margin risk to be unimportant, which implies that loss in value between the constrained and unconstrained short sales is entirely due to the possibility of early recall. 

\bibliography{ProbFinBiblio}
\bibliographystyle{chicago}
\end{document}